\documentclass[sigconf]{acmart}

\AtBeginDocument{%
  \providecommand\BibTeX{{%
    \normalfont B\kern-0.5em{\scshape i\kern-0.25em b}\kern-0.8em\TeX}}}

\usepackage{graphicx}
\usepackage{textcomp}
\usepackage{xcolor}
\usepackage{booktabs} 
\usepackage[ruled, linesnumbered]{algorithm2e}
\usepackage{balance}
\usepackage{tikz}
\usetikzlibrary{trees}
\usepackage{lipsum,adjustbox}
\usetikzlibrary{positioning}
\usetikzlibrary{mindmap,trees}
\usepackage{verbatim}
\usepackage{multirow}
\usepackage[T1]{fontenc}
\usepackage{stmaryrd}
\usepackage{breqn}
\usepackage[shortlabels]{enumitem}

\usepackage{multicol}
\usepackage[skins]{tcolorbox}
\newtcolorbox{myframe}[2][]{%
	enhanced,colback=white,colframe=black,coltitle=black,
	sharp corners,boxrule=0.6pt,
	fonttitle=\itshape,
	attach boxed title to top left={yshift=-0.3\baselineskip-0.4pt,xshift=2mm},
	boxed title style={tile,size=minimal,left=0.5mm,right=0.5mm,
		colback=white,before upper=\strut},
	title=#2,#1
}

\newtheorem{myAttack}{Attack}

\newtheorem{proposition}{Proposition}
\usepackage{caption}
\usepackage{subcaption}

\usepackage{color, colortbl}
\definecolor{Gray}{gray}{0.9}
\definecolor{airforceblue}{rgb}{0.36, 0.54, 0.66}
\definecolor{aliceblue}{rgb}{0.94, 0.97, 1.0}
\definecolor{alizarin}{rgb}{0.82, 0.1, 0.26}
\definecolor{amber}{rgb}{1.0, 0.75, 0.0}
\definecolor{amber(sae/ece)}{rgb}{1.0, 0.49, 0.0}
\definecolor{bronze}{rgb}{0.8, 0.5, 0.2}
\definecolor{battleshipgrey}{rgb}{0.52, 0.52, 0.51}
\definecolor{bole}{rgb}{0.47, 0.27, 0.23}
\definecolor{bulgarianrose}{rgb}{0.28, 0.02, 0.03}
\definecolor{cadet}{rgb}{0.33, 0.41, 0.47}
\definecolor{ceil}{rgb}{0.57, 0.63, 0.81}
\definecolor{cerulean}{rgb}{0.0, 0.48, 0.65}
\definecolor{charcoal}{rgb}{0.21, 0.27, 0.31}
\definecolor{coolblack}{rgb}{0.0, 0.18, 0.39}
\definecolor{coolgrey}{rgb}{0.55, 0.57, 0.67}
\definecolor{darkcandyapplered}{rgb}{0.64, 0.0, 0.0}
\definecolor{darkbrown}{rgb}{0.4, 0.26, 0.13}
\definecolor{darkcerulean}{rgb}{0.03, 0.27, 0.49}
\definecolor{darkgray}{rgb}{0.66, 0.66, 0.66}
\definecolor{darkjunglegreen}{rgb}{0.1, 0.14, 0.13}
\definecolor{darktaupe}{rgb}{0.28, 0.24, 0.2}
\definecolor{davy\'sgrey}{rgb}{0.33, 0.33, 0.33}
\definecolor{frenchblue}{rgb}{0.0, 0.45, 0.73}
\definecolor{almond}{rgb}{0.94, 0.87, 0.8}
\definecolor{beaublue}{rgb}{0.74, 0.83, 0.9}
\definecolor{beige}{rgb}{0.96, 0.96, 0.86}
\definecolor{bisque}{rgb}{1.0, 0.89, 0.77}
\definecolor{black}{rgb}{0.0, 0.0, 0.0}
\definecolor{fluorescentorange}{rgb}{1.0, 0.75, 0.0}
\definecolor{ghostwhite}{rgb}{0.97, 0.97, 1.0}
\definecolor{antiquewhite}{rgb}{0.98, 0.92, 0.84}
\definecolor{LightCyan}{rgb}{0.88,1,1}

\newcommand{\rulesep}{\unskip\ \vrule\ }

\def\BibTeX{{\rm B\kern-.05em{\sc i\kern-.025em b}\kern-.08em
    T\kern-.1667em\lower.7ex\hbox{E}\kern-.125emX}}%% Rights management information.  This information is sent to you
\begin{document}

\title{Make Split, not Hijack: Preventing Feature-Space Hijacking Attacks in Split Learning}

\author{Tanveer Khan}
\orcid{0000-0001-7296-2178}
\affiliation{%
  \institution{Tampere University}
  \streetaddress{}
  \city{Tampere}
  \country{Finland}
  \postcode{}
}
\email{tanveer.khan@tuni.fi}
\author{Mindaugas Budzys}
\orcid{0000-0002-1913-7985}
\affiliation{%
  \institution{Tampere University}
  \streetaddress{}
  \city{Tampere}
  \country{Finland}
  \postcode{}
}
\email{mindaugas.budzys@tuni.fi}
\author{Antonis Michalas}
\orcid{0000-0002-0189-3520}
\affiliation{%
  \institution{Tampere University}
  \streetaddress{}
  \city{Tampere}
  \country{Finland}
  \postcode{}
}
\affiliation{
  \institution{RISE Research Institutes of Sweden}
  \city{Gothenburg}
  \country{Sweden}
}
\email{antonios.michalas@tuni.fi}

\begin{abstract}
The popularity of Machine Learning (ML) makes the privacy of
sensitive data more imperative than ever. Collaborative learning techniques like Split Learning (SL) aim to protect client data while enhancing ML processes. Though promising, SL has been proved to be vulnerable to a plethora of attacks, %like feature-space hijacking, 
thus raising concerns about its effectiveness on data privacy. In this work, we introduce a hybrid approach combining SL and Function Secret Sharing (FSS) to ensure client data privacy. The client adds a random mask to the activation map before sending it to the servers. The servers cannot access the original function but instead work with shares generated using FSS. Consequently, during both forward and backward propagation, the servers cannot reconstruct the client's raw data from the activation map. Furthermore, through visual invertibility, we demonstrate that the server is incapable of reconstructing the raw image data from the activation map when using FSS.
It enhances privacy by reducing privacy leakage compared to other SL-based approaches where the server can access client input information. Our approach also ensures security against feature space hijacking attack, protecting sensitive information from potential manipulation. Our protocols yield promising results, reducing communication overhead by over~$2\times$ and training time by over~$7\times$ compared to the same model with FSS, without any SL. Also, we show that our approach achieves~$>96\%$ accuracy and remains equivalent to the plaintext models.
\end{abstract}






\maketitle

\section{Introduction}
\label{sec:introduction}

%In recent years, Machine Learning (ML) has developed into a highly valuable field used in everyday situations. Through the use of ML, large language models, like ChatGPT~\cite{openai_chatgpt}, and medical image segmentation~\cite{oktay2018attention}, have become more accessible than ever before. These applications have enabled both experts and non-experts to achieve extraordinary results in data science and make accurate predictions on unseen data. However, ML requires a suitable amount of processing power, which is not available to everyone, to compute large and accurate models. As such, various companies have started offering Machine-Learning-as-a-Service (MLaaS) to allow a wider range of consumers to gain access to ML. In MLaaS, a company offers maintained servers to outsource computationally expensive ML operations. 

In recent years, Machine Learning (ML) has developed into a highly valuable field used in everyday situations. Various ML applications, like large language models~\cite{openai_chatgpt}, have become more accessible than ever before. However, ML requires a suitable amount of processing power to train. As such, various companies have started offering Machine-Learning-as-a-Service (MLaaS) to allow clients to use maintained servers to outsource resource-intensive ML tasks. Despite its benefits, MLaaS raises privacy concerns, as it requires sharing client data to a company during model training~\cite{shokri2017membership,hu2022membership}. %and exposes it to a variety of ML attacks~\cite{shokri2017membership,hu2022membership}. %a rising number of people have shown concern of using these services, as they require sharing a large amount of data to properly train the ML models. 
%This data may contain sensitive information, including medical records, financial data and other identifying information, which is important to keep private. New regulations on data privacy, such as the GDPR and the CCPA, have made consumers and companies more conscious about their privacy and take greater care to protect it in the digital world. However, researchers have discovered, that ML attacks can impact the security of ML and MLaaS, as well as potentially break compliance with these regulations~\cite{hu2022membership,veale2018algorithms} by leaking private information through different attacks~\cite{shokri2017membership, fredrikson2015model, song2017machine}. 
To address this privacy issue, researchers have begun developing Privacy-Preserving Machine Learning (PPML) techniques to protect client data~\cite{khan2021blind, khan2023learning, khan2024wildest}. % The objective of PPML is to ensure that ML will respect the privacy of the training data as well as the algorithms themselves~\cite{cabrero2021sok}. 
%PPML approaches have been implemented through the use of cryptographic, collaborative learning, data modification as well as hybrid techniques. Some of the most used techniques involve Homomorphic Encryption (HE)~\cite{gentry2009fully}, Multi-Party Computation (MPC)~\cite{yao1986generate}, Differential Privacy (DP)~\cite{dwork2006calibrating}, Federated Learning (FL)~\cite{truex2019hybrid} and Split Learning (SL)~\cite{abuadbba2020can}. 
These techniques encompass cryptographic approaches like Homomorphic Encryption (HE)~\cite{gentry2009fully} and Multi-Party Computation (MPC)~\cite{yao1986generate} as well as collaborative learning approaches such as Federated Learning (FL)~\cite{mcmahan2016federated, konevcny2016federated} and Split Learning (SL)~\cite{abuadbba2020can}.
These techniques have been used in a number of works and successfully reduce the privacy leakage of ML~\cite{abuadbba2020can,sav2021poseidon,hesamifard2016cryptodl, nguyen2023split, khan2023love, khan2023more, khan2023split}. 
%Looking into improving the security and efficiency of these approaches, researchers have lately started exploring solutions that combine several privacy-preserving techniques~\cite{khan2023split,truex2019hybrid,sav2021poseidon}. 

FL and SL are two methods which allow to train a model collectively from various distributed data sources without sharing raw data. FL, introduced by Google AI Blog, allow each user to run a copy of the entire model on its data. The server receives updated weights from each user and aggregates them. However, in FL, the users’ model weights are shared with the server, which can leak sensitive information~\cite{hitaj2017deep}. The SL approach proposed by Gupta and Raskar~\cite{gupta2018distributed} offers a number of significant advantages over FL. With SL, multiple parties can train an ML model together while keeping their respective parts private. Similar to FL, SL does not share raw data but it has the benefit of not disclosing the model’s architecture and weights and clients can train an ML model without sharing their data with a server, which has the remaining part of the model. Hence, in this work, our focus is on the SL technique, where low-cost operations are performed on the client-side while, computationally expensive operations are outsourced to a powerful server, expediting the models training process. However, SL without any additional privacy-preserving measures is prone to privacy leakage through Feature Space Hijacking Attack (FSHA)~\cite{pasquini2021unleashing} or data analysis techniques, such as Visual Invertibility (VI)~\cite{abuadbba2020can}. 
%However, SL without any additional privacy-preserving measures is prone to privacy leakage through data analysis techniques, such as Visual Invertability (VI)~\cite{abuadbba2020can}, and is vulnerable to Feature Space Hijacking Attack (FSHA)~\cite{pasquini2021unleashing}.
% As an example, SL can outsource computationally expensive operations to a powerful server and speed up the models training process. However, SL without any additional privacy-preserving measures is prone to privacy leakages~\cite{abuadbba2020can} and is vulnerable to different attacks~\cite{abuadbba2020can, pasquini2021unleashing} 

Through VI a malicious server is able to infer information of the underlying input data by analysing the visual similarity between the original image and the produced activation map after a number of convolutional layers. % Other works attempt to reduce the privacy leakage caused by VI by employing DP~\cite{abuadbba2020can} or HE~\cite{nguyen2023split} but note either reduced model accuracy or greatly increased execution times.
Looking into improving the security and efficiency of individual approaches, researchers have lately started exploring solutions that combine several privacy-preserving techniques~\cite{khan2023split,truex2019hybrid,sav2021poseidon}. 
%To address the privacy leakage, 
As such, researchers proposed to use differential privacy (DP)~\cite{abuadbba2020can} with SL, %to reduce the privacy leakage,
but noted, that it greatly reduces the utility of the trained model. %and does not fully address the leakage. 
Others have proposed encrypting the intermediate results using HE before sending them to the server~\cite{khan2023split}. This allows the server to compute the final layers of the ML model and get the correct outputs without leaking additional information to the server. Despite the increased security, the authors report, that the computational overhead even for a single linear layer using HE increases the training time drastically, when compared to the unencrypted models. Regarding the FSHA, existing work in literature does not address this attack in SL. FSHA allows a malicious server to hijack the model and steer it towards a specific target feature space. Once the model maps into the target feature space, the malicious server can recover the private training instances by reversing the known feature space. In short, FSHA allows a malicious server to take control of the model and extract private information from training data. 

%Following these results, we propose to use a different cryptographic MPC-based approach to solve the privacy leakage of SL, while greatly reducing the computational complexity of an HE-based approach. The cryptographic technique in question is the novel 
In this work, to secure SL against FSHA and visual analysis, we use Function Secret Sharing (FSS)~\cite{boyle2015function} -- an MPC technique which allows two parties to perform computations on a public input using a private, secret shared function. %In our work, w
We show that FSS can be used to eliminate the privacy leakage of SL while at the same time keeping data safe from FSHA. In addition, it greatly reduces the computational complexity, when compared to existing %both HE and non-split FSS
techniques~\cite{khan2023split, nguyen2023split, ryffel2020ariann}. Furthermore, we show that combining FSS with SL results in lower communications overhead than combining HE with SL. 

\medskip

\noindent \textit{\textbf{Contributions:}}
The main contributions of this work, can be summarized as follows:

\begin{enumerate}[\bfseries C1., leftmargin=0.7cm]
    \item We designed an efficient PPML protocol using SL and FSS.  Our contribution lies in being the first to incorporate FSS into SL, aiming at enhancing the privacy %aspect 
    of SL.
    \item We have conducted experiments on the MNIST dataset to demonstrate privacy leakage in SL after the first and second convolution layers. We have used FSS to mitigate this issue, offering a more efficient alternative to the computationally expensive HE method mentioned in prior work by Nguyen \textit{et al.}~\cite{nguyen2023split}.  Our approach enables running multiple layers on the server-side, unlike~\cite{nguyen2023split}, which considered only one layer on the server-side.
    \item We designed an FSS based vanilla SL (refer to as private vanilla SL model) protocol and compared it with AriaNN~\cite{ryffel2020ariann} (private local). Our findings demonstrated that our approach is more efficient in terms of communication cost and complexity while maintaining the same level of accuracy and privacy.
    % \item Both private vanilla SL and private local model could not adequately protect client data, as the final layer is executed on the server-side. To mitigate this, we developed a more secure and efficient private U-shaped SL  (refer to as private U-shaped SL) protocol where the final layer is executed on the client-side, addressing the privacy issues of both protocols.
    \item To date, VI has not been considered an attack in research; rather, it has been seen as a privacy leakage metric in SL. In our work, we use VI as an attack, specifically a VI Inference Attack (VIIA) in SL.
    \item We demonstrated the security of our approach against %various attacks, including %feature-space hijacking (
    FSHA. Also, our VI analysis demonstrates that the server cannot reconstruct raw image data from the activation maps. %, label inference, and EXACT attacks. 
\end{enumerate}

\subsection{Organization}
The rest of the paper is organized as follows, in \autoref{sec:preliminaries}, we %start by providing 
provide the necessary background information on Convolution Neural Network (CNN), SL, %additive secret sharing, beaver triples 
and FSS. In \autoref{sec:FSS_primitives}, we define the notation used throughout the paper, discuss FSS in a higher detail, and describe its core primitives. In \autoref{sec:related_Work}, we present important published works in the area of SL and FSS.
The architectures of the two local models: one without a split and another with a split %, both applied to plaintext data 
are presented in \autoref{sec:architecture}, %. Also, in this section, we briefly explain 
where we elaborate on the private local model and introduce the key actors in our protocols. %We provide a detailed explanation of our proposed private vanilla SL protocol in \autoref{sec: methodology}. % and private U-Shaped SL protocol in \autoref{sec: methodology}. 
A detailed description of our private vanilla SL protocol is described in \autoref{sec:methodology}, followed by the considered threat model \& security analysis in \autoref{sec:threat_model_and_sec}, 
%In \autoref{sec: threat model}, we examine the capabilities of an adversary %various attacks applicable to SL protocol 
%and %use them to 
%assess the security of our proposed protocols. 
%Then, in \autoref{sec: security analysis}, we show that our proposed protocol is secure against the FSHA, while %We also address privacy concerns identified in the SL protocol and visually illustrate how our approaches resolve them. 
 while \autoref{sec:perfanal} provides a detailed explanation of the experimental setup and provides extensive experimental results for our proposed approach. Finally, we conclude the paper in \autoref{sec:conclusion}.

\section{Preliminaries}
\label{sec:preliminaries}
In this section, we cover the underlying ML and cryptographic techniques we utilize in our design. 

\subsection{Convolutional neural network}
CNN is a special type of neural network which have shown great results in various ML classification and recognition tasks for %both with one dimensional time series data~\cite{li2017classification, abuadbba2020can}, and 
two dimensional (2D) images data~\cite{lecun1998CNN}. %In our work we make use of a~2D CNN in a supervised learning method, where both the input data and corresponding labels are needed to train a network. 
For our task, we employ a~2D CNN which consists of the following layers:    
\begin{itemize}[leftmargin=0.6cm]
    \item \underline{2D Convolution (Conv2D)}: %Commonly abbreviated as, Conv2D, -- is an important layer of a CNN, and its main function is to extract features. 
    It is used in %Neural Networks 
    CNN to perform convolution operation on the input data, that is, this layer applies a set of kernels, to a sliding window of input data, which extracts local features from the data (see \autoref{fig:1dcnn})~\cite{lecun1998CNN}. %\autoref{fig:1dcnn} visualizes both the Conv1D and Conv2D operations and illustrates the difference between a Conv1D and a Conv2D layer. The output of Conv2D layer is a set of activation maps ($ATm$), which represents the learned features of the input data. Each $ATm$ corresponds to a different filter, and each element in the $ATm$ indicates the activation of the filter at that position in the input sequence.  
\item \underline{ReLU}:  
     It is a non-linear activation function that is commonly used in CNN to compute a simple function: %$f(x) = \max(0, x)$.
     
     \begin{center}
         $f(x) = \max(0, x)$.
     \end{center} %The above function returns~0 for negative input and returns positive input as is.

%      \begin{equation*}
%     f(x) =
%         \begin{cases}
%             \ 0,          & \text{if } x<0,\\
%             \ x,         & \text{if } x \geq 0,
%         \end{cases}
% \end{equation*}

     \item \underline{Max-pooling}: It is %a common operation in CNNs that is 
     used in CNN
     to create a downsampled activation map ($ATm$) with a lower resolution than the input. This operation helps to reduce the computational cost of the network and prevent overfitting by compressing information in the $ATm$. %The compressed $ATm$ only retain the crucial information of the input and add a small amount of input translation invariance to the network~\cite{maxpooling2010}. This means that small translations of the input has an insignificant effect on the pooled $ATm$.%
    
     \item \underline{Fully Connected (FC)}: 
     In FC layers, each neuron in the layer is connected to every other neuron in the previous layer. This means that the output of each neuron in the layer is computed as a weighted sum of the outputs of all the neurons in the previous layer. % This layer can be used in intermediate layers to extract high-level features from the input data or in the final layers of a CNN to perform classification tasks.
\end{itemize}

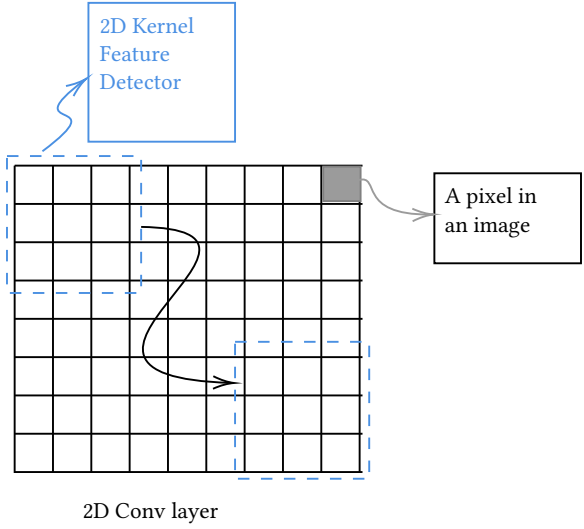
\begin{figure}
    \resizebox{0.45\textwidth}{!}{%
        \tikzset{every picture/.style={line width=0.75pt}}       
        
        \tikzset{every picture/.style={line width=0.75pt}} %set default line width to 0.75pt        
        
        \begin{tikzpicture}[x=0.75pt,y=0.75pt,yscale=-1,xscale=1]
            
            %Shape: Rectangle [id:dp8573947153710666] 
            % \draw  [color={rgb, 255:red, 0; green, 0; blue, 0 }  ,draw opacity=1 ] (45.5,149) -- (148.5,149) -- (148.5,229) -- (45.5,229) -- cycle ;
            % %Shape: Rectangle [id:dp9299829119675594] 
            % \draw  [color={rgb, 255:red, 74; green, 144; blue, 226 }  ,draw opacity=1 ] (80.5,22) -- (160,22) -- (160,97.5) -- (80.5,97.5) -- cycle ;
            % %Shape: Rectangle [id:dp9836898310539085] 
            \draw  [color={rgb, 255:red, 74; green, 144; blue, 226 }  ,draw opacity=1 ] (282.5,10) -- (359,10) -- (359,82.5) -- (282.5,82.5) -- cycle ;
            %Shape: Rectangle [id:dp9621068931040797] 
            \draw  [color={rgb, 255:red, 0; green, 0; blue, 0 }  ,draw opacity=1 ] (463,99) -- (539.5,99) -- (539.5,146) -- (463,146) -- cycle ;
            %Shape: Grid [id:dp5963779104456358] 
            \draw  [draw opacity=0] (244,95) -- (425.5,95) -- (425.5,255.5) -- (244,255.5) -- cycle ; \draw  [color={rgb, 255:red, 0; green, 0; blue, 0 }  ,draw opacity=1 ] (244,95) -- (244,255.5)(264,95) -- (264,255.5)(284,95) -- (284,255.5)(304,95) -- (304,255.5)(324,95) -- (324,255.5)(344,95) -- (344,255.5)(364,95) -- (364,255.5)(384,95) -- (384,255.5)(404,95) -- (404,255.5)(424,95) -- (424,255.5) ; \draw  [color={rgb, 255:red, 0; green, 0; blue, 0 }  ,draw opacity=1 ] (244,95) -- (425.5,95)(244,115) -- (425.5,115)(244,135) -- (425.5,135)(244,155) -- (425.5,155)(244,175) -- (425.5,175)(244,195) -- (425.5,195)(244,215) -- (425.5,215)(244,235) -- (425.5,235)(244,255) -- (425.5,255) ; \draw  [color={rgb, 255:red, 0; green, 0; blue, 0 }  ,draw opacity=1 ]  ;
            %Shape: Grid [id:dp8005998787827865] 
            % \draw  [draw opacity=0] (193,97) -- (213.5,97) -- (213.5,258.5) -- (193,258.5) -- cycle ; \draw  [color={rgb, 255:red, 0; green, 0; blue, 0 }  ,draw opacity=1 ] (193,97) -- (193,258.5)(213,97) -- (213,258.5) ; \draw  [color={rgb, 255:red, 0; green, 0; blue, 0 }  ,draw opacity=1 ] (193,97) -- (213.5,97)(193,117) -- (213.5,117)(193,137) -- (213.5,137)(193,157) -- (213.5,157)(193,177) -- (213.5,177)(193,197) -- (213.5,197)(193,217) -- (213.5,217)(193,237) -- (213.5,237)(193,257) -- (213.5,257) ; \draw  [color={rgb, 255:red, 0; green, 0; blue, 0 }  ,draw opacity=1 ]  ;
            %Shape: Rectangle [id:dp7094026277789093] 
            \draw  [color={rgb, 255:red, 74; green, 144; blue, 226 }  ,draw opacity=1 ][dash pattern={on 4.5pt off 4.5pt}] (240,90) -- (310,90) -- (310,161.5) -- (240,161.5) -- cycle ;
            %Shape: Rectangle [id:dp460639400683549] 
            \draw  [color={rgb, 255:red, 74; green, 144; blue, 226 }  ,draw opacity=1 ][dash pattern={on 4.5pt off 4.5pt}] (359,187) -- (429,187) -- (429,258.5) -- (359,258.5) -- cycle ;
            %Shape: Rectangle [id:dp3077675403085859] 
            \draw  [color={rgb, 255:red, 128; green, 128; blue, 128 }  ,draw opacity=1 ][fill={rgb, 255:red, 155; green, 155; blue, 155 }  ,fill opacity=1 ] (405,95.5) -- (424.5,95.5) -- (424.5,113.5) -- (405,113.5) -- cycle ;
            %Curve Lines [id:da9884084329614488] 
            \draw [color={rgb, 255:red, 0; green, 0; blue, 0 }  ,draw opacity=1 ]   (310,127) .. controls (405.02,128.49) and (235.21,204.73) .. (356.65,208.45) ;
            \draw [shift={(358.5,208.5)}, rotate = 181.37] [color={rgb, 255:red, 0; green, 0; blue, 0 }  ,draw opacity=1 ][line width=0.75]    (10.93,-3.29) .. controls (6.95,-1.4) and (3.31,-0.3) .. (0,0) .. controls (3.31,0.3) and (6.95,1.4) .. (10.93,3.29)   ;
            %Shape: Rectangle [id:dp4390834825704706] 
            % \draw  [color={rgb, 255:red, 74; green, 144; blue, 226 }  ,draw opacity=1 ][dash pattern={on 4.5pt off 4.5pt}] (189,89) -- (217.5,89) -- (217.5,145.5) -- (189,145.5) -- cycle ;
            %Shape: Rectangle [id:dp5497849642791494] 
            % \draw  [color={rgb, 255:red, 74; green, 144; blue, 226 }  ,draw opacity=1 ][dash pattern={on 4.5pt off 4.5pt}] (189,209.5) -- (217.5,209.5) -- (217.5,261) -- (189,261) -- cycle ;
            %Shape: Rectangle [id:dp831465060670659] 
            % \draw  [color={rgb, 255:red, 155; green, 155; blue, 155 }  ,draw opacity=1 ][fill={rgb, 255:red, 128; green, 128; blue, 128 }  ,fill opacity=1 ] (194,157.5) -- (213.5,157.5) -- (213.5,175.5) -- (194,175.5) -- cycle ;
            %Straight Lines [id:da769886376297227] 
            % \draw [color={rgb, 255:red, 0; green, 0; blue, 0 }  ,draw opacity=1 ]   (202.5,145.5) -- (202.5,204.5) ;
            % \draw [shift={(202.5,206.5)}, rotate = 270] [color={rgb, 255:red, 0; green, 0; blue, 0 }  ,draw opacity=1 ][line width=0.75]    (10.93,-3.29) .. controls (6.95,-1.4) and (3.31,-0.3) .. (0,0) .. controls (3.31,0.3) and (6.95,1.4) .. (10.93,3.29)   ;
            %Curve Lines [id:da5072714715369583] 
            \draw [color={rgb, 255:red, 155; green, 155; blue, 155 }  ,draw opacity=1 ]   (423.5,102.5) .. controls (435.32,99.55) and (421.91,120.85) .. (460.69,120.53) ;
            \draw [shift={(462.5,120.5)}, rotate = 178.6] [color={rgb, 255:red, 155; green, 155; blue, 155 }  ,draw opacity=1 ][line width=0.75]    (10.93,-3.29) .. controls (6.95,-1.4) and (3.31,-0.3) .. (0,0) .. controls (3.31,0.3) and (6.95,1.4) .. (10.93,3.29)   ;
            %Curve Lines [id:da7578465093238256] 
            % \draw [color={rgb, 255:red, 155; green, 155; blue, 155 }  ,draw opacity=1 ]   (192.5,163.5) .. controls (167.54,171.18) and (161.93,180.7) .. (152.68,192.95) ;
            % \draw [shift={(151.5,194.5)}, rotate = 307.57] [color={rgb, 255:red, 155; green, 155; blue, 155 }  ,draw opacity=1 ][line width=0.75]    (10.93,-3.29) .. controls (6.95,-1.4) and (3.31,-0.3) .. (0,0) .. controls (3.31,0.3) and (6.95,1.4) .. (10.93,3.29)   ;
            %Curve Lines [id:da7189609412416013] 
            % \draw [color={rgb, 255:red, 74; green, 144; blue, 226 }  ,draw opacity=1 ]   (185.5,104.5) .. controls (162.23,118.86) and (163.47,135.17) .. (113.04,99.59) ;
            % \draw [shift={(111.5,98.5)}, rotate = 35.43] [color={rgb, 255:red, 74; green, 144; blue, 226 }  ,draw opacity=1 ][line width=0.75]    (10.93,-3.29) .. controls (6.95,-1.4) and (3.31,-0.3) .. (0,0) .. controls (3.31,0.3) and (6.95,1.4) .. (10.93,3.29)   ;
            %Curve Lines [id:da1639863827892727] 
            \draw [color={rgb, 255:red, 74; green, 144; blue, 226 }  ,draw opacity=1 ]   (258.5,87.5) .. controls (298.1,57.8) and (242.63,79.07) .. (280.33,50.38) ;
            \draw [shift={(281.5,49.5)}, rotate = 143.13] [color={rgb, 255:red, 74; green, 144; blue, 226 }  ,draw opacity=1 ][line width=0.75]    (10.93,-3.29) .. controls (6.95,-1.4) and (3.31,-0.3) .. (0,0) .. controls (3.31,0.3) and (6.95,1.4) .. (10.93,3.29)   ;
                
            % \draw (154,267) node [anchor=north west][inner sep=0.75pt]   [align=left] {\textcolor[rgb]{0,0,0}{1D Conv layer}};
            
            \draw (278,269) node [anchor=north west][inner sep=0.75pt]   [align=left] {\textcolor[rgb]{0,0,0}{2D Conv layer}};
            
            \draw (467,105) node [anchor=north west][inner sep=0.75pt]   [align=left] {\textcolor[rgb]{0,0,0}{A pixel in }\\\textcolor[rgb]{0,0,0}{an image}};
            
            % \draw (89,32) node [anchor=north west][inner sep=0.75pt]  [color={rgb, 255:red, 74; green, 144; blue, 226 }  ,opacity=1 ] [align=left] {\textcolor[rgb]{0.29,0.56,0.89}{1D Kernel }\\\textcolor[rgb]{0.29,0.56,0.89}{Feature}\\\textcolor[rgb]{0.29,0.56,0.89}{Detector}};
            
            % \draw (58,161) node [anchor=north west][inner sep=0.75pt]   [align=left] {\textcolor[rgb]{0,0,0}{A value in }\\\textcolor[rgb]{0,0,0}{the ECG }\\\textcolor[rgb]{0,0,0}{signal vector}};
            
            \draw (287,16) node [anchor=north west][inner sep=0.75pt]   [align=left] {\textcolor[rgb]{0.29,0.56,0.89}{2D Kernel }\\\textcolor[rgb]{0.29,0.56,0.89}{Feature}\\\textcolor[rgb]{0.29,0.56,0.89}{Detector}};
            
        \end{tikzpicture}
    }%
    \caption{1D convolution layer vs 2D convolution layer}
    \label{fig:1dcnn}
\end{figure}

\subsection{Split learning}
\label{subsec:splitlear}

It is a collaborative learning technique~\cite{gupta2018distributed} in which an ML model is split into two parts, the client part ($f_{C}$), comprises the first few layers of model ($1, \cdots, l$) -- $l$ being the last layer on the client-side and the server part ($f_{P}$), encompassing the remaining layers of model ($l+1, \cdots , L$) -- $L$ being the last layer on the server-side. The client and server collaborate to train the split model without having access to each other's parts. The client who owns the data uses forward propagation to train their part of the model and sends the $ATm$ from the split layer (final layer of the client-side) to the server. The server continues the forward propagation on the $ATm$ using their part of the model. After completing the forward propagation and computing the loss, the server performs the backward propagation, only returning to the client the gradients to complete the backward propagation. This process is repeated until the model converges and learns a suitable set of parameters. %Although the client and server do not share any raw input data, this configuration requires label sharing. The sharing of labels can be eliminated by using a U-shaped SL configuration. The U-shaped SL configuration is nearly identical to the simple SL setup, with the exception that it does not require clients to share labels~\cite{vepakomma2018split}. On the server-side, the network is wrapped at the end layer, and the outputs are sent back to the client. Upon reception, client generates gradients from the end layers and utilize them for backward propagation without revealing the corresponding labels.

The aim of SL is to protect client privacy by allowing clients to train part of the model, and share $ATm$ (instead of their raw data) with a server running the remaining part of the model. It is also utilized to reduce the client's computational overhead by merely running a few layers rather than the entire model. %\autoref{fig:local vs sl} illustrates three different models: the local model on the left, the vanilla SL model in the middle, and the U-shaped SL model on the right. In the local model, there is no split, while in the vanilla SL there is a split with labels being shared with the server. In contrast, the U-shaped SL model also has a split but with the final layer being executed on the client-side, resulting in no sharing of labels with the server.

\subsection{Function secret sharing}
\label{subsec:fss}

FSS provides a way for additively secret-sharing a function $f$ from a given function family $F$. More concretely, a two-party FSS scheme splits a function $f: {(0,1)}^{n} \rightarrow \textit{G}$, for some abelian group $G$, into functions described by keys such that $f = \mathsf{f_{0}} + \mathsf{f}_{1}$ and every strict subset of the keys hides $f$. An FSS scheme for some class $F$ has two algorithms ($\mathsf{KeyGen, EvalAll}$)~\cite{boyle2015function}:
\begin{itemize}[leftmargin=0.6cm]
	\item $\mathsf{KeyGen(1}^{\lambda}, f) \rightarrow (\mathsf{f_{0}, f_{1}})$: 
	The $\mathsf{KeyGen}$ algorithm that takes the security parameter $1^{\lambda}$ and some secret function $f$ to be shared and outputs two different function shares $(\mathsf{f_{0}, f_{1}})$ (also called function keys).% $(k_{0}, k_{1})$).
	\item $\mathsf{EvalAll}(j, \mathsf{f}_{j}, x_{pub}) \rightarrow \mathsf{f}_{j}(x_{pub})$: The $\mathsf{EvalAll}$ algorithm takes three parameters, the input bit $j \in \{0,1\}$, the key $\mathsf{f}_j$, and the public data $x_{pub}$ and outputs the shares $\mathsf{f}_{j}(x_{pub})$. 
\end{itemize}

\smallskip

\noindent An FSS scheme should satisfy the following two properties:
\begin{itemize}[leftmargin=0.6cm]
	\item \underline{Secrecy}: A single key $(\mathsf{f}_{j})$ hides the original function $f$. 
	\item \underline{Correctness}: Adding the local shares get the same output as the original function $(\mathsf{f_{0}}(x_{pub}) + \mathsf{f_{1}}(x_{pub}) = f(x_{pub}))$.
\end{itemize}
As can be seen in \autoref{fig:fss}, function $f$ is first split into two shares $\mathsf{f_{0}}$ and $\mathsf{f_{1}}$ and then each party locally evaluates the function on the input $x_{pub}$. In FSS, adding the locally evaluated functions $(\mathsf{f_{0}}(x_{pub}) + \mathsf{f_{1}}(x_{pub}))$, gives the additive reconstruction of the original function applied on the input $f(x_{pub})$.

\begin{figure}[!ht]
	\centering
	\tikzset{every picture/.style={line width=0.75pt}} %set default line width to 0.75pt        
	
	\begin{tikzpicture}[x=0.75pt,y=0.75pt,yscale=-1,xscale=1]

		\draw    (156,137) -- (202.18,115.83) ;
		\draw [shift={(204,115)}, rotate = 155.38] [color={rgb, 255:red, 0; green, 0; blue, 0 }  ][line width=0.75]    (10.93,-3.29) .. controls (6.95,-1.4) and (3.31,-0.3) .. (0,0) .. controls (3.31,0.3) and (6.95,1.4) .. (10.93,3.29)   ;

		\draw    (154,151) -- (197.22,173.09) ;
		\draw [shift={(199,174)}, rotate = 207.07] [color={rgb, 255:red, 0; green, 0; blue, 0 }  ][line width=0.75]    (10.93,-3.29) .. controls (6.95,-1.4) and (3.31,-0.3) .. (0,0) .. controls (3.31,0.3) and (6.95,1.4) .. (10.93,3.29)   ;
		%Straight Lines [id:da1479747982354629] 
		\draw    (229,113) -- (279,113) ;
		\draw [shift={(281,113)}, rotate = 180] [color={rgb, 255:red, 0; green, 0; blue, 0 }  ][line width=0.75]    (10.93,-3.29) .. controls (6.95,-1.4) and (3.31,-0.3) .. (0,0) .. controls (3.31,0.3) and (6.95,1.4) .. (10.93,3.29)   ;
		%Straight Lines [id:da6439309973080443] 
		\draw    (229,173) -- (279,173) ;
		\draw [shift={(281,173)}, rotate = 180] [color={rgb, 255:red, 0; green, 0; blue, 0 }  ][line width=0.75]    (10.93,-3.29) .. controls (6.95,-1.4) and (3.31,-0.3) .. (0,0) .. controls (3.31,0.3) and (6.95,1.4) .. (10.93,3.29)   ;
		%Straight Lines [id:da9537830600101493] 
		\draw    (335,117) -- (367.16,135.21) ;
		\draw [shift={(369,135.21)}, rotate = 203.36] [color={rgb, 255:red, 0; green, 0; blue, 0 }  ][line width=0.75]    (10.93,-3.29) .. controls (6.95,-1.4) and (3.31,-0.3) .. (0,0) .. controls (3.31,0.3) and (6.95,1.4) .. (10.93,3.29)   ;
		%Straight Lines [id:da715699290677348] 
		\draw    (335,174) -- (367.16,161.74) ;
		\draw [shift={(369,161)}, rotate = 158.2] [color={rgb, 255:red, 0; green, 0; blue, 0 }  ][line width=0.75]    (10.93,-3.29) .. controls (6.95,-1.4) and (3.31,-0.3) .. (0,0) .. controls (3.31,0.3) and (6.95,1.4) .. (10.93,3.29)   ;

		\draw (139,131.4) node [anchor=north west][inner sep=0.75pt]    {$f$};
		
		\draw (209,105.4) node [anchor=north west][inner sep=0.75pt]    {$\mathsf{f_{0}}$};
		
		\draw (207,165.4) node [anchor=north west][inner sep=0.75pt]    {$\mathsf{f_{1}}$};
		
		\draw (284,165.4) node [anchor=north west][inner sep=0.75pt]    {$\mathsf{f_{1}}(x_{pub})$};
		
		\draw (283,107.4) node [anchor=north west][inner sep=0.75pt]    {$\mathsf{f_{0}}(x_{pub})$};
		
		\draw (366,137.4) node [anchor=north west][inner sep=0.75pt]    {$f(x_{pub})$};
		
		\draw (229,137.4) node [anchor=north west][inner sep=0.75pt]    {$\forall x_{pub}$};
		
		\draw (343,138.4) node [anchor=north west][inner sep=0.75pt]    {$+$};
		
		\draw (150,102.4) node [anchor=north west][inner sep=0.75pt]    {$\mathbf{KeyGen}$};
		
		\draw (235,88.4) node [anchor=north west][inner sep=0.75pt]    {$\mathbf{EvalAll}$};

	\end{tikzpicture}
	\caption{Function secret sharing}
	\label{fig:fss}
\end{figure}
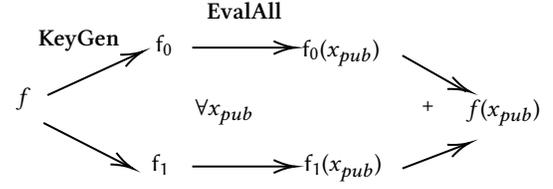

\section{Function Secret Sharing Primitives}
\label{sec:FSS_primitives}
Before going into the details of the FSS primitives, first we define all the notation used in the paper. After that, we go into the details of the FSS primitives taking into consideration ReLU as an example. 

\paragraph*{\textbf{Notation}}
%Notation for all parameters and their descriptions are in \autoref{table:paranddes}. Moving forward, we use $f_{\theta_{c}}$ to denote the client part of the model, while $f_{\theta_{P_{0}}}$ and $f_{\theta_{P_{1}}}$, represent the split part of the remaining model associated with Server~1 ($P_{0}$) and Server~2 ($P_{1}$).

\autoref{table:paranddes} summarizes the core notation of this paper. In addition to that, $f_{\theta_{c}}$ denotes the client part of the model, while $f_{\theta_{P_{0}}}$ and $f_{\theta_{P_{1}}}$, represent the split part of the remaining model associated with Server~1 ($P_{0}$) and Server~2 ($P_{1}$).

\begin{table}
    \centering
    \caption{Parameters and description in the algorithms}
    \resizebox{0.47\textwidth}{!}{%
    \label{table:paranddes}
    \begin{tabular}{l|l|l|l|l}
        \hline
        \rowcolor{gray}
        \color{white}\textbf{\#}	& \color{white}\textbf{ML Parameters} & \color{white}\textbf{Description} & \color{white}\textbf{FSS Parameters} & \color{white}\textbf{Description}\\
        % &Learning& &Secret&\\
        %&Parameters&&Parameters&\\   
        \hline
        1	& $\mathbf{D}$, $m$  & Dataset, Number of data samples & $s_{j}$ & Random Seed \\ \hline
        2	& $x, y$ & Input data samples, Ground-truth labels &  $\alpha$  & Random mask\\ \hline
        3	& $\eta$ & Learning rate  &  $P_{0}$, $P_{1}$ &  Server~1 \ and \ 2\\ \hline
	4	& $p$ & Momentum & $\mathsf{f_0, f_1}$ & Function shares (function keys) \\ \hline
        5	& $n$ & Batch size  & $f_{\theta_{P}}$  & Server-side model \\ \hline
	6	& $N$ & Number of batches to be trained & $f_{\theta_{C}}$   & Client-side model \\ \hline
        7	& $E$ & Number of training epochs   &   $\tilde{f}$ & Encoder \\ \hline
        8	& $f^{i}$ & Linear or non-linear operation of layer $i$ & $\tilde{f}^{-1}$ & Decoder  \\ \hline
        % 7	& $g^{i}$ & Activation function of layer $i$  &   $\tilde{f}^{-1}$ & Decoder \\ \hline
        9	& $ATm^{i}$ & Output activation map of $g^{i}$& $\mathsf{f^{ReLU}_{j}}$  & ReLU FSS key for server $j$  \\ \hline
        10	& $x_{pub}$ & Public input & $\mathsf{f^{FC}_{j}}$  & FC part for server $j$ \\ \hline
    \end{tabular}
    }
\end{table}

\smallskip

As mentioned earlier, the ReLU function returns~0 for negative inputs and keeps positive value as is. To create shares of the ReLU function, we have employed two FSS tests that fit our needs: \textit{(1)} the equality test and \textit{(2)} the comparison test~\cite{boyle2015function, boyle2016function, boyle2019secure, boyle2021function}. Before going into the details of these tests, let's briefly cover the concepts of the Pseudo-random Number Generator (PRNG) and Correction Word (CW).

\smallskip

\noindent\textit{Pseudo-random number generator.} \enskip 
A PRNG is an algorithm that takes a short, random seed and produces a longer sequence of pseudo-random bits. In FSS, it is used to generate the keys for the FSS protocol. Each party starts with a random seed, and uses the PRNG to generate a sequence of keys. These keys are then used to generate CW for the FSS protocol.

\smallskip

\noindent\textit{Correction word.} \enskip 
It is a special value used in FSS to correct the private shares being compared between two parties. The CW shifts a pseudo-random string so that it matches another random string on one side and remains independent on the other side. Each node in the binary decision tree generates its own CW, which is computed by applying a PRNG to the private shares being compared at that node. The resulting value is then used to correct the shares by computing an XOR operation on them with the CW. This ensures that the shares are equal if the input bits are the same, and different if the input bits are different. The CW is passed down to the child nodes of the current node in the binary decision tree, where it is used to correct the shares at those nodes as well. This process continues until the final output of the binary decision tree is computed.

\smallskip

\noindent \textit{Equality test.} \enskip
\label{para:equality}
A test that involves comparing $x_{pub}$ to a private value $\alpha$ and the goal is to determine whether $x_{pub}$ is equal to $\alpha$ or not. The random mask $\alpha$ and the private input $ATm$ are combined using modular arithmetic and published to obtain a public value $x_{pub}$. The equality test is performed using a binary decision tree of depth $n$, where $n$ is the number of bits in the input $x_{pub}$. The tree evaluates the input using $\mathsf{f_0, f_1}$, which are generated using FSS. The path from the root to $\alpha$ is called the special path. % path ensures that the computation is correct even if one of the parties deviates from the protocol. 
The two servers ($P_{0}, P_{1}$) %each 
start from the root and update their state depending on the bit $x_{pub}[i]$ using a CW from $\mathsf{f_0, f_1}$. If $x_{pub}[i] == \alpha[i]$, they stay on the special path and output~1, otherwise, they output~0. 
The main idea is that the servers on the special path have specific states $(s_{0}, t_{0})$ and $(s_{1}, t_{1})$, where $s_{0}$ and $s_{1}$ are independent and identically distributed and $t_{0} \oplus t_{1} = 1$ (see the dotted lines in \autoref{fig:fssequal}). When the servers are out of the special path, the CW ensures that ${s_{0} = s_{1}}$ but remain random and indistinguishable, and ${t_{0} = t_{1}}$, guaranteeing ${t_{0} \oplus t_{1} = 0}$. To reconstruct the result, each server outputs its value $t_{j}$, and the final result is obtained by ${t_{0} \oplus t_{1}}$.

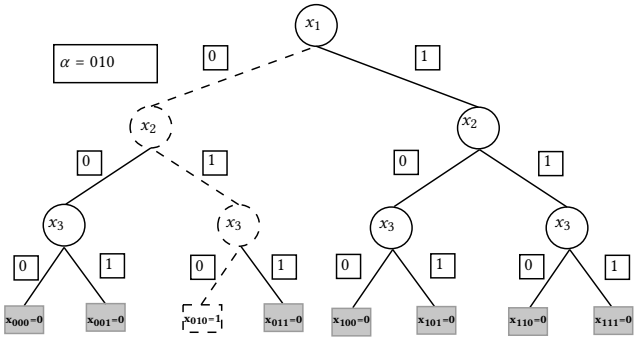
\begin{figure}
\centering
\scalebox{0.75}{
\tikzset{every picture/.style={line width=0.75pt}} 

\begin{tikzpicture}[x=0.75pt,y=0.75pt,yscale=-1,xscale=1]
\draw   (291.99,55.99) .. controls (291.99,48.26) and (298.26,41.99) .. (306,41.99) .. controls (313.74,41.99) and (320.01,48.26) .. (320.01,55.99) .. controls (320.01,63.73) and (313.74,70) .. (306,70) .. controls (298.26,70) and (291.99,63.73) .. (291.99,55.99) -- cycle ;

\draw [line width=0.75]  [dash pattern={on 4.5pt off 4.5pt}]  (194.8,110.4) -- (306,70) ;
\draw    (415,111) -- (306,70) ;
\draw  [dash pattern={on 4.5pt off 4.5pt}][line width=0.75]  (180.79,124.41) .. controls (180.79,116.67) and (187.06,110.4) .. (194.8,110.4) .. controls (202.54,110.4) and (208.81,116.67) .. (208.81,124.41) .. controls (208.81,132.14) and (202.54,138.41) .. (194.8,138.41) .. controls (187.06,138.41) and (180.79,132.14) .. (180.79,124.41) -- cycle ;
\draw   (400.99,125.01) .. controls (400.99,117.27) and (407.26,111) .. (415,111) .. controls (422.74,111) and (429.01,117.27) .. (429.01,125.01) .. controls (429.01,132.74) and (422.74,139.01) .. (415,139.01) .. controls (407.26,139.01) and (400.99,132.74) .. (400.99,125.01) -- cycle ;
\draw    (137.58,176.27) -- (194.8,138.41) ;
\draw [line width=0.75]  [dash pattern={on 4.5pt off 4.5pt}]  (253.8,176.6) -- (194.8,138.41) ;
\draw   (122.39,190.41) .. controls (122.39,182.67) and (128.66,176.4) .. (136.4,176.4) .. controls (144.14,176.4) and (150.41,182.67) .. (150.41,190.41) .. controls (150.41,198.14) and (144.14,204.41) .. (136.4,204.41) .. controls (128.66,204.41) and (122.39,198.14) .. (122.39,190.41) -- cycle ;
\draw  [dash pattern={on 4.5pt off 4.5pt}][line width=0.75]  (239.79,190.61) .. controls (239.79,182.87) and (246.06,176.6) .. (253.8,176.6) .. controls (261.54,176.6) and (267.81,182.87) .. (267.81,190.61) .. controls (267.81,198.34) and (261.54,204.61) .. (253.8,204.61) .. controls (246.06,204.61) and (239.79,198.34) .. (239.79,190.61) -- cycle ;
\draw    (358.38,177.87) -- (415.8,139.6) ;
\draw    (474.6,178.2) -- (415.8,139.6) ;
\draw   (342.39,192.21) .. controls (342.39,184.47) and (348.66,178.2) .. (356.4,178.2) .. controls (364.14,178.2) and (370.41,184.47) .. (370.41,192.21) .. controls (370.41,199.94) and (364.14,206.21) .. (356.4,206.21) .. controls (348.66,206.21) and (342.39,199.94) .. (342.39,192.21) -- cycle ;
\draw   (460.59,192.21) .. controls (460.59,184.47) and (466.86,178.2) .. (474.6,178.2) .. controls (482.34,178.2) and (488.61,184.47) .. (488.61,192.21) .. controls (488.61,199.94) and (482.34,206.21) .. (474.6,206.21) .. controls (466.86,206.21) and (460.59,199.94) .. (460.59,192.21) -- cycle ;
\draw    (109.53,244.59) -- (136.4,205.41) ;
\draw    (164.4,243.6) -- (136.4,205.41) ;
\draw [line width=0.75]  [dash pattern={on 4.5pt off 4.5pt}]  (228.53,243.97) -- (255.4,204.8) ;
\draw    (283.4,242.99) -- (255.4,204.8) ;
\draw    (330.53,245.97) -- (357.4,206.8) ;
\draw    (385.4,244.99) -- (357.4,206.8) ;
\draw    (449.53,245.97) -- (476.4,206.8) ;
\draw    (504.4,244.99) -- (476.4,206.8) ;
\draw  [color={rgb, 255:red, 155; green, 155; blue, 155 }  ,draw opacity=1 ][fill={rgb, 255:red, 155; green, 155; blue, 155 }  ,fill opacity=0.56 ] (96.6,244.8) -- (122.57,244.8) -- (122.57,263.29) -- (96.6,263.29) -- cycle ;
\draw  [dash pattern={on 4.5pt off 4.5pt}][line width=0.75]  (216.4,244) -- (242.37,244) -- (242.37,262.49) -- (216.4,262.49) -- cycle ;
\draw  [color={rgb, 255:red, 155; green, 155; blue, 155 }  ,draw opacity=1 ][fill={rgb, 255:red, 155; green, 155; blue, 155 }  ,fill opacity=0.56 ] (151.4,243.2) -- (177.37,243.2) -- (177.37,261.69) -- (151.4,261.69) -- cycle ;
\draw  [color={rgb, 255:red, 155; green, 155; blue, 155 }  ,draw opacity=1 ][fill={rgb, 255:red, 155; green, 155; blue, 155 }  ,fill opacity=0.56 ] (271,243.4) -- (296.97,243.4) -- (296.97,261.89) -- (271,261.89) -- cycle ;

\draw  [color={rgb, 255:red, 155; green, 155; blue, 155 }  ,draw opacity=1 ][fill={rgb, 255:red, 155; green, 155; blue, 155 }  ,fill opacity=0.56 ] (316.6,246.4) -- (342.57,246.4) -- (342.57,264.89) -- (316.6,264.89) -- cycle ;

\draw  [color={rgb, 255:red, 155; green, 155; blue, 155 }  ,draw opacity=1 ][fill={rgb, 255:red, 155; green, 155; blue, 155 }  ,fill opacity=0.56 ] (373.8,245.2) -- (399.77,245.2) -- (399.77,263.69) -- (373.8,263.69) -- cycle ;

\draw  [color={rgb, 255:red, 155; green, 155; blue, 155 }  ,draw opacity=1 ][fill={rgb, 255:red, 155; green, 155; blue, 155 }  ,fill opacity=0.56 ] (435.4,246) -- (461.37,246) -- (461.37,264.49) -- (435.4,264.49) -- cycle ;

\draw  [color={rgb, 255:red, 155; green, 155; blue, 155 }  ,draw opacity=1 ][fill={rgb, 255:red, 155; green, 155; blue, 155 }  ,fill opacity=0.56 ] (492.8,245.4) -- (518.77,245.4) -- (518.77,263.89) -- (492.8,263.89) -- cycle ;

\draw   (145.4,140.49) -- (161.57,140.49) -- (161.57,156.65) -- (145.4,156.65) -- cycle ;
\draw   (230.2,69.69) -- (246.37,69.69) -- (246.37,85.85) -- (230.2,85.85) -- cycle ;
\draw   (102.8,211.09) -- (118.97,211.09) -- (118.97,227.25) -- (102.8,227.25) -- cycle ;

\draw   (221.6,210.69) -- (237.77,210.69) -- (237.77,226.85) -- (221.6,226.85) -- cycle ;

\draw   (319.6,210.09) -- (335.77,210.09) -- (335.77,226.25) -- (319.6,226.25) -- cycle ;

\draw   (442.8,210.29) -- (458.97,210.29) -- (458.97,226.45) -- (442.8,226.45) -- cycle ;

\draw   (358.4,140.09) -- (374.57,140.09) -- (374.57,156.25) -- (358.4,156.25) -- cycle ;

\draw   (372.8,69.69) -- (388.97,69.69) -- (388.97,85.85) -- (372.8,85.85) -- cycle ;

\draw   (276.4,210.49) -- (292.57,210.49) -- (292.57,226.65) -- (276.4,226.65) -- cycle ;

\draw   (383,209.69) -- (399.17,209.69) -- (399.17,225.85) -- (383,225.85) -- cycle ;

\draw   (500.4,210.69) -- (516.57,210.69) -- (516.57,226.85) -- (500.4,226.85) -- cycle ;

\draw   (229.8,140.49) -- (245.97,140.49) -- (245.97,156.65) -- (229.8,156.65) -- cycle ;

\draw   (455.6,140.89) -- (471.77,140.89) -- (471.77,157.05) -- (455.6,157.05) -- cycle ;

\draw   (160.8,209.69) -- (176.97,209.69) -- (176.97,225.85) -- (160.8,225.85) -- cycle ;

\draw   (129.45,68.95) -- (199.05,68.95) -- (199.05,94.49) -- (129.45,94.49) -- cycle ;

\draw (186.4,119.8) node [anchor=north west][inner sep=0.75pt]   [align=left] {$\displaystyle x_{2}$};

\draw (297.2,50.6) node [anchor=north west][inner sep=0.75pt]   [align=left] {$\displaystyle x_{1}$};

    \draw (402.4,116.6) node [anchor=north west][inner sep=0.75pt]   [align=left] {$\displaystyle x_{2}$};

\draw (124.4,185) node [anchor=north west][inner sep=0.75pt]   [align=left] {$\displaystyle x_{3}$};

\draw (243.8,185.2) node [anchor=north west][inner sep=0.75pt]   [align=left] {$\displaystyle x_{3}$};

\draw (346.8,185.8) node [anchor=north west][inner sep=0.75pt]   [align=left] {$\displaystyle x_{3}$};

\draw (465,185.2) node [anchor=north west][inner sep=0.75pt]   [align=left] {$\displaystyle x_{3}$};

\draw (231.8,71.2) node [anchor=north west][inner sep=0.75pt]   [align=left] {0};

\draw (215.4,249) node [anchor=north west][inner sep=0.75pt]   [align=left] {\textbf{{\scriptsize x\textsubscript{010}=1}}};

\draw (132,75.2) node [anchor=north west][inner sep=0.75pt]   [align=left] {$\displaystyle \alpha =010$};

\draw (162.4,211.2) node [anchor=north west][inner sep=0.75pt]   [align=left] {1};

\draw (457.2,141.4) node [anchor=north west][inner sep=0.75pt]   [align=left] {1};

\draw (231.4,141) node [anchor=north west][inner sep=0.75pt]   [align=left] {1};

\draw (502,212.2) node [anchor=north west][inner sep=0.75pt]   [align=left] {1};

\draw (384.6,211.2) node [anchor=north west][inner sep=0.75pt]   [align=left] {1};

\draw (278,212) node [anchor=north west][inner sep=0.75pt]   [align=left] {1};

\draw (374.4,71.2) node [anchor=north west][inner sep=0.75pt]   [align=left] {1};

\draw (360,141.6) node [anchor=north west][inner sep=0.75pt]   [align=left] {0};

\draw (444.4,211.8) node [anchor=north west][inner sep=0.75pt]   [align=left] {0};

\draw (321.2,211.6) node [anchor=north west][inner sep=0.75pt]   [align=left] {0};

\draw (223.2,212.2) node [anchor=north west][inner sep=0.75pt]   [align=left] {0};

\draw (104.4,212.6) node [anchor=north west][inner sep=0.75pt]   [align=left] {0};

\draw (147,142) node [anchor=north west][inner sep=0.75pt]   [align=left] {0};

\draw (491.8,250) node [anchor=north west][inner sep=0.75pt]   [align=left] {\textbf{{\scriptsize x\textsubscript{111}=0}}};

\draw (434.4,250) node [anchor=north west][inner sep=0.75pt]   [align=left] {\textbf{{\scriptsize x\textsubscript{110}=0}}};

\draw (372.8,250) node [anchor=north west][inner sep=0.75pt]   [align=left] {\textbf{{\scriptsize x\textsubscript{101}=0}}};

\draw (315.6,250) node [anchor=north west][inner sep=0.75pt]   [align=left] {\textbf{{\scriptsize x\textsubscript{100}=0}}};

\draw (270,250) node [anchor=north west][inner sep=0.75pt]   [align=left] {\textbf{{\scriptsize x\textsubscript{011}=0}}};

\draw (150.4,250) node [anchor=north west][inner sep=0.75pt]   [align=left] {\textbf{{\scriptsize x\textsubscript{001}=0}}};

\draw (95.6,250) node [anchor=north west][inner sep=0.75pt]   [align=left] {\textbf{{\scriptsize x\textsubscript{000}=0}}};

\end{tikzpicture}
}
 \caption{Function secret sharing for equality test}
	\label{fig:fssequal}
\end{figure}

\smallskip

\noindent \textit{Comparison test.} \enskip
\label{para:comparison}
The comparison test involves a binary decision tree with a depth of $n$. In this test, the path from the root down to the greater value is called the special path. %The servers are given $\mathsf{f_{0}}$ and $\mathsf{f_{1}}$. %which include a distinct initial random state $(s, t)$. 
Each server starts at the root of the binary decision tree and updates their state based on the bits using CW from $\mathsf{f_{0}}$ and $\mathsf{f_{1}}$. The goal is to evaluate all the paths simultaneously. When a server deviates from the special path, they either fall on the left side for $x_{pub} < \alpha$ or on the right side for $x_{pub} > \alpha$ (see \autoref{fig:fsscom}). Hence, at each node, an additional step is taken where a leaf label is output depending on the bit value $x_{pub}[i]$. The label is~1 only if $x_{pub}[i] < \alpha[i]$ and all previous bits are same. Only one label, either the final label (which corresponds to $x_{pub} = \alpha$) or the leaf label can be equal to~1, because only one path is taken. Therefore, the servers return the sum of all the labels to obtain the final output.

\begin{figure}
\centering
\scalebox{0.75}{
\tikzset{every picture/.style={line width=0.75pt}} 
\begin{tikzpicture}[x=0.75pt,y=0.75pt,yscale=-1,xscale=1]

\draw   (291.99,55.99) .. controls (291.99,48.26) and (298.26,41.99) .. (306,41.99) .. controls (313.74,41.99) and (320.01,48.26) .. (320.01,55.99) .. controls (320.01,63.73) and (313.74,70) .. (306,70) .. controls (298.26,70) and (291.99,63.73) .. (291.99,55.99) -- cycle ;
\draw  [dash pattern={on 4.5pt off 4.5pt}]  (194.8,110.4) -- (306,70) ;
\draw    (415,111) -- (306,70) ;
\draw  [dash pattern={on 4.5pt off 4.5pt}] (180.79,124.41) .. controls (180.79,116.67) and (187.06,110.4) .. (194.8,110.4) .. controls (202.54,110.4) and (208.81,116.67) .. (208.81,124.41) .. controls (208.81,132.14) and (202.54,138.41) .. (194.8,138.41) .. controls (187.06,138.41) and (180.79,132.14) .. (180.79,124.41) -- cycle ;
\draw   (400.99,125.01) .. controls (400.99,117.27) and (407.26,111) .. (415,111) .. controls (422.74,111) and (429.01,117.27) .. (429.01,125.01) .. controls (429.01,132.74) and (422.74,139.01) .. (415,139.01) .. controls (407.26,139.01) and (400.99,132.74) .. (400.99,125.01) -- cycle ;
\draw  [dash pattern={on 4.5pt off 4.5pt}]  (137.58,176.27) -- (194.8,138.41) ;
\draw  [dash pattern={on 4.5pt off 4.5pt}]  (253.8,176.6) -- (194.8,138.41) ;
\draw  [dash pattern={on 4.5pt off 4.5pt}] (122.39,190.41) .. controls (122.39,182.67) and (128.66,176.4) .. (136.4,176.4) .. controls (144.14,176.4) and (150.41,182.67) .. (150.41,190.41) .. controls (150.41,198.14) and (144.14,204.41) .. (136.4,204.41) .. controls (128.66,204.41) and (122.39,198.14) .. (122.39,190.41) -- cycle ;
\draw  [dash pattern={on 4.5pt off 4.5pt}] (239.79,190.61) .. controls (239.79,182.87) and (246.06,176.6) .. (253.8,176.6) .. controls (261.54,176.6) and (267.81,182.87) .. (267.81,190.61) .. controls (267.81,198.34) and (261.54,204.61) .. (253.8,204.61) .. controls (246.06,204.61) and (239.79,198.34) .. (239.79,190.61) -- cycle ;
\draw    (358.38,177.87) -- (415.8,139.6) ;
\draw    (474.6,178.2) -- (415.8,139.6) ;
\draw   (342.39,192.21) .. controls (342.39,184.47) and (348.66,178.2) .. (356.4,178.2) .. controls (364.14,178.2) and (370.41,184.47) .. (370.41,192.21) .. controls (370.41,199.94) and (364.14,206.21) .. (356.4,206.21) .. controls (348.66,206.21) and (342.39,199.94) .. (342.39,192.21) -- cycle ;
\draw   (460.59,192.21) .. controls (460.59,184.47) and (466.86,178.2) .. (474.6,178.2) .. controls (482.34,178.2) and (488.61,184.47) .. (488.61,192.21) .. controls (488.61,199.94) and (482.34,206.21) .. (474.6,206.21) .. controls (466.86,206.21) and (460.59,199.94) .. (460.59,192.21) -- cycle ;
\draw  [dash pattern={on 4.5pt off 4.5pt}]  (109.53,244.59) -- (131.27,212.89) -- (136.4,205.41) ;
\draw  [dash pattern={on 4.5pt off 4.5pt}]  (164.4,243.6) -- (136.4,205.41) ;
\draw  [dash pattern={on 4.5pt off 4.5pt}]  (228.53,243.97) -- (255.4,204.8) ;
\draw    (283.4,242.99) -- (255.4,204.8) ;
\draw    (330.53,245.97) -- (357.4,206.8) ;
\draw    (385.4,244.99) -- (357.4,206.8) ;
\draw    (449.53,245.97) -- (476.4,206.8) ;
\draw    (504.4,244.99) -- (476.4,206.8) ;
\draw  [color={rgb, 255:red, 0; green, 0; blue, 0 }  ,draw opacity=1 ][fill={rgb, 255:red, 255; green, 255; blue, 255 }  ,fill opacity=1 ][dash pattern={on 4.5pt off 4.5pt}] (96.6,244.8) -- (122.57,244.8) -- (122.57,263.29) -- (96.6,263.29) -- cycle ;
\draw  [dash pattern={on 4.5pt off 4.5pt}] (216,244) -- (241.97,244) -- (241.97,262.49) -- (216,262.49) -- cycle ;
\draw  [color={rgb, 255:red, 0; green, 0; blue, 0 }  ,draw opacity=1 ][fill={rgb, 255:red, 255; green, 255; blue, 255 }  ,fill opacity=1 ][dash pattern={on 4.5pt off 4.5pt}] (152,243.2) -- (177.97,243.2) -- (177.97,261.69) -- (152,261.69) -- cycle ;
\draw  [color={rgb, 255:red, 155; green, 155; blue, 155 }  ,draw opacity=1 ][fill={rgb, 255:red, 155; green, 155; blue, 155 }  ,fill opacity=0.56 ] (271,243.4) -- (296.97,243.4) -- (296.97,261.89) -- (271,261.89) -- cycle ;
\draw  [color={rgb, 255:red, 155; green, 155; blue, 155 }  ,draw opacity=1 ][fill={rgb, 255:red, 155; green, 155; blue, 155 }  ,fill opacity=0.56 ] (316.6,246.4) -- (342.57,246.4) -- (342.57,264.89) -- (316.6,264.89) -- cycle ;
\draw  [color={rgb, 255:red, 155; green, 155; blue, 155 }  ,draw opacity=1 ][fill={rgb, 255:red, 155; green, 155; blue, 155 }  ,fill opacity=0.56 ] (373.8,245.2) -- (399.77,245.2) -- (399.77,263.69) -- (373.8,263.69) -- cycle ;
\draw  [color={rgb, 255:red, 155; green, 155; blue, 155 }  ,draw opacity=1 ][fill={rgb, 255:red, 155; green, 155; blue, 155 }  ,fill opacity=0.56 ] (435.4,246) -- (461.37,246) -- (461.37,264.49) -- (435.4,264.49) -- cycle ;
\draw  [color={rgb, 255:red, 155; green, 155; blue, 155 }  ,draw opacity=1 ][fill={rgb, 255:red, 155; green, 155; blue, 155 }  ,fill opacity=0.56 ] (492.8,245.4) -- (518.77,245.4) -- (518.77,263.89) -- (492.8,263.89) -- cycle ;

\draw  [dash pattern={on 4.5pt off 4.5pt}] (145.4,140.49) -- (161.57,140.49) -- (161.57,156.65) -- (145.4,156.65) -- cycle ;
\draw  [dash pattern={on 4.5pt off 4.5pt}] (230.2,69.69) -- (246.37,69.69) -- (246.37,85.85) -- (230.2,85.85) -- cycle ;
\draw  [dash pattern={on 4.5pt off 4.5pt}] (102.8,211.09) -- (118.97,211.09) -- (118.97,227.25) -- (102.8,227.25) -- cycle ;
\draw   (221.6,210.69) -- (237.77,210.69) -- (237.77,226.85) -- (221.6,226.85) -- cycle ;

\draw   (319.6,210.09) -- (335.77,210.09) -- (335.77,226.25) -- (319.6,226.25) -- cycle ;
\draw   (442.8,210.29) -- (458.97,210.29) -- (458.97,226.45) -- (442.8,226.45) -- cycle ;

\draw   (358.4,140.09) -- (374.57,140.09) -- (374.57,156.25) -- (358.4,156.25) -- cycle ;

\draw   (372.8,69.69) -- (388.97,69.69) -- (388.97,85.85) -- (372.8,85.85) -- cycle ;

\draw   (276.4,210.49) -- (292.57,210.49) -- (292.57,226.65) -- (276.4,226.65) -- cycle ;

\draw   (383,209.69) -- (399.17,209.69) -- (399.17,225.85) -- (383,225.85) -- cycle ;

\draw   (500.4,210.69) -- (516.57,210.69) -- (516.57,226.85) -- (500.4,226.85) -- cycle ;

\draw   (229.8,140.49) -- (245.97,140.49) -- (245.97,156.65) -- (229.8,156.65) -- cycle ;

\draw   (455.6,140.89) -- (471.77,140.89) -- (471.77,157.05) -- (455.6,157.05) -- cycle ;

\draw  [dash pattern={on 4.5pt off 4.5pt}] (160.8,209.69) -- (176.97,209.69) -- (176.97,225.85) -- (160.8,225.85) -- cycle ;
\draw   (129.45,68.95) -- (199.05,68.95) -- (199.05,94.49) -- (129.45,94.49) -- cycle ;

\draw (186.4,118.8) node [anchor=north west][inner sep=0.75pt]   [align=left] {$\displaystyle x_{2}$};

\draw (297.2,49.6) node [anchor=north west][inner sep=0.75pt]   [align=left] {$\displaystyle x_{1}$};

\draw (406.4,118.6) node [anchor=north west][inner sep=0.75pt]   [align=left] {$\displaystyle x_{2}$};

\draw (124.4,185) node [anchor=north west][inner sep=0.75pt]   [align=left] {$\displaystyle x_{3}$};

\draw (243.8,185.2) node [anchor=north west][inner sep=0.75pt]   [align=left] {$\displaystyle x_{3}$};

\draw (346.8,184.8) node [anchor=north west][inner sep=0.75pt]   [align=left] {$\displaystyle x_{3}$};

\draw (465,185.2) node [anchor=north west][inner sep=0.75pt]   [align=left] {$\displaystyle x_{3}$};

\draw (231.8,71.2) node [anchor=north west][inner sep=0.75pt]   [align=left] {0};

\draw (214.6,250) node [anchor=north west][inner sep=0.75pt]   [align=left] {\textbf{{\scriptsize x\textsubscript{010}=1}}};

\draw (147,142) node [anchor=north west][inner sep=0.75pt]   [align=left] {0};

\draw (104.4,212.6) node [anchor=north west][inner sep=0.75pt]   [align=left] {0};

\draw (162.4,211.2) node [anchor=north west][inner sep=0.75pt]   [align=left] {1};

\draw (150.4,250) node [anchor=north west][inner sep=0.75pt]   [align=left] {\textbf{{\scriptsize x\textsubscript{001}=1}}};

\draw (95.6,250) node [anchor=north west][inner sep=0.75pt]   [align=left] {\textbf{{\scriptsize x\textsubscript{000}=1}}};

\draw (132,77.2) node [anchor=north west][inner sep=0.75pt]   [align=left] {$\displaystyle \alpha =010$};

\draw (457.2,143.4) node [anchor=north west][inner sep=0.75pt]   [align=left] {1};

\draw (231.4,142) node [anchor=north west][inner sep=0.75pt]   [align=left] {1};

\draw (502,212.2) node [anchor=north west][inner sep=0.75pt]   [align=left] {1};

\draw (384.6,211.2) node [anchor=north west][inner sep=0.75pt]   [align=left] {1};

\draw (278,212) node [anchor=north west][inner sep=0.75pt]   [align=left] {1};

\draw (374.4,71.2) node [anchor=north west][inner sep=0.75pt]   [align=left] {1};

\draw (360,141.6) node [anchor=north west][inner sep=0.75pt]   [align=left] {0};

\draw (444.4,211.8) node [anchor=north west][inner sep=0.75pt]   [align=left] {0};

\draw (321.2,211.6) node [anchor=north west][inner sep=0.75pt]   [align=left] {0};

\draw (223.2,211.2) node [anchor=north west][inner sep=0.75pt]   [align=left] {0};

\draw (491.8,250) node [anchor=north west][inner sep=0.75pt]   [align=left] {\textbf{{\scriptsize x\textsubscript{111}=0}}};

\draw (434.4,250) node [anchor=north west][inner sep=0.75pt]   [align=left] {\textbf{{\scriptsize x\textsubscript{110}=0}}};

\draw (372.8,250) node [anchor=north west][inner sep=0.75pt]   [align=left] {\textbf{{\scriptsize x\textsubscript{101}=0}}};

\draw (315.6,250) node [anchor=north west][inner sep=0.75pt]   [align=left] {\textbf{{\scriptsize x\textsubscript{100}=0}}};

\draw (270,250) node [anchor=north west][inner sep=0.75pt]   [align=left] {\textbf{{\scriptsize x\textsubscript{011}=0}}};

\end{tikzpicture}

}
 \caption{Function secret sharing for comparison test}
	\label{fig:fsscom}
\end{figure}
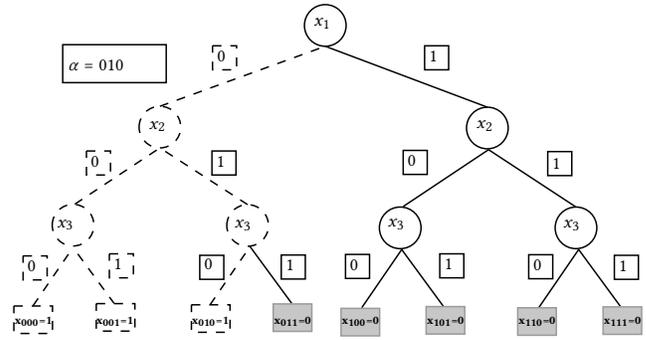

\section{Related Work} 
\label{sec:related_Work}
% In this section we cover relevant state-of-the-art works in PPML which make use of FSS or SL.

SL provides significant advantages in terms of reducing the computational cost of training deep neural networks. However, a potential vulnerability arises during end-to-end training, where the exchange of gradients at the cut layer could inadvertently encode private features or labels, posing a risk to data privacy. Abuadbba et al.~\cite{abuadbba2020can} proposed using SL when training and classifying 1-dimensional data on CNN models. In their work, they use two CNN models, one with two Conv2D layers and another with three Conv2D layers to see the applicability of SL on time-series data, when compared to a centralized system. They identified that SL achieves comparable accuracy to centralized models, but noted that SL by itself has multiple privacy leakages, which hamper the confidentiality of the input data. Moreover, researchers have found that SL is susceptible to different types of attacks, such FSHA~\cite{pasquini2021unleashing}. %, Label inference attack~\cite{erdougan2022unsplit}, Exact attack~\cite{qiu2023exact} and Visual invertability~\cite{abuadbba2020can}.

The leakages highlight the importance of addressing data privacy concerns when implementing SL techniques. As such, researchers have attempted combining different cryptographic techniques with SL to address these attacks. For example, Abuadbba et al.~\cite{abuadbba2020can} tried two privacy mitigation techniques -- adding more hidden layers on the client-side and using DP. However, both of these techniques %suffer from a loss of model accuracy,
indicate a trade-off between reducing privacy leakage and experiencing a decline in model accuracy, particularly when DP is used. Thapa \textit{et al.}~\cite{thapa2022splitfed} proposed to combine SL with FL to reduce the computational complexity of SL. To address the remaining privacy leakage, the authors made use of DP-enhanced client model training and PixelDP~\cite{lecuyer2019certified}, however faced the same drawback of greatly reduced model utility. Some other works like~\cite{khan2023split, khan2023love, khan2023more}  make use of HE to encrypt the intermediate $ATm$ of SL before outsourcing them to the server. %Through HE, 
This way, the model retains equivalent accuracy when compared to plaintext SL, and addresses the privacy leakages known in SL.
However, a key limitation of this method is identified during backward propagation. By analyzing the gradients sent from the client, the server gains valuable information about the client's input data, leading to potential privacy leaks. To tackle this, an improved protocol is proposed by Nguyen \textit{et al.}~\cite{nguyen2023split}, aiming to minimize privacy leakage. The authors~\cite{nguyen2023split} claim that their proposed approach not only helps maintain privacy but also accelerates the training time and reduces the communication overhead compared to~\cite{khan2023more}). Despite these advancements, the training time and communication cost in~\cite{nguyen2023split} continue to pose challenges, even for a simple model featuring two convolutional layers on the client-side (processing plaintext data) and only one fully connected layer on the server-side (processing encrypted data). Specifically, it takes a substantial amount of time and communication cost, which poses a challenge for practical implementation.

There is a plethora of works, which employ FSS in PPML, such as AriaNN~\cite{ryffel2020ariann} which is the first work that makes use of FSS for secure neural network training and inference against a semi-honest adversary. 
Pika~\cite{wagh2022pika} introduces a 3PC PPML solution against either a semi-honest or a malicious adversary. 
Another recent FSS work is LLAMA~\cite{gupta2022llama}, which implements a low-latency library for computing non-linear mathematical functions, useful to PPML, such as sigmoid, tanh, reciprocal square root, using FSS primitives. 
Orca~\cite{jawalkar2023orca} directly uses the library proposed in LLAMA~\cite{gupta2022llama} performs secure training on neural network architectures. 
SIGMA~\cite{gupta2023sigma} demonstrates how FSS can be used for secure inference of generative pre-trained transformers, which contain GeLU, Softmax,
and layer normalizations, and shows that the secure inference can be accelerated through GPUs.
Agarwal~\textit{et al.}~\cite{agarwal2022communication} designs a secure logistic regression training construction with FSS.  
FssNN~\cite{yang2023fssnn} proposes a two-party FSS framework for secure NN training and makes improvements to various aspects in prior arts. 

In all of the covered FSS works, the researchers employ FSS on the entire model which in turn increases computational complexity and communication overhead. However, to the best of our knowledge, the idea of combining SL with FSS has not been explored in the literature. % has not been used in conjunction with FSS in any works. 
To this end, our approach provides new insights into how the use of these techniques can improve the privacy guarantees of ML applications, solve the privacy leakage caused by SL and reduce both the communication and computation costs of FSS.

\section{Architecture}
\label{sec:architecture}
This section provides an overview of the ML model used for classification. We begin with the description of the initial non-split or local model followed by a brief explanation of a related FSS model from a previous work~\cite{ryffel2020ariann}.  Then, we explain how we convert the local model into a split model. Finally, we introduce our private vanilla SL protocl. %two proposed protocol %s: one using FSS for private vanilla SL. % SL and another for U-shaped SL. 
Additionally, we introduce the key entities involved in the training process of the split model, namely the client and the server, elaborating on their respective roles and assigned parameters during the training phase.

\subsection{Local model without split}
\label{subsec:nosplit}

\autoref{fig:slfs} illustrates our local model, which consists of two Conv2D layers, two max-pooling layers, two FC layers and four ReLU layers. These layers operate without any split between client and server, with one party performing computations on the entire model. The training process of this model can be described as follows. 

%Given a training dataset $\mathbf{|D|} = \{(x_{i}, y_{i})| i = 1,2,\ldots, m\}$ which consists of $m$ data samples. 
Let $\mathbf{|D|} = \{(x_{i}, y_{i})| i = 1,2,\ldots, m\}$ be a training dataset consisting of $m$ data samples. Each data sample $x$ has a corresponding encoded label vector $y$ that represents its ground-truth class. 
Then, $f_{\theta}$ is a function representing the local model with $\theta$ being a set of adjustable parameters. 
%We can write the local model as a function $f_{\theta}$, where $\theta$ is a set of adjustable parameters. 
$\theta$ is first initialized to small random values denoted as $\Phi$. The goal is to find the optimal parameters $\theta$ to map $x$ to a predicted output vector $\hat{y}$, where $\hat{y}$ is as close as possible to $y$. $\hat{y}$ %can be 
is a vector of %m 
probabilities (called posteriors), and we use the %infer $x$ to belong to the class with the
highest probability to determine the class to which $x$ belongs. To find the closest value of $\hat{y}$ with respect to $y$, we try to minimize a loss function $J=\mathcal{L}(\hat{y}, y)$. Training the local model is an iterative process to find the best $\theta$ to minimize $J$. This process involves %consists of 
two steps: %sub-processes called 
\textit{forward} %propagation 
and \textit{backward} propagation. The
forward propagation computes each layer in the network %all of the layers of a model 
(\autoref{fig:slfs}) on %the 
input $x$, %throughout the network (which consists of all the layers of \autoref{fig: slfs}), 
reaching the output layer to predict % and producing the predicted output
$\hat{y}$. %Conversely, 
Backward propagation starts from the %network’s 
output layer and goes back to the input layer to calculate the gradients of the loss function $J$ w.r.t the network's weights $\theta$. % of the network. 
These weights are then updated according to the calculated gradients.  
We train the local model with numerous %thousands of 
samples of $x$ and corresponding $y$, through many iterations of forward and backward propagation. We do not train the network on each single data example, but use a number of them at a time (using batch size $n$). The total number of training batches is $N = \frac{\mathbf{|D|}}{n}$ , where $\mathbf{|D|}$ is the size of the dataset. Once the model goes through all the training batches, it has completed one training $E$ and this process continues %repeats 
for a total of $E$ epochs. % in total. 
We implement and reproduce the results for this model. The best test accuracy that we obtained for this model is~99.29\%.

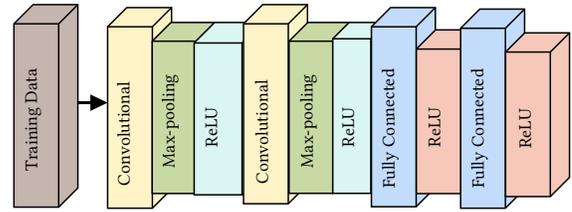
\begin{figure}
	\centering
\tikzset{every picture/.style={line width=0.75pt}}
\scalebox{0.75}{

\tikzset{every picture/.style={line width=0.75pt}} 

\begin{tikzpicture}[x=0.75pt,y=0.75pt,yscale=-1,xscale=1]

\draw  [fill={rgb, 255:red, 255; green, 244; blue, 199 }  ,fill opacity=1 ] (182,62.62) -- (194.9,49.72) -- (225,49.72) -- (225,171.82) -- (212.1,184.72) -- (182,184.72) -- cycle ; \draw   (225,49.72) -- (212.1,62.62) -- (182,62.62) ; \draw   (212.1,62.62) -- (212.1,184.72) ;
%Shape: Cube [id:dp15690098520285733] 
\draw  [fill={rgb, 255:red, 200; green, 218; blue, 164 }  ,fill opacity=1 ] (212,72.32) -- (224.6,59.72) -- (254,59.72) -- (254,162.22) -- (241.4,174.82) -- (212,174.82) -- cycle ; \draw   (254,59.72) -- (241.4,72.32) -- (212,72.32) ; \draw   (241.4,72.32) -- (241.4,174.82) ;
%Shape: Cube [id:dp07011132670875497] 
\draw  [fill={rgb, 255:red, 218; green, 246; blue, 242 }  ,fill opacity=1 ] (240.2,73.52) -- (254,59.72) -- (286.2,59.72) -- (286.2,161.92) -- (272.4,175.72) -- (240.2,175.72) -- cycle ; \draw   (286.2,59.72) -- (272.4,73.52) -- (240.2,73.52) ; \draw   (272.4,73.52) -- (272.4,175.72) ;
%Shape: Cube [id:dp9052407084984176] 
\draw  [fill={rgb, 255:red, 197; green, 181; blue, 175 }  ,fill opacity=1 ] (119,60.62) -- (131.9,47.72) -- (162,47.72) -- (162,171.1) -- (149.1,184) -- (119,184) -- cycle ; \draw   (162,47.72) -- (149.1,60.62) -- (119,60.62) ; \draw   (149.1,60.62) -- (149.1,184) ;
%Shape: Cube [id:dp946083645949927] 
\draw  [fill={rgb, 255:red, 255; green, 244; blue, 199 }  ,fill opacity=1 ] (273.4,60.62) -- (286.3,47.72) -- (316.4,47.72) -- (316.4,168.82) -- (303.5,181.72) -- (273.4,181.72) -- cycle ; \draw   (316.4,47.72) -- (303.5,60.62) -- (273.4,60.62) ; \draw   (303.5,60.62) -- (303.5,181.72) ;
%Shape: Cube [id:dp561876895010877] 
\draw  [fill={rgb, 255:red, 200; green, 218; blue, 164 }  ,fill opacity=1 ] (304,72.32) -- (316.6,59.72) -- (346,59.72) -- (346,162.12) -- (333.4,174.72) -- (304,174.72) -- cycle ; \draw   (346,59.72) -- (333.4,72.32) -- (304,72.32) ; \draw   (333.4,72.32) -- (333.4,174.72) ;
%Shape: Cube [id:dp6824115334908267] 
\draw  [fill={rgb, 255:red, 218; green, 246; blue, 242 }  ,fill opacity=1 ] (333.4,70.7) -- (344.38,59.72) -- (370,59.72) -- (370,163.74) -- (359.02,174.72) -- (333.4,174.72) -- cycle ; \draw   (370,59.72) -- (359.02,70.7) -- (333.4,70.7) ; \draw   (359.02,70.7) -- (359.02,174.72) ;
%Straight Lines [id:da8244967353500522] 
\draw [line width=1.5]    (162,113.72) -- (178,113.72) ;
\draw [shift={(182,113.72)}, rotate = 180] [fill={rgb, 255:red, 0; green, 0; blue, 0 }  ][line width=0.08]  [draw opacity=0] (11.61,-5.58) -- (0,0) -- (11.61,5.58) -- cycle    ;
%Shape: Cube [id:dp2029841540595868] 
\draw  [fill={rgb, 255:red, 195; green, 220; blue, 252 }  ,fill opacity=1 ] (359.4,62.62) -- (372.3,49.72) -- (402.4,49.72) -- (402.4,170.82) -- (389.5,183.72) -- (359.4,183.72) -- cycle ; \draw   (402.4,49.72) -- (389.5,62.62) -- (359.4,62.62) ; \draw   (389.5,62.62) -- (389.5,183.72) ;

%Shape: Cube [id:dp16251121133381885] 
\draw  [fill={rgb, 255:red, 246; green, 200; blue, 185 }  ,fill opacity=1 ] (389.6,77.62) -- (402.5,64.72) -- (432.6,64.72) -- (432.6,161.99) -- (419.7,174.89) -- (389.6,174.89) -- cycle ; \draw   (432.6,64.72) -- (419.7,77.62) -- (389.6,77.62) ; \draw   (419.7,77.62) -- (419.7,174.89) ;
%Shape: Cube [id:dp24336407037889285] 
\draw  [fill={rgb, 255:red, 195; green, 220; blue, 252 }  ,fill opacity=1 ] (419.4,63.62) -- (432.3,50.72) -- (462.4,50.72) -- (462.4,171.82) -- (449.5,184.72) -- (419.4,184.72) -- cycle ; \draw   (462.4,50.72) -- (449.5,63.62) -- (419.4,63.62) ; \draw   (449.5,63.62) -- (449.5,184.72) ;

%Shape: Cube [id:dp730727199775288] 
\draw  [fill={rgb, 255:red, 246; green, 200; blue, 185 }  ,fill opacity=1 ] (449.6,79.62) -- (462.5,66.72) -- (492.6,66.72) -- (492.6,163.99) -- (479.7,176.89) -- (449.6,176.89) -- cycle ; \draw   (492.6,66.72) -- (479.7,79.62) -- (449.6,79.62) ; \draw   (479.7,79.62) -- (479.7,176.89) ;

\draw (187.42,182.8) node [anchor=north west][inner sep=0.75pt]  [rotate=-269.64] [align=left] {\ \ \ \ Convolutional};

\draw (280.42,181.8) node [anchor=north west][inner sep=0.75pt]  [rotate=-269.64] [align=left] {\ \ \ \ Convolutional};

\draw (217.42,158.8) node [anchor=north west][inner sep=0.75pt]  [rotate=-269.64] [align=left] {Max-pooling};

\draw (310.42,158.8) node [anchor=north west][inner sep=0.75pt]  [rotate=-269.64] [align=left] {Max-pooling};

\draw (246.42,144.99) node [anchor=north west][inner sep=0.75pt]  [rotate=-269.64] [align=left] {ReLU};

\draw (338.42,144.99) node [anchor=north west][inner sep=0.75pt]  [rotate=-269.64] [align=left] {ReLU};

\draw (125.42,160.79) node [anchor=north west][inner sep=0.75pt]  [rotate=-269.64] [align=left] {Training Data};

\draw (396.42,144.99) node [anchor=north west][inner sep=0.75pt]  [rotate=-269.64] [align=left] {ReLU};

\draw (457.42,145.99) node [anchor=north west][inner sep=0.75pt]  [rotate=-269.64] [align=left] {ReLU};

\draw (425.42,172.8) node [anchor=north west][inner sep=0.75pt]  [rotate=-269.64] [align=left] {Fully Connected};

\draw (365.42,172.8) node [anchor=north west][inner sep=0.75pt]  [rotate=-269.64] [align=left] {Fully Connected};

\end{tikzpicture}

}
	\caption{Local model without split learning}
	\label{fig:slfs}
\end{figure}

\subsection{Local model with vanilla split}
\label{subsec:lvsplit}
In this protocol, we split the local model shown in \autoref{fig:slfs} between the client and the server in such a way that the starting layers, Conv2D, max-pooling and the first two ReLU layers are executed on client-side while two FC layers and the last two ReLU layers are executed on the server-side. As can be seen in \autoref{fig:localsl}, the training input data is residing on the client-side while the final prediction of the model is on the server-side.  

\begin{figure}
	\centering
\tikzset{every picture/.style={line width=0.75pt}} %set default line width to 0.75pt        
\scalebox{0.65}{
\begin{tikzpicture}[x=0.75pt,y=0.75pt,yscale=-1,xscale=1]

\draw  [fill={rgb, 255:red, 255; green, 244; blue, 199 }  ,fill opacity=1 ] (69,131.62) -- (81.9,118.72) -- (112,118.72) -- (112,240.82) -- (99.1,253.72) -- (69,253.72) -- cycle ; \draw   (112,118.72) -- (99.1,131.62) -- (69,131.62) ; \draw   (99.1,131.62) -- (99.1,253.72) ;
%Shape: Cube [id:dp38623877719303756] 
\draw  [fill={rgb, 255:red, 200; green, 218; blue, 164 }  ,fill opacity=1 ] (99,141.32) -- (111.6,128.72) -- (141,128.72) -- (141,231.22) -- (128.4,243.82) -- (99,243.82) -- cycle ; \draw   (141,128.72) -- (128.4,141.32) -- (99,141.32) ; \draw   (128.4,141.32) -- (128.4,243.82) ;
%Shape: Cube [id:dp8708575104766393] 
\draw  [fill={rgb, 255:red, 218; green, 246; blue, 242 }  ,fill opacity=1 ] (127.2,142.52) -- (141,128.72) -- (173.2,128.72) -- (173.2,230.92) -- (159.4,244.72) -- (127.2,244.72) -- cycle ; \draw   (173.2,128.72) -- (159.4,142.52) -- (127.2,142.52) ; \draw   (159.4,142.52) -- (159.4,244.72) ;
%Shape: Cube [id:dp536378639681616] 
\draw  [fill={rgb, 255:red, 197; green, 181; blue, 175 }  ,fill opacity=1 ] (6,129.62) -- (18.9,116.72) -- (49,116.72) -- (49,239.1) -- (36.1,252) -- (6,252) -- cycle ; \draw   (49,116.72) -- (36.1,129.62) -- (6,129.62) ; \draw   (36.1,129.62) -- (36.1,252) ;
%Shape: Cube [id:dp9491748167308633] 
\draw  [fill={rgb, 255:red, 255; green, 244; blue, 199 }  ,fill opacity=1 ] (159.4,129.62) -- (172.3,116.72) -- (202.4,116.72) -- (202.4,237.82) -- (189.5,250.72) -- (159.4,250.72) -- cycle ; \draw   (202.4,116.72) -- (189.5,129.62) -- (159.4,129.62) ; \draw   (189.5,129.62) -- (189.5,250.72) ;
%Shape: Cube [id:dp27336604283361543] 
\draw  [fill={rgb, 255:red, 200; green, 218; blue, 164 }  ,fill opacity=1 ] (189,141.32) -- (201.6,128.72) -- (231,128.72) -- (231,231.12) -- (218.4,243.72) -- (189,243.72) -- cycle ; \draw   (231,128.72) -- (218.4,141.32) -- (189,141.32) ; \draw   (218.4,141.32) -- (218.4,243.72) ;
%Shape: Cube [id:dp35003917239516746] 
\draw  [fill={rgb, 255:red, 218; green, 246; blue, 242 }  ,fill opacity=1 ] (218.4,139.7) -- (229.38,128.72) -- (255,128.72) -- (255,232.74) -- (244.02,243.72) -- (218.4,243.72) -- cycle ; \draw   (255,128.72) -- (244.02,139.7) -- (218.4,139.7) ; \draw   (244.02,139.7) -- (244.02,243.72) ;
%Straight Lines [id:da562587812754598] 
\draw [line width=1.5]    (49,182.72) -- (65,182.72) ;
\draw [shift={(69,182.72)}, rotate = 180] [fill={rgb, 255:red, 0; green, 0; blue, 0 }  ][line width=0.08]  [draw opacity=0] (11.61,-5.58) -- (0,0) -- (11.61,5.58) -- cycle    ;
%Shape: Rectangle [id:dp5672808989440136] 
\draw  [dash pattern={on 0.84pt off 2.51pt}] (-1,110) -- (258,110) -- (258,264) -- (-1,264) -- cycle ;
%Shape: Cube [id:dp8238252994902237] 
\draw  [fill={rgb, 255:red, 195; green, 220; blue, 252 }  ,fill opacity=1 ] (318.4,128.62) -- (331.3,115.72) -- (361.4,115.72) -- (361.4,236.82) -- (348.5,249.72) -- (318.4,249.72) -- cycle ; \draw   (361.4,115.72) -- (348.5,128.62) -- (318.4,128.62) ; \draw   (348.5,128.62) -- (348.5,249.72) ;

%Shape: Cube [id:dp8729010820643909] 
\draw  [fill={rgb, 255:red, 246; green, 200; blue, 185 }  ,fill opacity=1 ] (348.6,143.62) -- (361.5,130.72) -- (391.6,130.72) -- (391.6,227.99) -- (378.7,240.89) -- (348.6,240.89) -- cycle ; \draw   (391.6,130.72) -- (378.7,143.62) -- (348.6,143.62) ; \draw   (378.7,143.62) -- (378.7,240.89) ;
%Shape: Cube [id:dp6784030540073792] 
\draw  [fill={rgb, 255:red, 195; green, 220; blue, 252 }  ,fill opacity=1 ] (378.4,129.62) -- (391.3,116.72) -- (421.4,116.72) -- (421.4,237.82) -- (408.5,250.72) -- (378.4,250.72) -- cycle ; \draw   (421.4,116.72) -- (408.5,129.62) -- (378.4,129.62) ; \draw   (408.5,129.62) -- (408.5,250.72) ;

%Shape: Cube [id:dp2336303612681001] 
\draw  [fill={rgb, 255:red, 246; green, 200; blue, 185 }  ,fill opacity=1 ] (408.6,145.62) -- (421.5,132.72) -- (451.6,132.72) -- (451.6,229.99) -- (438.7,242.89) -- (408.6,242.89) -- cycle ; \draw   (451.6,132.72) -- (438.7,145.62) -- (408.6,145.62) ; \draw   (438.7,145.62) -- (438.7,242.89) ;
%Shape: Rectangle [id:dp15787125517459333] 
\draw  [dash pattern={on 0.84pt off 2.51pt}] (308.3,109) -- (464.3,109) -- (464.3,263) -- (308.3,263) -- cycle ;
%Straight Lines [id:da8484933126162] 
\draw    (305.3,183.65) -- (255.3,183.65) ;
\draw [shift={(308.3,183.65)}, rotate = 180] [fill={rgb, 255:red, 0; green, 0; blue, 0 }  ][line width=0.08]  [draw opacity=0] (12.5,-6.01) -- (0,0) -- (12.5,6.01) -- cycle    ;

%client side layers
\draw (75.42,233.8) node [anchor=north west][inner sep=0.75pt]  [rotate=-269.64] [align=left] {Convolutional};

\draw (163.42,233.8) node [anchor=north west][inner sep=0.75pt]  [rotate=-269.64] [align=left] {Convolutional};

\draw (104.42,228.8) node [anchor=north west][inner sep=0.75pt]  [rotate=-269.64] [align=left] {Max-pooling};

\draw (197.42,227.8) node [anchor=north west][inner sep=0.75pt]  [rotate=-269.64] [align=left] {Max-pooling};

\draw (136.42,210.99) node [anchor=north west][inner sep=0.75pt]  [rotate=-269.64] [align=left] {ReLU};

\draw (222.42,210.99) node [anchor=north west][inner sep=0.75pt]  [rotate=-269.64] [align=left] {ReLU};

\draw (62,266.72) node [anchor=north west][inner sep=0.75pt]   [align=left] {\textbf{Client}};

\draw (9.42,231.79) node [anchor=north west][inner sep=0.75pt]  [rotate=-269.64] [align=left] {Training Data};

%Server side layers
\draw (356.42,210.99) node [anchor=north west][inner sep=0.75pt]  [rotate=-269.64] [align=left] {ReLU};

\draw (417.42,211.99) node [anchor=north west][inner sep=0.75pt]  [rotate=-269.64] [align=left] {ReLU};

\draw (331.19,265.72) node [anchor=north west][inner sep=0.75pt]   [align=left] {\textbf{Server}};

\draw (260,158) node [anchor=north west][inner sep=0.75pt]   [align=left] {$ATm$};

\draw (384.42,240.8) node [anchor=north west][inner sep=0.75pt]  [rotate=-269.64] [align=left] {Fully Connected};

\draw (325.42,240.8) node [anchor=north west][inner sep=0.75pt]  [rotate=-269.64] [align=left] {Fully Connected};

\end{tikzpicture}
}
	\caption{Local model with vanilla split learning}
	\label{fig:localsl}
\end{figure}
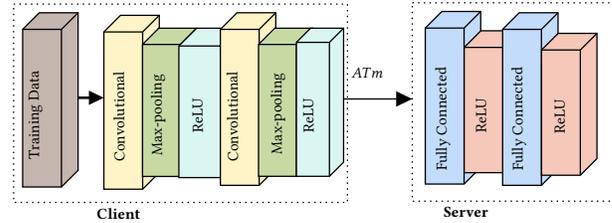

\subsection{Private local model}
\label{subsec:lfss}

In~\cite{ryffel2020ariann}, the authors proposed AriaNN, an efficient protocol for achieving private equality and comparison through FSS.
The model used in this study, as depicted in \autoref{fig:slfs}, consists of two Conv2D layers, two Max-pooling, four ReLUs and two FC layers. Among these layers, ReLU is supported as a direct application of the equality and comparison protocol, while the Conv2D layers are computed using beaver triples. More details, about the computations of these layers can be found in~\cite{ryffel2020ariann}. We reproduced the results for this variant of the protocol, achieving the best test accuracy of~97.29\%.

\subsection{Private vanilla SL}
\label{subsec:vsplit}

In our %proposed 
private vanilla SL protocol, the client handles the initial layers of the model, while the server takes care of the rest (see \autoref{fig:vanilla_sl}). The client, owning the data, executes their part of the model without splitting. On the server-side, as depicted in \autoref{fig:vanilla_sl}, consists of four layers: two ReLUs and two FC layers. For efficient online communication, we train the FC layers using beaver triples and the ReLU function using FSS. Our FSS implementation for ReLU layers ensures privacy using a comparison protocol. 

With FSS on the server-side, we now have two servers, each working on its private shared function and evaluating it on the public input $x_{pub}$. This protocol allows client-server collaboration in ML model training without exposing client's input data. Additionally, with some of the layers on the server-side, computational complexity is reduced, making the model more efficient. More details on the private vanilla SL will be discussed in \autoref{subsec: protocol vanillasl}.

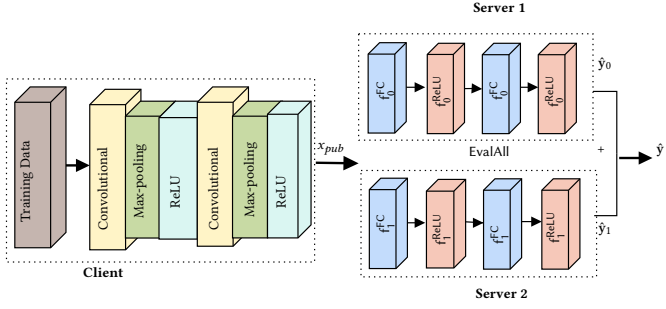
\begin{figure}
	\centering
	
	\tikzset{every picture/.style={line width=0.75pt}} %set default line width to 0.75pt        
	\scalebox{0.6}{
\tikzset{every picture/.style={line width=0.75pt}} %set default line width to 0.75pt        

\begin{tikzpicture}[x=0.75pt,y=0.75pt,yscale=-1,xscale=1]

\draw [line width=1.5]    (259,180.5) -- (291,180.94) ;
\draw [shift={(295,181)}, rotate = 180.8] [fill={rgb, 255:red, 0; green, 0; blue, 0 }  ][line width=0.08]  [draw opacity=0] (12.77,-6.13) -- (0,0) -- (12.77,6.13) -- cycle    ;
%Shape: Cube [id:dp7825681335956628] 
\draw  [fill={rgb, 255:red, 195; green, 220; blue, 252 }  ,fill opacity=1 ] (303,89.1) -- (312.6,79.5) -- (335,79.5) -- (335,145.9) -- (325.4,155.5) -- (303,155.5) -- cycle ; \draw   (335,79.5) -- (325.4,89.1) -- (303,89.1) ; \draw   (325.4,89.1) -- (325.4,155.5) ;
%Shape: Cube [id:dp12052028784844204] 
\draw  [fill={rgb, 255:red, 246; green, 200; blue, 185 }  ,fill opacity=1 ] (446,89.6) -- (455.6,80) -- (478,80) -- (478,146.4) -- (468.4,156) -- (446,156) -- cycle ; \draw   (478,80) -- (468.4,89.6) -- (446,89.6) ; \draw   (468.4,89.6) -- (468.4,156) ;
%Shape: Rectangle [id:dp5358150986207897] 
\draw  [dash pattern={on 0.84pt off 2.51pt}] (295,72) -- (492,72) -- (492,161.5) -- (295,161.5) -- cycle ;
%Shape: Cube [id:dp9752021665052911] 
\draw  [fill={rgb, 255:red, 195; green, 220; blue, 252 }  ,fill opacity=1 ] (399,89.6) -- (408.6,80) -- (431,80) -- (431,146.4) -- (421.4,156) -- (399,156) -- cycle ; \draw   (431,80) -- (421.4,89.6) -- (399,89.6) ; \draw   (421.4,89.6) -- (421.4,156) ;
%Shape: Cube [id:dp7398232142187168] 
\draw  [fill={rgb, 255:red, 195; green, 220; blue, 252 }  ,fill opacity=1 ] (304,203.1) -- (313.6,193.5) -- (336,193.5) -- (336,259.9) -- (326.4,269.5) -- (304,269.5) -- cycle ; \draw   (336,193.5) -- (326.4,203.1) -- (304,203.1) ; \draw   (326.4,203.1) -- (326.4,269.5) ;
%Shape: Cube [id:dp5714327045645562] 
\draw  [fill={rgb, 255:red, 246; green, 200; blue, 185 }  ,fill opacity=1 ] (449,203.6) -- (458.6,194) -- (481,194) -- (481,260.4) -- (471.4,270) -- (449,270) -- cycle ; \draw   (481,194) -- (471.4,203.6) -- (449,203.6) ; \draw   (471.4,203.6) -- (471.4,270) ;
%Shape: Rectangle [id:dp7525251258226833] 
\draw  [dash pattern={on 0.84pt off 2.51pt}] (297,187) -- (495,187) -- (495,276.5) -- (297,276.5) -- cycle ;
%Shape: Cube [id:dp22812306843286678] 
\draw  [fill={rgb, 255:red, 195; green, 220; blue, 252 }  ,fill opacity=1 ] (400,203.6) -- (409.6,194) -- (432,194) -- (432,260.4) -- (422.4,270) -- (400,270) -- cycle ; \draw   (432,194) -- (422.4,203.6) -- (400,203.6) ; \draw   (422.4,203.6) -- (422.4,270) ;
%Straight Lines [id:da7496274506383563] 
\draw    (335,117) -- (348,117) ;
\draw [shift={(351,117)}, rotate = 180] [fill={rgb, 255:red, 0; green, 0; blue, 0 }  ][line width=0.08]  [draw opacity=0] (8.93,-4.29) -- (0,0) -- (8.93,4.29) -- cycle    ;
%Straight Lines [id:da5246343752203286] 
\draw    (431,117) -- (443,117) ;
\draw [shift={(446,117)}, rotate = 180] [fill={rgb, 255:red, 0; green, 0; blue, 0 }  ][line width=0.08]  [draw opacity=0] (8.93,-4.29) -- (0,0) -- (8.93,4.29) -- cycle    ;
%Straight Lines [id:da10682521163034064] 
\draw    (336,231) -- (348,231) ;
\draw [shift={(351,231)}, rotate = 180] [fill={rgb, 255:red, 0; green, 0; blue, 0 }  ][line width=0.08]  [draw opacity=0] (8.93,-4.29) -- (0,0) -- (8.93,4.29) -- cycle    ;
%Straight Lines [id:da6984638920209232] 
\draw    (432,230) -- (446,230) ;
\draw [shift={(449,230)}, rotate = 180] [fill={rgb, 255:red, 0; green, 0; blue, 0 }  ][line width=0.08]  [draw opacity=0] (8.93,-4.29) -- (0,0) -- (8.93,4.29) -- cycle    ;
%Straight Lines [id:da9037829729643017] 
\draw [line width=1.5]    (514,176.5) -- (538,176.5) ;
\draw [shift={(542,176.5)}, rotate = 180] [fill={rgb, 255:red, 0; green, 0; blue, 0 }  ][line width=0.08]  [draw opacity=0] (11.61,-5.58) -- (0,0) -- (11.61,5.58) -- cycle    ;
%Shape: Right Angle [id:dp7860669210999071] 
\draw   (492,120.5) -- (512,120.5) -- (512,172.5) ;
%Shape: Right Angle [id:dp7461752522888249] 
\draw   (493,225) -- (512,225) -- (512,172.5) ;

%Shape: Cube [id:dp6308705772823255] 
\draw  [fill={rgb, 255:red, 255; green, 244; blue, 199 }  ,fill opacity=1 ] (69,131.62) -- (81.9,118.72) -- (112,118.72) -- (112,240.82) -- (99.1,253.72) -- (69,253.72) -- cycle ; \draw   (112,118.72) -- (99.1,131.62) -- (69,131.62) ; \draw   (99.1,131.62) -- (99.1,253.72) ;
%Shape: Cube [id:dp03532585046232328] 
\draw  [fill={rgb, 255:red, 200; green, 218; blue, 164 }  ,fill opacity=1 ] (99,141.32) -- (111.6,128.72) -- (141,128.72) -- (141,231.22) -- (128.4,243.82) -- (99,243.82) -- cycle ; \draw   (141,128.72) -- (128.4,141.32) -- (99,141.32) ; \draw   (128.4,141.32) -- (128.4,243.82) ;
%Shape: Cube [id:dp39065299563877653] 
\draw  [fill={rgb, 255:red, 218; green, 246; blue, 242 }  ,fill opacity=1 ] (127.2,142.52) -- (141,128.72) -- (173.2,128.72) -- (173.2,230.92) -- (159.4,244.72) -- (127.2,244.72) -- cycle ; \draw   (173.2,128.72) -- (159.4,142.52) -- (127.2,142.52) ; \draw   (159.4,142.52) -- (159.4,244.72) ;
%Shape: Cube [id:dp604764818339186] 
\draw  [fill={rgb, 255:red, 197; green, 181; blue, 175 }  ,fill opacity=1 ] (6,129.62) -- (18.9,116.72) -- (49,116.72) -- (49,239.1) -- (36.1,252) -- (6,252) -- cycle ; \draw   (49,116.72) -- (36.1,129.62) -- (6,129.62) ; \draw   (36.1,129.62) -- (36.1,252) ;
%Shape: Cube [id:dp8247691124859855] 
\draw  [fill={rgb, 255:red, 255; green, 244; blue, 199 }  ,fill opacity=1 ] (159.4,129.62) -- (172.3,116.72) -- (202.4,116.72) -- (202.4,237.82) -- (189.5,250.72) -- (159.4,250.72) -- cycle ; \draw   (202.4,116.72) -- (189.5,129.62) -- (159.4,129.62) ; \draw   (189.5,129.62) -- (189.5,250.72) ;
%Shape: Cube [id:dp7386769947810874] 
\draw  [fill={rgb, 255:red, 200; green, 218; blue, 164 }  ,fill opacity=1 ] (189,141.32) -- (201.6,128.72) -- (231,128.72) -- (231,231.12) -- (218.4,243.72) -- (189,243.72) -- cycle ; \draw   (231,128.72) -- (218.4,141.32) -- (189,141.32) ; \draw   (218.4,141.32) -- (218.4,243.72) ;
%Shape: Cube [id:dp6257190075387117] 
\draw  [fill={rgb, 255:red, 218; green, 246; blue, 242 }  ,fill opacity=1 ] (218.4,139.7) -- (229.38,128.72) -- (255,128.72) -- (255,232.74) -- (244.02,243.72) -- (218.4,243.72) -- cycle ; \draw   (255,128.72) -- (244.02,139.7) -- (218.4,139.7) ; \draw   (244.02,139.7) -- (244.02,243.72) ;
%Straight Lines [id:da5459238052532606] 
\draw [line width=1.5]    (49,182.72) -- (65,182.72) ;
\draw [shift={(69,182.72)}, rotate = 180] [fill={rgb, 255:red, 0; green, 0; blue, 0 }  ][line width=0.08]  [draw opacity=0] (11.61,-5.58) -- (0,0) -- (11.61,5.58) -- cycle    ;
%Shape: Rectangle [id:dp8918933631203068] 
\draw  [dash pattern={on 0.84pt off 2.51pt}] (-1,110) -- (258,110) -- (258,264) -- (-1,264) -- cycle ;
%Shape: Cube [id:dp7593728236622705] 
\draw  [fill={rgb, 255:red, 246; green, 200; blue, 185}  ,fill opacity=1 ] (352,204.39) -- (361.88,194.51) -- (384.93,194.51) -- (384.93,259.12) -- (375.05,269) -- (352,269) -- cycle ; \draw   (384.93,194.51) -- (375.05,204.39) -- (352,204.39) ; \draw   (375.05,204.39) -- (375.05,269) ;
%Straight Lines [id:da11629447261397008] 
\draw    (383,118.13) -- (392,118.04) -- (395.91,118.25) ;
\draw [shift={(398.9,118.42)}, rotate = 183.17] [fill={rgb, 255:red, 0; green, 0; blue, 0 }  ][line width=0.08]  [draw opacity=0] (8.93,-4.29) -- (0,0) -- (8.93,4.29) -- cycle    ;
%Straight Lines [id:da7322289221047917] 
\draw    (384.62,230.17) -- (400,231.02) ;
\draw [shift={(400,231.02)}, rotate = 180] [fill={rgb, 255:red, 0; green, 0; blue, 0 }  ][line width=0.08]  [draw opacity=0] (8.93,-4.29) -- (0,0) -- (8.93,4.29) -- cycle    ;
%Shape: Cube [id:dp389792891915983] 
\draw  [fill={rgb, 255:red, 246; green, 200; blue, 185}  ,fill opacity=1 ] (352.7,89.39) -- (362.09,80) -- (384,80) -- (384,146.61) -- (374.61,156) -- (352.7,156) -- cycle ; \draw   (384,80) -- (374.61,89.39) -- (352.7,89.39) ; \draw   (374.61,89.39) -- (374.61,156) ;

\draw (386,165.5) node [anchor=north west][inner sep=0.75pt]    {$\mathsf{EvalAll}$};

\draw (308.28,138.15) node [anchor=north west][inner sep=0.75pt]  [rotate=-270.59]  {\ \ \ $\mathsf{f^{FC}_{0}}$};

\draw (390,44) node [anchor=north west][inner sep=0.75pt]   [align=left] {\textbf{Server 1}};

\draw (392,285) node [anchor=north west][inner sep=0.75pt]   [align=left] {\textbf{Server 2}};

\draw (403.28,138.15) node [anchor=north west][inner sep=0.75pt]  [rotate=-270.59]  {\ \ \ $\mathsf{f^{FC}_{0}}$};

\draw (450.28,144.15) node [anchor=north west][inner sep=0.75pt]  [rotate=-270.59]  { \ \ \ $\mathsf{f^{ReLU}_{0}}$};

\draw (308.28,250.38) node [anchor=north west][inner sep=0.75pt]  [rotate=-270.59]  {\ \ \ $\mathsf{f^{FC}_{1}}$};

\draw (404.28,250.38) node [anchor=north west][inner sep=0.75pt]  [rotate=-270.59]  {$\mathsf{f^{FC}_{1}}$};

\draw (451.28,258.38) node [anchor=north west][inner sep=0.75pt]  [rotate=-270.59]  {\ \ \  $\mathsf{f^{ReLU}_{1}}$};

\draw (494.93,90.23) node [anchor=north west][inner sep=0.75pt]    {$\mathbf{\hat{y}}_{0}$};

\draw (495,229) node [anchor=north west][inner sep=0.75pt]    {$\mathbf{\hat{y}}_{1}$};

\draw (259,162) node [anchor=north west][inner sep=0.75pt]    {$x_{pub} $};

\draw (494,165.5) node [anchor=north west][inner sep=0.75pt]    {$+$};

\draw (543.93,162.23) node [anchor=north west][inner sep=0.75pt]    {$\mathbf{\hat{y}}$};

\draw (74.42,250.8) node [anchor=north west][inner sep=0.75pt]  [rotate=-269.64] [align=left] {\ \ \ \ Convolutional};

\draw (167.42,250.8) node [anchor=north west][inner sep=0.75pt]  [rotate=-269.64] [align=left] {\ \ \ \ Convolutional};

\draw (104.42,228.8) node [anchor=north west][inner sep=0.75pt]  [rotate=-269.64] [align=left] {Max-pooling};

\draw (197.42,227.8) node [anchor=north west][inner sep=0.75pt]  [rotate=-269.64] [align=left] {Max-pooling};

\draw (135.42,231.8) node [anchor=north west][inner sep=0.75pt]  [rotate=-269.64] [align=left] {\ \ \ \ ReLU};

\draw (224.42,231.8) node [anchor=north west][inner sep=0.75pt]  [rotate=-269.64] [align=left] {\ \ \ \ ReLU};

\draw (62,266.72) node [anchor=north west][inner sep=0.75pt]   [align=left] {\textbf{Client}};

\draw (9.42,231.79) node [anchor=north west][inner sep=0.75pt]  [rotate=-269.64] [align=left] {Training Data};

\draw (355.57,260.38) node [anchor=north west][inner sep=0.75pt]  [rotate=-270.59]  {\ \ \ $\mathsf{f^{ReLU}_{1}}$};

\draw (355.57,144.15) node [anchor=north west][inner sep=0.75pt]  [rotate=-270.59]  {\ \ \   $\mathsf{f^{ReLU}_{0}}$};

\end{tikzpicture}

}
	\caption{Private vanilla split learning}
	\label{fig:vanilla_sl}
\end{figure}

\subsection{Actors in the model}
\label{subsec:slActors}
%In this subsection
Now, we define the main actors participating in the protocol and detail their abilities:

\begin{itemize}[leftmargin=0.6cm]
    \item \textbf{Client}: The client provides the training data and is capable of computing the initial layer of the model in plaintext. The client uses their private data to compute the initial layers of the model and shares the intermediate $ATm$ and the remaining layers with the servers. In our %proposed 
    approach, the client has to share the labels with the server. %In the U-shaped SL setup, the client computes the final prediction of the model independently and does not share the truth labels to the server.
    \item \textbf{Servers}: Our protocol requires two servers to perform computations on the secret shares. %The servers compute on the ML model shares using the received $x_{pub}$. 
    Based on the received $x_{pub}$, each server make computations on the underlying ML model. Both servers can compute the final outputs and compare it to the truth labels received from the client. %, but in U-shaped SL, the servers do not receive the truth labels. 
    We assume that there is no collusion between the two servers and the servers are hosted independently to avoid reconstructing the original data from the shares. This assumption is realistic and consistent with other 2PC works in the area~\cite{mohassel2017secureml,ryffel2020ariann,agrawal2019quotient}.
\end{itemize}

\section{FSS based SL Protocol}
\label{sec:methodology}

This section provides a more detailed explanation of the algorithms for training the private vanilla SL model. %and the private U-shaped SL model.

\subsection{Private vanilla SL training protocol}
\label{subsec: protocol vanillasl}

We have used \autoref{alg:clientvsl} and \autoref{alg:servervsl} to train the private vanilla SL model reported in \autoref{subsec:vsplit}.  This model consists of the initialization, forward and backward propagation. The initialization phase only takes place once at the beginning of the procedure, whereas the other phases continue until the model iterates through all epochs. The initialization phase consists of socket %initialization 
and random weight initialization $\Phi$. %weighting. 
The client first establishes a socket connection and synchronize the hyperparameters $\eta, n, N, E$. $ATm^{(i)}$ and the gradient are initially set to zero $\emptyset$. These parameters must be synchronized on both sides to be trained in the same way. There are two phases in training the whole protocol, forward and backward propagation phase. During the forward propagation, the client forward propagates the input $x$ through $f_{\theta_{C}}$ %a number of Conv2D layers and also through Max pooling and ReLU layers until the $l^{th}$ layer 
to generate an activation map $ATm$. The private input which is $ATm$ is first masked using random mask $\alpha$ to construct the public input $x_{pub} = ATm+ \alpha$ before sending it to the server. Once $x_{pub}$ is constructed, the next step is to send $x_{pub}$ to the server. As can be seen in \autoref{fig:serverside_FSS}, we have two servers labelled as $P_{0}$ and $P_{1}$. Each server receives the public input $x_{pub}$ and initiates forward propagation by executing their respective portion of the model. On the server-side, we employ FSS, allowing us to compute the function keys. Specifically, these keys are generated for the ReLU layers. We employ the FSS $\mathsf{KeyGen}$ algorithm, to create these keys for each server. 
Additionally, to compute the FC layers %$ f(x_{pub})$, and $h(x_{pub})$ 
on the server-side, we used the beaver triples. % as outlined in \autoref{alg:beaver}.
As in private vanilla SL, the final output layer is executed on the server-side, hence both servers execute their respective share of the layer to produce the final output $\hat{y_{j}}$. 
Both the servers calculate the loss using the equation $J \leftarrow \mathcal{L} (\hat{y}_{j}, y_{j})$. During the loss calculation, the intermediate results and the final loss remain in the form of additive shares. Specifically, $y_{j}$ is the additive shares of the target value while $\hat{y}_{j}$ are the predicted output computed by each server. As the loss computation is performed by each servers, now, in order to get the final predicted output, the two servers must collaborate to add their secret shares. In this setting, the servers must know the actual target values as well as the predicted output of the model. 

\begin{figure*}
    \centering
    \tikzset{every picture/.style={line width=0.75pt}} %set default line width to 0.75pt        

\begin{tikzpicture}[x=1.0pt,y=1.0pt,yscale=-1,xscale=1]

\draw (156.52,60.31) node  {\includegraphics[width=12.38pt,height=19.04pt]{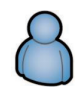}};
%Image [id:dp4349524323517665] 
\draw (157.31,214.11) node  {\includegraphics[width=12.38pt,height=19.04pt]{figures/user.png}};
%Straight Lines [id:da8535505111099704] 
\draw    (167.34,69.33) -- (222.06,69.32) ;
\draw [shift={(225.06,69.32)}, rotate = 179.99] [fill={rgb, 255:red, 0; green, 0; blue, 0 }  ][line width=0.08]  [draw opacity=0] (8.93,-4.29) -- (0,0) -- (8.93,4.29) -- cycle    ;
%Straight Lines [id:da8737159841172678] 
\draw    (171.68,218) -- (221.28,218.09) ;
\draw [shift={(224.28,218.1)}, rotate = 180.11] [fill={rgb, 255:red, 0; green, 0; blue, 0 }  ][line width=0.08]  [draw opacity=0] (8.93,-4.29) -- (0,0) -- (8.93,4.29) -- cycle    ;
%Image [id:dp884727302173338] 
\draw (275.19,64.51) node  {\includegraphics[width=73.5pt,height=45.27pt]{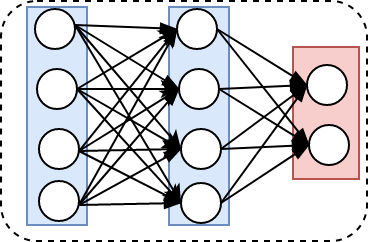}};
%Rounded Rect [id:dp7262350831836604] 
\draw  [fill={rgb, 255:red, 203; green, 221; blue, 243 }  ,fill opacity=1 ][dash pattern={on 0.84pt off 2.51pt}] (208.91,105.73) .. controls (208.91,103.23) and (210.94,101.2) .. (213.44,101.2) -- (311.48,101.2) .. controls (313.98,101.2) and (316.01,103.23) .. (316.01,105.73) -- (316.01,113.82) .. controls (316.01,116.32) and (313.98,118.35) .. (311.48,118.35) -- (213.44,118.35) .. controls (210.94,118.35) and (208.91,116.32) .. (208.91,113.82) -- cycle ;
%Shape: Circle [id:dp47686851571758015] 
\draw   (212.24,110.11) .. controls (212.24,106.49) and (215.18,103.56) .. (218.79,103.56) .. controls (222.41,103.56) and (225.34,106.49) .. (225.34,110.11) .. controls (225.34,113.72) and (222.41,116.66) .. (218.79,116.66) .. controls (215.18,116.66) and (212.24,113.72) .. (212.24,110.11) -- cycle ;
%Rounded Rect [id:dp5078092953098581] 
\draw  [fill={rgb, 255:red, 203; green, 221; blue, 243 }  ,fill opacity=1 ][dash pattern={on 0.84pt off 2.51pt}] (440.08,52.24) .. controls (440.08,49.71) and (442.13,47.66) .. (444.66,47.66) -- (539.09,47.66) .. controls (541.62,47.66) and (543.67,49.71) .. (543.67,52.24) -- (543.67,60.42) .. controls (543.67,62.95) and (541.62,65) .. (539.09,65) -- (444.66,65) .. controls (442.13,65) and (440.08,62.95) .. (440.08,60.42) -- cycle ;
%Shape: Ellipse [id:dp3422508349839063] 
\draw   (443.07,54.83) .. controls (443.07,51.03) and (445.64,47.95) .. (448.8,47.95) .. controls (451.96,47.95) and (454.53,51.03) .. (454.53,54.83) .. controls (454.53,58.64) and (451.96,61.72) .. (448.8,61.72) .. controls (445.64,61.72) and (443.07,58.64) .. (443.07,54.83) -- cycle ;
%Shape: Right Angle [id:dp6666247806900433] 
\draw   (407,36) -- (418.01,36) -- (418.01,59.33) ;
%Straight Lines [id:da4364283518276354] 
\draw    (418.01,59.33) -- (437.41,58.8) ;
\draw [shift={(440.41,58.72)}, rotate = 178.44] [fill={rgb, 255:red, 0; green, 0; blue, 0 }  ][line width=0.08]  [draw opacity=0] (8.93,-4.29) -- (0,0) -- (8.93,4.29) -- cycle    ;
%Straight Lines [id:da7024240170973469] 
\draw    (598.05,103.58) -- (349.42,103.58) ;
\draw [shift={(346.42,103.58)}, rotate = 360] [fill={rgb, 255:red, 0; green, 0; blue, 0 }  ][line width=0.08]  [draw opacity=0] (8.93,-4.29) -- (0,0) -- (8.93,4.29) -- cycle    ;
%Straight Lines [id:da04875716219606829] 
\draw    (598.05,103.58) -- (597.82,91.2) -- (597.73,86.36) ;

%Image [id:dp6456661936045222] 
\draw (279.02,216.4) node  {\includegraphics[width=79.83pt,height=52.5pt]{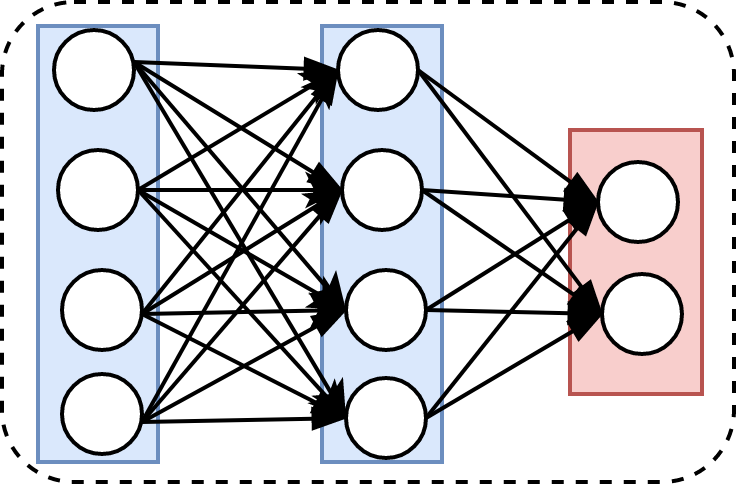}};
%Straight Lines [id:da868589847934592] 
\draw    (402.92,216.2) -- (588.92,216.2) ;
\draw [shift={(591.92,216.2)}, rotate = 180] [fill={rgb, 255:red, 0; green, 0; blue, 0 }  ][line width=0.08]  [draw opacity=0] (8.93,-4.29) -- (0,0) -- (8.93,4.29) -- cycle    ;
%Straight Lines [id:da7641759839014531] 
\draw    (600.42,182.83) -- (350.72,182.83) ;
\draw [shift={(347.72,182.83)}, rotate = 360] [fill={rgb, 255:red, 0; green, 0; blue, 0 }  ][line width=0.08]  [draw opacity=0] (8.93,-4.29) -- (0,0) -- (8.93,4.29) -- cycle    ;
%Straight Lines [id:da3190932552870749] 
\draw    (600.42,182.83) -- (600.19,196.68) -- (600.1,202.08) ;

%Rounded Rect [id:dp477500268035691] 
\draw  [fill={rgb, 255:red, 203; green, 221; blue, 243 }  ,fill opacity=1 ][dash pattern={on 0.84pt off 2.51pt}] (442.74,227.57) .. controls (442.74,224.73) and (445.04,222.43) .. (447.88,222.43) -- (541.45,222.43) .. controls (544.29,222.43) and (546.59,224.73) .. (546.59,227.57) -- (546.59,236.73) .. controls (546.59,239.56) and (544.29,241.86) .. (541.45,241.86) -- (447.88,241.86) .. controls (445.04,241.86) and (442.74,239.56) .. (442.74,236.73) -- cycle ;
%Shape: Ellipse [id:dp1384528490796405] 
\draw   (444.8,232.32) .. controls (444.8,228.59) and (447.41,225.58) .. (450.62,225.58) .. controls (453.84,225.58) and (456.44,228.59) .. (456.44,232.32) .. controls (456.44,236.04) and (453.84,239.06) .. (450.62,239.06) .. controls (447.41,239.06) and (444.8,236.04) .. (444.8,232.32) -- cycle ;
%Shape: Right Angle [id:dp434780471215224] 
\draw   (404.92,248.83) -- (414.92,248.83) -- (414.92,230.17) ;
%Straight Lines [id:da230155165863561] 
\draw    (414.92,230.17) -- (434.32,229.64) ;
\draw [shift={(437.32,229.56)}, rotate = 178.44] [fill={rgb, 255:red, 0; green, 0; blue, 0 }  ][line width=0.08]  [draw opacity=0] (8.93,-4.29) -- (0,0) -- (8.93,4.29) -- cycle    ;
%Rounded Rect [id:dp740772540478533] 
\draw  [fill={rgb, 255:red, 203; green, 221; blue, 243 }  ,fill opacity=1 ][dash pattern={on 0.84pt off 2.51pt}] (157.23,133.97) .. controls (157.23,130.95) and (159.68,128.5) .. (162.7,128.5) -- (235.87,128.5) .. controls (238.9,128.5) and (241.34,130.95) .. (241.34,133.97) -- (241.34,143.73) .. controls (241.34,146.75) and (238.9,149.2) .. (235.87,149.2) -- (162.7,149.2) .. controls (159.68,149.2) and (157.23,146.75) .. (157.23,143.73) -- cycle ;
%Rounded Rect [id:dp010945981004860439] 
\draw  [fill={rgb, 255:red, 255; green, 198; blue, 198 }  ,fill opacity=1 ][dash pattern={on 0.84pt off 2.51pt}] (157.59,159.68) .. controls (157.59,156.75) and (159.96,154.38) .. (162.89,154.38) -- (235.38,154.38) .. controls (238.31,154.38) and (240.68,156.75) .. (240.68,159.68) -- (240.68,169.13) .. controls (240.68,172.05) and (238.31,174.42) .. (235.38,174.42) -- (162.89,174.42) .. controls (159.96,174.42) and (157.59,172.05) .. (157.59,169.13) -- cycle ;
%Straight Lines [id:da08688829926337771] 
\draw    (409.92,71.2) -- (570.92,71.2) ;
\draw [shift={(573.92,71.2)}, rotate = 180] [fill={rgb, 255:red, 0; green, 0; blue, 0 }  ][line width=0.08]  [draw opacity=0] (8.93,-4.29) -- (0,0) -- (8.93,4.29) -- cycle    ;
%Straight Lines [id:da9420916636877349] 
\draw    (331.92,216.2) -- (362.92,216.2) ;
\draw [shift={(365.92,216.2)}, rotate = 180] [fill={rgb, 255:red, 0; green, 0; blue, 0 }  ][line width=0.08]  [draw opacity=0] (8.93,-4.29) -- (0,0) -- (8.93,4.29) -- cycle    ;
%Straight Lines [id:da21263323328862305] 
\draw    (324.92,71.2) -- (355.92,71.2) ;
\draw [shift={(358.92,71.2)}, rotate = 180] [fill={rgb, 255:red, 0; green, 0; blue, 0 }  ][line width=0.08]  [draw opacity=0] (8.93,-4.29) -- (0,0) -- (8.93,4.29) -- cycle    ;
%Rounded Rect [id:dp027203989563673536] 
\draw  [color={rgb, 255:red, 218; green, 240; blue, 195 }  ,draw opacity=1 ][fill={rgb, 255:red, 218; green, 240; blue, 195 }  ,fill opacity=1 ] (352.33,138.39) .. controls (352.33,136.17) and (354.14,134.37) .. (356.36,134.37) -- (380.97,134.37) .. controls (383.2,134.37) and (385,136.17) .. (385,138.39) -- (385,150.48) .. controls (385,152.7) and (383.2,154.51) .. (380.97,154.51) -- (356.36,154.51) .. controls (354.14,154.51) and (352.33,152.7) .. (352.33,150.48) -- cycle ;
%Straight Lines [id:da43165188632966534] 
\draw [line width=0.75]  [dash pattern={on 4.5pt off 4.5pt}]  (362.42,154.74) -- (362.56,182.76) ;
%Straight Lines [id:da48239769370164065] 
\draw  [dash pattern={on 4.5pt off 4.5pt}]  (362.22,106.97) -- (362.23,118.67) -- (362.55,133.4) ;
\draw [shift={(362.21,103.97)}, rotate = 89.91] [fill={rgb, 255:red, 0; green, 0; blue, 0 }  ][line width=0.08]  [draw opacity=0] (3.57,-1.72) -- (0,0) -- (3.57,1.72) -- cycle    ;
%Straight Lines [id:da5136666885516068] 
\draw  [dash pattern={on 4.5pt off 4.5pt}]  (376.67,178.73) -- (377.04,154.29) ;
\draw [shift={(376.63,181.72)}, rotate = 270.86] [fill={rgb, 255:red, 0; green, 0; blue, 0 }  ][line width=0.08]  [draw opacity=0] (3.57,-1.72) -- (0,0) -- (3.57,1.72) -- cycle    ;
%Straight Lines [id:da32693560193072313] 
\draw [line width=0.75]  [dash pattern={on 3.75pt off 3pt on 7.5pt off 1.5pt}]  (376.9,104.49) -- (377.04,133.51) ;
%Rounded Rect [id:dp822797349877506] 
\draw  [fill={rgb, 255:red, 203; green, 221; blue, 243 }  ,fill opacity=1 ][dash pattern={on 0.84pt off 2.51pt}] (485.91,113.13) .. controls (485.91,110.2) and (488.28,107.83) .. (491.2,107.83) -- (590.38,107.83) .. controls (593.31,107.83) and (595.68,110.2) .. (595.68,113.13) -- (595.68,122.57) .. controls (595.68,125.5) and (593.31,127.87) .. (590.38,127.87) -- (491.2,127.87) .. controls (488.28,127.87) and (485.91,125.5) .. (485.91,122.57) -- cycle ;
%Shape: Circle [id:dp24965482991081056] 
\draw   (487.24,116.44) .. controls (487.24,112.82) and (490.18,109.89) .. (493.79,109.89) .. controls (497.41,109.89) and (500.34,112.82) .. (500.34,116.44) .. controls (500.34,120.06) and (497.41,122.99) .. (493.79,122.99) .. controls (490.18,122.99) and (487.24,120.06) .. (487.24,116.44) -- cycle ;

% Text Node
\draw (178.21,52.67) node [anchor=north west][inner sep=0.75pt]  [font=\small]  {${\textstyle \langle x_{( pub)} \rangle }$};
% Text Node
\draw (150.27,76.4) node [anchor=north west][inner sep=0.75pt]    {$P_{0}$};
% Text Node
\draw (151.06,230.2) node [anchor=north west][inner sep=0.75pt]    {$P_{1}$};
% Text Node
\draw (178.54,200.01) node [anchor=north west][inner sep=0.75pt]  [font=\small]  {${\textstyle \langle x_{( pub)} \rangle }$};
% Text Node
\draw (292.54,76.34) node [anchor=north west][inner sep=0.75pt]    {${\displaystyle \langle w \rangle _{0}}$};
% Text Node
\draw (292.54,190.67) node [anchor=north west][inner sep=0.75pt]    {${\displaystyle \langle w \rangle _{1}}$};
% Text Node
\draw (370.54,26.67) node [anchor=north west][inner sep=0.75pt]    {${\displaystyle \langle y\rangle _{0}}$};
% Text Node
\draw (370.54,51.73) node [anchor=north west][inner sep=0.75pt]    {${\textstyle \langle \hat{y} \rangle _{0}}$};
% Text Node
\draw (367.88,239.4) node [anchor=north west][inner sep=0.75pt]    {${\displaystyle \langle y\rangle _{1}}$};
% Text Node
\draw (367.88,208.46) node [anchor=north west][inner sep=0.75pt]    {${\displaystyle \langle \hat{y} \rangle _{1}}$};
% Text Node
\draw (582.38,58.73) node [anchor=north west][inner sep=0.75pt]    {$J_{0}$};
% Text Node
%\draw (389.54,110.23) node [anchor=north west][inner sep=0.75pt]    {${\displaystyle \langle \hat{y} \rangle _{0} \ =\ ReLU(ATm)}$};
% Text Node
\draw (596.38,202.73) node [anchor=north west][inner sep=0.75pt]    {${\textstyle J_{1}}$};
% Text Node
\draw (379.88,73.73) node [anchor=north west][inner sep=0.75pt]    {${\textstyle w{_{0}^{n+1}} \ \leftarrow (w_{0}^{n} \ -\ \eta \nabla _{w}\mathcal{L})}$};
% Text Node
\draw (412.38,186.23) node [anchor=north west][inner sep=0.75pt]    {$w{_{1}^{n+1}} \ \leftarrow (w_{1}^{n} \ -\ \eta \nabla _{w}\mathcal{L})$};
% Text Node
\draw (389.38,134.73) node [anchor=north west][inner sep=0.75pt]    {${\textstyle \nabla _{w}\mathcal{L} \ =\ \left[\frac{\partial \mathcal{L}}{\partial w}\right]}$};
% Text Node
\draw (393.88,157.73) node [anchor=north west][inner sep=0.75pt]    {$J\ \leftarrow \mathcal{L}( \langle \hat{y} \rangle \ -\ \langle y\rangle )$};
% Text Node
\draw (247.04,130.73) node [anchor=north west][inner sep=0.75pt]    {$z\ =\ w*x+b$};
% Text Node
\draw (246.54,153.23) node [anchor=north west][inner sep=0.75pt]    {${\displaystyle ATm\ =\ ReLU( z)}$};
% Text Node
\draw (172.58,133.52) node [anchor=north west][inner sep=0.75pt]   [align=left] {{\small Linear Layer}};
% Text Node
\draw (172.76,160.01) node [anchor=north west][inner sep=0.75pt]   [align=left] {{\small ReLU Layer}};
% Text Node
\draw (230.17,105.5) node [anchor=north west][inner sep=0.75pt]   [align=left] {{\small Forward Prop.}};
% Text Node
\draw (215.17,106.5) node [anchor=north west][inner sep=0.75pt]   [align=left] {{\small 1}};
% Text Node
\draw (456.97,50.67) node [anchor=north west][inner sep=0.75pt]   [align=left] {{\small  Compute Cost}};
% Text Node
\draw (444.81,50.52) node [anchor=north west][inner sep=0.75pt]   [align=left] {{\small 2}};
% Text Node
\draw (359.59,141.01) node [anchor=north west][inner sep=0.75pt]   [align=left] {\textit{{\small 2PC}}};
% Text Node
\draw (504.17,112.83) node [anchor=north west][inner sep=0.75pt]   [align=left] {{\small Backward Prop.}};
% Text Node
\draw (491.17,112.83) node [anchor=north west][inner sep=0.75pt]   [align=left] {{\small 3}};
% Text Node
\draw (458.9,228) node [anchor=north west][inner sep=0.75pt]   [align=left] {{\small  Compute Cost}};
% Text Node
\draw (447.26,228.14) node [anchor=north west][inner sep=0.75pt]   [align=left] {{\small 2}};

\end{tikzpicture}

    \caption{Utilizing function secret sharing between two servers}
    \label{fig:serverside_FSS}
\end{figure*}

After completing the forward propagation,% the next step is the backward propagation phase. After calculating the loss, 
each server starts the backward propagation by first computing $\frac{\partial J_{j}}{\partial \hat{\mathbf{y_{j}}}}$ %and $\frac{\partial J_{j}}{\partial \hat{\mathbf{y_{j}}}}$ 
and continue the backward propagation algorithm until layer $l+1$ and update the weights and biases of the layers using the given equations.

\begin{align}
	\label{equ:serverUpdateWB}
	 \boldsymbol{w} =  \boldsymbol{w} - \eta\frac{\partial J}{\partial \boldsymbol{w}}, \quad  
\end{align}

\begin{align}
	\label{equ:serverUpdateWB}
	  \boldsymbol{b} = \boldsymbol{b} - \eta\frac{\partial J}{\partial \boldsymbol{b}} 
\end{align}

During backward propagation, each server also calculates: 

\begin{equation}
	\frac{\partial J_{j}}{\partial ATm^{(l)}} = \frac{\partial J}{\partial ATm^{(l+1)}} \frac{\partial ATm^{(l+1)}}{\partial ATm^{(l)}},
\end{equation}

and sends $\frac{\partial J}{\partial ATm^{(l)}}$ to the client. After receiving $\frac{\partial J}{\partial ATm^{(l)}}$, the client calculates the gradients of $J$ w.r.t the weights and biases of its part of layers using the chain-rule, which can generally be described %in general%
as:

\begin{align}
	\label{equ: gradients}
	\frac{\partial J_{j}}{\partial \boldsymbol{w}^{(i-1)}} &= \frac{\partial J_{j}}{\partial \boldsymbol{w}^{(i)}}\frac{\partial \boldsymbol{w}^{(i)}}{\partial \boldsymbol{w}^{(i-1)}} \\
	\frac{\partial J_{j}}{\partial \boldsymbol{b}^{(i-1)}} &= \frac{\partial J_{j}}{\partial \boldsymbol{b}^{(i)}}\frac{\partial \boldsymbol{b}^{(i)}}{\partial \boldsymbol{b}^{(i-1)}}    
\end{align}

Finally, after calculating the gradients $\frac{\partial J}{\partial \boldsymbol{w}^{(i)}}, \ \frac{\partial J}{\partial \boldsymbol{b}^{(i)}}$, the client updates $\boldsymbol{w}^{(i)}$ and $\boldsymbol{b}^{(i)}$. The forward and backward propagation continue until the model converges to learn a suitable set of parameters.

%{\SetAlgoNoLine%
% \begin{minipage}{0.5\textwidth}
\begin{algorithm}
%\KwResult{Write here the result }
 \textbf{\underline{Initialization}:}\\
 $soc\leftarrow$ socket initialized with port and address\;
 $soc.\mathsf{connect}$\\
 $\eta, n, N, E \leftarrow soc.\mathsf{synchronize}$\\
 $ \{\boldsymbol{w}^{( i)}, \boldsymbol{b}^{( i)}\}_{\forall i\in \{0..l\}} \ \leftarrow$ \ initialize using $\Phi $\\
 $%\{\mathbf{z}^{( i)}\}_{\forall i\in \{0..l\}} ,
 \{ATm^{( i)}\}_{\forall i\in \{0..l\}}\leftarrow \emptyset \ $\\
 $ %\left\{\frac{\partial J}{\partial \mathbf{z}^{( i)}}\right\}_{\forall i\in \{0..l\}} ,
 \left\{\frac{\partial J}{\partial ATm^{( i)}}\right\}_{\forall i\in \{0..l\}}\leftarrow \emptyset \ $\\
 \For{$\displaystyle e \ \in \ E $}
    {
 	\For{$\displaystyle \text{each} \ \text{batch}\ ( x,\ y) \ \text{generated\ from}\ D\ $}
        {
        $\displaystyle  \mathbf{\underline{Forward\ Propagation}:}$ \\
        \For{$i \leftarrow 1$ to $l$}
            {
            $\displaystyle ATm^{( i)} \ \leftarrow \ f_{\theta_{C}}\left( x^{( i)}\right)$\\}
            $\displaystyle  x_{pub}^{i} \ \leftarrow \ ATm^{( i)} + \alpha$\\
         	$\displaystyle   soc.\mathsf{send}\ ( x_{pub}^{i})$\\
         	$\displaystyle \mathbf{\underline{Backward\ Propagation}:}$\\
         	$\displaystyle  soc.\mathsf{receive}\ \left( \frac{\partial J}{\partial ATm_{j}^{(l+1)}} \right)$\\
           	$\displaystyle \text{Compute}\left\{\frac{\partial J}{\partial ATm^{l}}\right\}$\\
         	\For{$i\leftarrow l$ to $1$}
                {
                $\text{Compute}\ \left\{ \frac{\partial J}{\partial \boldsymbol{w}^{( i)}}, \ \frac{\partial J}{\partial \boldsymbol{b}^{( i)}} \right\}$\\
 	            $\displaystyle\  \text{Update}\ \boldsymbol{w}^{( i)},\ \boldsymbol{b}^{( i)}$
                }
 	    }
    }
 \caption{\textbf{Client-Side Vanilla SL}}
 \label{alg:clientvsl}
\end{algorithm}%}	
%{\SetAlgoNoLine%
\begin{algorithm}
%\KwResult{Write here the result }
 \textbf{\underline{Initialization}:}\\
 $soc\leftarrow$ socket initialized with port and address\;
 $soc.\mathsf{connect}$\\
 $\eta, n, N, E \leftarrow soc.\mathsf{synchronize}$\\
 \For{$\displaystyle e \ \in \ E $}
    {
 	\For{$\displaystyle i \leftarrow l+1 \ \mathbf{to} \ N \ $}
        {
        $\displaystyle \mathbf{\underline{Forward\ Propagation}:}$\\
        $\displaystyle soc.\mathsf{receive}\ (x_{pub}^{i}) \ \ $ \\
        \For{$\displaystyle j= 0, 1  $}
            {
            $\displaystyle   \hat{y_{j}} \ \leftarrow \ f_{\theta_{Pj}}\left(x_{pub}\right)$\\
            $\displaystyle  J_{j} \leftarrow \mathcal{L} (\hat{y_{j}}, y_{j})$\\}
            $\displaystyle \mathbf{\underline{Backward\ Propagation}:}$\\
            $\displaystyle \text{Compute}\left\{\frac{\partial J_{j}}{\partial \hat{y_{j}}}\right\}$\\
            \For{$i\leftarrow L$ to $l+1$}
                {
                $\text{Compute}\ \left\{ \frac{\partial J_{j}}{\partial \boldsymbol{w}_{j}^{( i)}}, \ \frac{\partial J_{j}}{\partial \boldsymbol{b}_{j}^{( i)}} \right\}$\\
                $\displaystyle\ \text{Update}\ \boldsymbol{w}_{j}^{( i)},\ \boldsymbol{b}_{j}^{( i)}$
                }
                $\displaystyle soc.\mathsf{send} \ \left( \frac{\partial J_{j}}{\partial ATm_{j}^{(l+1)}}\right)$\\
 		}
    }
 \caption{\textbf{Server-Side Vanilla SL}}
 \label{alg:servervsl}
\end{algorithm}

\section{Threat Model \& Security Analysis}
\label{sec:threat_model_and_sec}

%In this section, we define the threat model governing the security analysis of our %proposed 
%protocol. Then we will give a description of the FSHA, followed by the security analysis of our proposed private vanilla SL protocol.

In this section we first define the threat model we consider. Then, we proceed by proving the security of our protocol. 

\subsection{Threat model}
\label{subsec:threat_model}

In SL, collaborative learning of local models introduces potential risk due to interactions among the participants~\cite{li2022ressfl}. %Hence, we investigate the
%privacy leakage threats from two perspectives.
%We consider the security against semi-honest adversaries ($\mathcal{ADV}$) where the servers are honest but curious about the private data of each other. 
%In this work, we consider a semi-honest adversary $\mathcal{ADV}$ that can corrupt one of the two servers involved in our protocol. This means that corrupted parties follow protocol's specifications correctly while trying to gather as much information about others’ inputs or function shares. In our setting, secure protocols can be used to enable three parties (a client and two servers) to train a CNN model. The client who is also a data owner\footnote{It is worth mentioning that it is pointless to consider a malicious client since this is the data owner.}, initializes the process by executing the first few layers of the model. This generates an intermediate $ATm$, which is then secretly shared between two non-colluding servers. Subsequently, these two servers proceed to train the remaining part of the models using the client data by running the FSS protocol. 
In this study, we examine a scenario involving a semi-honest adversary, denoted as $\mathcal{ADV}$, capable of corrupting one of the two servers engaged in our protocol. In this context, the corrupted parties adhere to the protocol's specifications, all the while attempting to gather as much information as possible about the inputs and function shares of other participants. Within our specific framework, secure protocols are employed to facilitate collaborative CNN model training by three parties: a client and two servers. The client, who also serves as the data owner (note that considering a malicious client in this context is pointless), initiates the process by executing the initial layers of the model. This action produces an intermediary result, denoted as $ATm$, which is subsequently kept confidential through secret sharing between two independent servers. Following this, these two servers proceed to jointly train the remaining segments of the model using the client's data, achieved by employing the FSS protocol.
%In our threat model, $\mathcal{ADV}$ holds the potential to manifest as either an external entity with no direct involvement in the protocol's execution or as one of the servers actively participating in the protocol. $\mathcal{ADV}$ possesses distinct capabilities such as information gathering, data analysis and honest protocol execution that define their behavior and impact within the security analysis context. We define the capabilities of $\mathcal{ADV}$ as follows:
Apart from that, the behavior of $\mathcal{ADV}$  can be summarized as follows: 

\begin{itemize}[leftmargin=0.6cm]
    \item $\mathcal{ADV}$ does not possess information about the
client participating in the distributed training, except the necessary details required to execute the SL protocol. 
\item $\mathcal{ADV}$ is familiar with a \textit{shadow} dataset, which captures the same domain as the clients' training sets.

\item $\mathcal{ADV}$ has no knowledge of the architecture and weights of the whole model.
\item  $\mathcal{ADV}$ is unaware of the specific task on which the distributed model is trained.
\end{itemize}

% However, the $\mathcal{ADV}$ is familiar with a dataset denoted as $X_{pub}$, which captures the same domain as the clients' training sets, referred to as a ``shadow dataset.'' For instance, if the model is trained on face images, $X_{pub}$ also comprises face images. Nevertheless, there is no requirement for any overlap between private training sets and $X_{pub}$. %This assumption aligns with previous attacks on collaborative inference, making our threat model more realistic and less restrictive compared to other related works, where the adversary $\mathcal{ADV}$ is assumed to have direct access to leaked pairs of smashed data and private data.

%Additionally, the $\mathcal{ADV}$ attempts to perform the given attack: 

Finally, we extend the above threat model by defining two attacks available to $\mathcal{ADV}$ who wishes to launch an FSHA or perform data analysis through VI. 

\begin{myAttack}[Feature-Space Hijacking Attack] %\textit{Let $\mathcal{ADV}$ represents a malicious adversary acting as one of the servers. $\mathcal{ADV}$ successfully performs \textit{FSHA} if she manages to substitute the original learning task chosen by the client with a new function intentionally designed to mold the feature space of $f$. During this attack, $\mathcal{ADV}$ exploits its control over the training process to hijack $f$ and direct it towards a specifically crafted target feature space. Once $f$ maps into the target feature space, the attacker can recover the private training instances by locally inverting the known feature space.}
Let $\mathcal{ADV}$ be an adversary that has compromised at most one of the servers ($P_0$ or $P_1$). $\mathcal{ADV}$ successfully performs an \textit{FSHA} if she manages to substitute the original learning task chosen by the client with a new function intentionally designed to mold the feature space of $f$. 
\end{myAttack} 

\begin{myAttack}[Visual Invertibility Inference Attack] Let $\mathcal{ADV}$ be an adversary, that has compromised at most one of the servers ($P_0$ or $P_1$). $\mathcal{ADV}$ successfully performs a Visual Invertibility Inference Attack (VIIA) if she manages to infer sensitive information of the original input data using the acquired convolutions from the client. $\mathcal{ADV}$ attempts to identify the original images characteristics and truth labels using intermediate data received during the SL %split learning 
process.
\end{myAttack}

\subsection{Feature-space hijacking attack}
\label{subsec:fsha}
% In the following, we briefly explain the original FSHA in SL. We will focus on a scenario where the $\mathcal{ADV}$ takes control of one of the server and leaves the other server and client untouched. 
% Before going into the detail of this attack, let's start with some initial assumptions. These include that the server does not know anything about the client's network $f_{\theta_{C}}$ and is unaware of the specifics of the training task. 
% However, $\mathcal{ADV}$ is familiar with a \textit{shadow} dataset, which captures the same domain as the clients' training sets. For instance, if the model is trained on face images, shadow dataset also comprises face images. Nevertheless, there is no requirement for any overlap between private training sets and the shadow dataset.
%In this subsection, we describe the FSHA before going into the detailed description of the security of our proposed protocol.
Here, we describe the main concept behind an FSHA. 
In such an attack, the client owns its private data and is responsible for training the $f_{\theta_{C}}$. Instead of running its part of the SL protocol, the server runs the code provided by the $\mathcal{ADV}$. To carry out the attack $\mathcal{ADV}$ relies on three networks:
The first two $\tilde{f}$ and $\tilde{f}^{-1}$ are used to build an autoencoder which is trained with a shadow dataset. 
Specifically, $\tilde{f}$ maps the data space to a feature space whereas $\tilde{f}^{-1}$ act as a decoder that maps back the feature space into the original data space. The third one is Discriminator (D), the network that substitutes $f_{\theta_{P_{j}}}$, $j\in \{0, 1\}$, and is trained to distinguish between the output of $f_{\theta_{C}}$ and the output of network $\tilde{f}$. It is used to caste the gradients for $f_{\theta_{C}}$ during training and bring it in an insecure state. D indirectly guides $f_{\theta_{C}}$ to learn a mapping between the private data and the feature space defined from $\tilde{f}$. More specifically, D is used to train $f_{\theta_{C}}$ to produce an output that is as close as possible to the one of $\tilde{f}$. Ultimately, this is the network that substitutes server network in the protocol, and that defines the gradient which is sent to the client during the distributed training process. After certain number of iterations, the feature space learned by $f_{\theta_{C}}$ overlaps with one defined from $\tilde{f}$. Once this happens, $\tilde{f}^{-1}$ is used to map back this smashed data sent by the client to do the training and recover the approximation of the original private data. 

\subsection{Visual invertibility inference attack}
\label{subsec:visualinvertibility}

\begin{figure}[!h]
    \centering
    \includegraphics[width=0.45\textwidth]{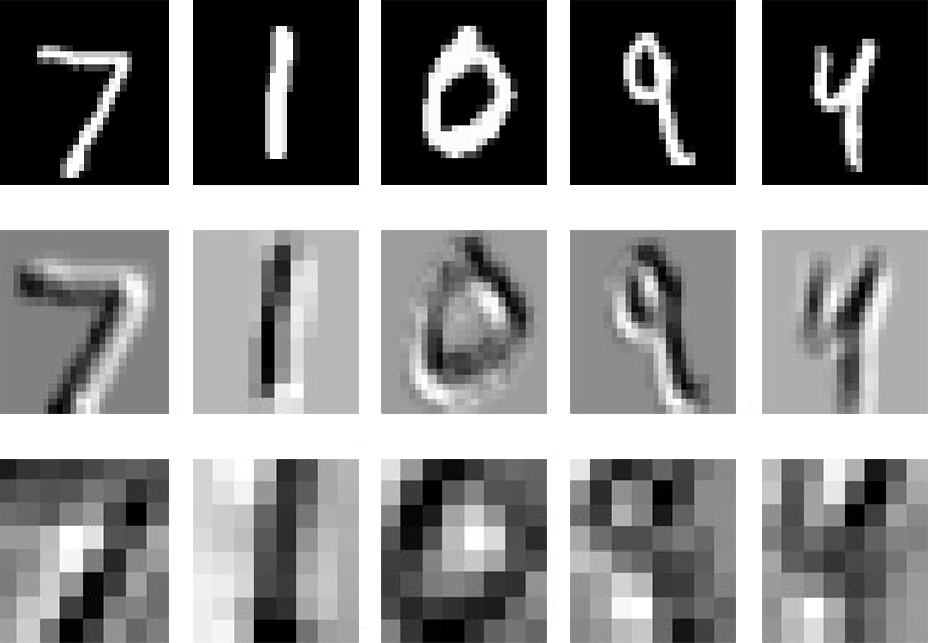}
    \caption{Visual invertibility inference attack: First row shows the original images (image size: $28\times28$) from the MNIST dataset. Second row shows the $ATm$ after the first Conv2D layer before the first max-pooling operation (image size: $24\times24$). Third row displays the $ATm$ after the second Conv2D layer before the second max-pooling operation (image size: $8\times8$).}
    \label{fig:MNIST_vis_inv}
\end{figure}

VI serves as a privacy assessment metric designed to measure the correlations between the split layer $ATm$ and original raw data samples. The underlying concept of VIIA involves the client initially training its part of the model to create the $ATm$. Afterwards, the client sends the $ATm$ to the server. Upon reception, if the server is malicious, it could exploit the received $ATm$ to extract private information by inferring the client's raw data and the corresponding truth label. In essence, VIIA assesses the risk for unintended data exposure during the model's collaborative learning process.  
In our SL setup, the client computes $f_{\theta_{C}}$ before outsourcing the remaining operations to a server. The client must share $ATm$ to the server for computing $f_{\theta_{P}}$. However, $ATm$ have patterns that are very similar to the raw image data and cause a privacy leakage. This similarity is visible in \autoref{fig:MNIST_vis_inv}, where the first row shows the raw image data and the second and third rows show intermediate outputs after the first and second Conv2D layers respectively. Even though the similarities are reduced after the second Conv2D layer, the $\mathcal{ADV}$ could extract some information and features of the private data, compromising the client's data privacy.

\subsection{Security analysis}
\label{sec: security-analysis}

%In this section, we prove the security of our proposed protocol assuming the threat model defined in \autoref{sec:threat_model}. To demonstrates the security of our protocol, we prove the following theorem: 
With our threat model established, we can now move forward to demonstrate the security of our protocol.

\begin{proposition}[FSHA Soundness]
\label{theorem:key-distribution}
   Let $\mathcal{ADV}$ be a semi-honest adversary that corrupts at most one of the two servers ($P_0$ or $P_1$) involved in the protocol. Then $\mathcal{ADV}$ cannot launch a successful Feature-Space Hijacking attack.
   % Let's consider $\mathcal{ADV}$ %who gets hold of 
   % %obtaining at most one of the two FSS keys, $\mathsf{f}_{j}$, where $j$ is either 0 or 1. 
   % be an adversary that has compromised at most one of the servers ($P_0$ or $P_1$). If the key distribution ($\mathsf{f}_j$, where $j$ is either 0 or 1) from one particular function, let's suppose $f$, is indistinguishable from the distribution of that key from some other function, say $f'$, then the $\mathcal{ADV}$ will not be able to distinguish which of the two functions it came from.
% \begin{equation}
%     \forall\ f,\ f'\ \in\ F,\ \{\mathsf{f}_{j}\ from\ f\}\ \approx\ \{\mathsf{f}^'_j\ from\ f'\}
% \end{equation}
\end{proposition}

\begin{proof}
    Let $\mathcal{ADV}$ be a semi-honest %who gets hold of 
    %obtaining at most one of the two FSS keys, $\mathsf{f}_{j}$, where $j$ is either 0 or 1. 
    adversary that has compromised $P_0$ or $P_1$. In addition to that, assume $\mathsf{f}_j$, $j \in [0,1]$ be the key distribution of an arbitrary function $f$. Then, if $\mathcal{ADV}$ can only distinguish the key distribution $\mathsf{f}_j$ from another $\mathsf{f}^'_j$ (for some other function $f'$), with negligible probability then our protocol is secure against the Feature-Space Hijacking attack. 
    
    %If the key distribution ($\mathsf{f}_j$, where $j$ is either 0 or 1) from one particular function $f$, is indistinguishable from the distribution of that key from some other function $f'$, then $\mathcal{ADV}$ cannot distinguish which of the two functions it came from.
 %\begin{equation}
 %    \forall\ f,\ f'\ \in\ F,\ \{\mathsf{f}_{j}\ from\ f\}\ \approx\ \{\mathsf{f}^'_j\ from\ f'\}
 %\end{equation}

To show that $\mathcal{ADV}$ can only distinguish between these values with negligible probability, we will use a proof by contradiction. Let's assume the opposite, that $\mathcal{ADV}$ can distinguish from which function ($f$ or $f'$) the key came from. This means that $\mathcal{ADV}$ has a non-negligible advantage to correctly determine whether the key came from $f$ or $f'$. 

    To test this assumption, let's consider a security game, where $\mathcal{ADV}$ has access to both $f$ and $f'$ and receives $\mathsf{f}_j$ from a challenger $\mathcal{CHAL}$:

    Following this game, the $\mathcal{ADV}$ would have a method to correctly guess which function the key $\mathsf{f}_j$ represents, through the bit $b'$, so that $b' == b$. However, the key $\mathsf{f}_j$ is generated in such a way, that it does not reveal any information about the original function $f$ \autoref{subsec:fss}. As such, the $\mathcal{ADV}$ has a~50\% chance of successfully guessing the original function, just like flipping a fair coin with a negligible advantage. This reaches a contradiction because the $\mathcal{ADV}$ cannot accurately guess the original function with a non-negligible advantage from the obtained function key $\mathsf{f}_j$.

Now, in order to show the soundness of our protocol against FSHA, we will focus on a scenario where the $\mathcal{ADV}$ takes control of one of the servers and leaves the other server and client untouched. 
% Before going into the detail of this attack, let's start with some initial assumptions. These include that the server does not know anything about the client's network $f_{\theta_{C}}$ and is unaware of the specifics of the training task. 
% However, $\mathcal{ADV}$ is familiar with a \textit{shadow} dataset, which captures the same domain as the clients' training sets. For instance, if the model is trained on face images, shadow dataset also comprises face images. Nevertheless, there is no requirement for any overlap between private training sets and the shadow dataset.
%In the context of the FSHA, where an 
The $\mathcal{ADV}$ tries to manipulate the feature space or the model's output and to accomplish this %do so, %having the model's parameters split in this way is a significant advantage. 
the $\mathcal{ADV}$ would need to understand the underlying feature space of $f_{\theta_P}$. %compromise both servers. %to fully understand the model and its feature space. 
However, through the use of FSS, $f_{\theta_{P}}$ is split between two servers which makes understanding the model and feature space improbable. Each server holds a share of the function but not the entire function itself. This makes it more challenging for an $\mathcal{ADV}$ to compromise the entire model by taking control of one server. Additionally, guessing the function share of the non-compromised server from the share of the compromised server is not possible (\autoref{theorem:key-distribution}). Hence, we prove that our proposed private vanilla SL protocol is secure against FSHA.

\begin{myframe}{}
\begin{multicols}{2}

    \underline{$\mathcal{ADV}$ sends}:
    
        $f$ and $f'$ to $\mathcal{CHAL}$
    \columnbreak 
    
    \underline{$\mathcal{ADV}$ guesses}:
    
    $b'$, where $b' == b$

\medskip
\end{multicols}
 \begin{center}

     \underline{$\mathcal{CHAL}$ computes}:

    $b \leftarrow \mathsf{Rand(\cdot)}$, where $b \in \{0,1\};$\\
    If $b==0$, $\mathsf{f}_j \leftarrow \mathsf{KeyGen}(1^\lambda,f);$\\
    Else, $\mathsf{f}_j \leftarrow \mathsf{KeyGen}(1^\lambda,f');$\\
    $\mathcal{CHAL}$ sends $\mathsf{f}_j$ to $\mathcal{ADV}$

 \end{center}
 
\end{myframe}

\end{proof}

% \begin{figure}[!h]
%     \centering
%     \includegraphics[width=0.45\textwidth]{figures/visual_inv_fix_var.jpg}
%     \caption{Visual invertibility: First row shows the original images (image size: $28\times28$) from the MNIST dataset. Second row shows the $ATm$ after the first Conv2D layer before the first max-pooling operation (image size: $24\times24$). Third row displays the $ATm$ after the second Conv2D layer before the second max-pooling operation (image size: $8\times8$).}
%     \label{fig:MNIST_vis_inv}
% \end{figure}

\begin{proposition}[VIIA Soundness]
\label{theorem:vi-soundness}
   Let $\mathcal{ADV}$ be a semi-honest adversary that corrupts at most one of the two servers ($P_0$ or $P_1$) involved in the protocol. Then $\mathcal{ADV}$ cannot successfully infer the raw input data through VIIA.
   % Let's consider $\mathcal{ADV}$ %who gets hold of 
   % %obtaining at most one of the two FSS keys, $\mathsf{f}_{j}$, where $j$ is either 0 or 1. 
   % be an adversary that has compromised at most one of the servers ($P_0$ or $P_1$). If the key distribution ($\mathsf{f}_j$, where $j$ is either 0 or 1) from one particular function, let's suppose $f$, is indistinguishable from the distribution of that key from some other function, say $f'$, then the $\mathcal{ADV}$ will not be able to distinguish which of the two functions it came from.
% \begin{equation}
%     \forall\ f,\ f'\ \in\ F,\ \{\mathsf{f}_{j}\ from\ f\}\ \approx\ \{\mathsf{f}^'_j\ from\ f'\}
% \end{equation}
\end{proposition}
\begin{proof}
Let $\mathcal{ADV}$ be a semi-honest adversary that has compromised one of the two servers $P_j$, $j \in [0,1]$. 
In addition, assume a noisy intermediate map $x_{pub} = ATm + \alpha$, where $ATm$ is the original intermediate activation map and $\alpha$ is a set amount of noise. Then, if $\mathcal{ADV}$ cannot acquire $ATm$ from $x_{pub}$, then our protocol is secure against raw data inference through Visual Invertibility Inference Attack.

The use of FSS addresses the privacy leakage caused by VIIA as the $ATm$ %is split into separate shares $ATm_j$, 
is masked with a specific amount of noise ($x_{pub} = ATm + \alpha$) for both servers $P_j$. Following this a single corrupted server $P_j$ would not be able to discern any additional information from the obtained $x_{pub}$ without guessing the function share held by the other, non-corrupted server as proven in Proposition~\autoref{theorem:key-distribution}. Additionally, if the $\mathcal{ADV}$ is unable to obtain $\alpha$, then it is improbable that $ATm$ could be reconstructed through $x_{pub}$, when an adequate amount of noise is added. To this extent, we also showcase a visualization of the masked $ATm$ from both servers in \autoref{fig:vis_inv_fss_conv1} and \autoref{fig:vis_inv_fss_conv2} to show that $x_{pub}$ becomes unrecognisable to the original convolutions $ATm$.

%Additionally, we show our protocols protects against data analysis through VI. Because of this we visualize the split layer output as an initial assessment to observe the possibility of reconstructing the raw image data. For this, we use VI which allow an $\mathcal{ADV}$ to gain insight into a client private input data using the $ATm$ after the Conv2D layers (see~\autoref{fig:MNIST_vis_inv}). %In our SL setup, the client computes $f_{\theta_{C}}$ before outsourcing the remaining operations to a server. The client must share $ATm$ to the server for computing $f_{\theta_{P}}$. However, $ATm$ have patterns that are very similar to the raw image data and cause a privacy leakage. This similarity is visible in \autoref{fig:MNIST_vis_inv}, where the first row shows the raw image data and the second and third rows show intermediate outputs after the first and second Conv2D layers respectively. Even though the similarities are reduced after the second Conv2D layer, the $\mathcal{ADV}$ could extract some information and features of the private data, compromising the client's data privacy. 

\begin{figure}[!h]
    \centering
    \subfloat[First Convolution]{\includegraphics[width=0.23\textwidth]{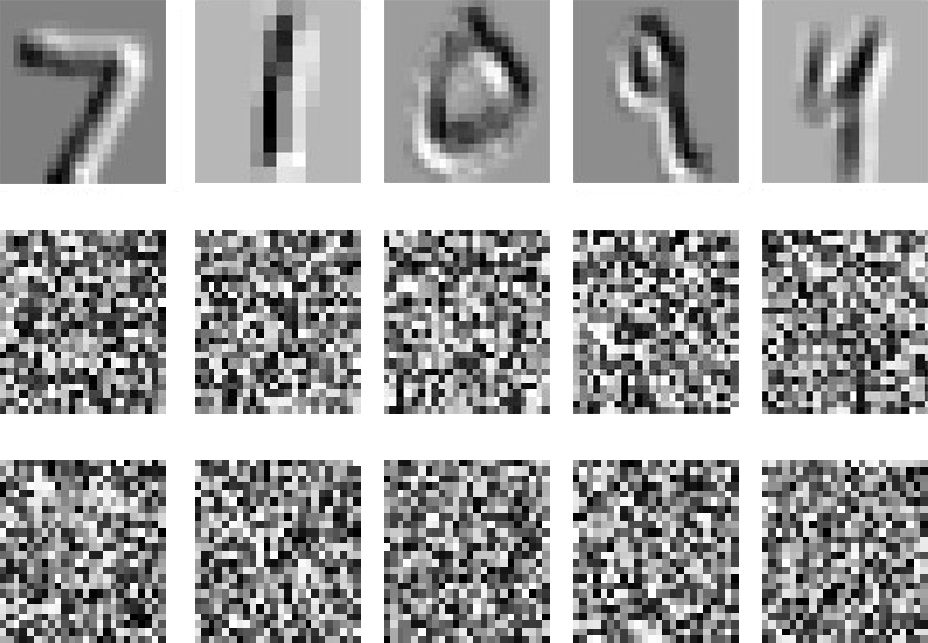}\label{fig:vis_inv_fss_conv1}}
    %\\
    \rulesep
    \subfloat[Second Convolution]{\includegraphics[width=0.23\textwidth]{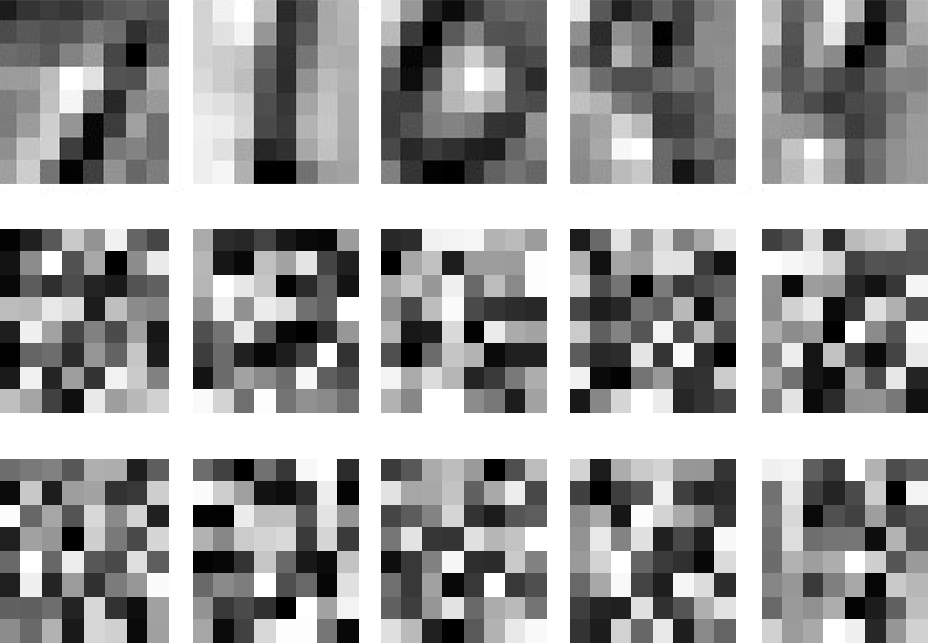}\label{fig:vis_inv_fss_conv2}}
    \caption{Visual invertibility inference attack with FSS applied: First row shows the original $ATm$ after the first Conv2D (before the max-pooling operations) from the MNIST dataset. Second row shows what $P_{0}$ receives when the FSS protocol is applied. Third row what $P_{1}$ receives when the FSS protocol is applied. (a) Image Size: $24\times24$; (b) Image Size: $8\times8$}
    \label{fig:vis_inv_fss}
\end{figure}

\end{proof}

\section{Performance Analysis}
\label{sec:perfanal}
This section covers our experimental setup and the results of our private vanilla SL protocol.

\subsection{Dataset}
\label{subsec:datasets}

%\paragraph*{\textbf{MNIST}} 
\noindent \textit{\textbf{MNIST.}} \enskip
To test the applicability of our approach we used the MNIST dataset~\cite{lecun_cortes_j.c.burges}. It consists of a total of~70,000 images of which 60,000 are in the training set and the remaining~10,000 are in the test set. All of the images are~28x28 pixels and display a black and white handwritten number ranging from~0 to~9. %Examples of the dataset are displayed in \autoref{fig:MNISTdata}. 
%The dataset is commonly used for benchmarking ML model efficiency and accuracy and has been used in prior PPML works such as AriaNN~\cite{ryffel2020ariann}. For model training, we used the entire training set of~60,000 images each epoch (dropping the final batch if it does not create a full batch) and test the accuracy after each $E$ using the remaining~10,000 test images. After each $E$ both the training and testing sets are randomly permuted to avoid pattern memorization.
% \begin{figure}[!hb]
%     \centering
%     \includegraphics[width=0.3\textwidth]{figures/mnist-3.0.1.png}
%     \caption{Example images from the MNIST dataset}
%     \label{fig:MNISTdata}
% \end{figure}

\subsection{Experimental setup}
\label{subsec:experiment}

We test and compare our %private vanilla and U-shaped 
SL protocol with %the implementation proposed in 
AriaNN~\cite{ryffel2020ariann}. For the experiments, we %make use of 
used a machine running Ubuntu~20.04 LTS, processor~12th generation Intel Core i7-12700, 32 GB RAM mesa Intel graphics (ADL-S GT1). %as described in \autoref{tab:bench}:
% \begin{table}[!ht]
%     \centering
%     \begin{tabular}{ |c|c| } 
%      \hline
%      OS & Ubuntu 20.04 LTS \\ 
%      \hline
%      CPU & 12th Gen Intel Core i7-12700\\ 
%      \hline
%      RAM & 32GB \\ 
%      \hline
%     \end{tabular}
%     \caption{Machine benchmark details}
%     \label{tab:bench}
% \end{table}
We implemented our code using Python~3.7 and it is made %publicly
available on %GitHub\footnote{\url{https://github.com/UnoriginalOrigi/SplitFSS}}. 
GitHub\footnote{\url{https://github.com/UnoriginalOrigi/SplitFSS}}.
FSS implemented using the PySyft framework\footnote{\url{https://github.com/OpenMined/PySyft/tree/ryffel/0.2.x-fix-training}} and a modified version of AriaNN\footnote{\url{https://github.com/LaRiffle/ariann}}. 
%We made use of the PySyft framework\footnote{\url{https://github.com/OpenMined/PySyft/tree/ryffel/0.2.x-fix-training}} for implementing FSS and based our work on a modified version of AriaNN\footnote{\url{https://github.com/LaRiffle/ariann}}. 
To get a more consistent and reproducible result, we ran our experiments a total of~10 times for each experiment, totaling~60 executions of the code with different settings. However, in our study, we are reporting the top three results for each model (refer to \autoref{table:performance}). In the experiments we measured the communication costs, training times, and accuracy of the model after each epoch and report the total costs after full training. The communication costs are calculated by measuring the byte length of serialized objects sent through packets. Time measurements are taken by making use of the in-built \textit{time} module from Python. The loss is calculated using mean squared error.% A major limitation of MSE is that a single outlier can strongly deteriorate the model accuracy and through our test some of the experiments dropped to as low as~75\% accuracy with the same parameters. We leave testing other loss functions on the models as a future work. 

%~(\autoref{eq:MSE}), where $N$ -- batch size, $y_i$ -- output prediction and $\hat{y_i}$ -- target prediction.
% \begin{equation}
%     \mathsf{MSE} = \frac{1}{N}\sum (y_i - \hat{y_i})^2
% \label{eq:MSE}
% \end{equation}

In terms of hyperparameters, we train all networks with $E = 10$  epochs, $\eta=0.002$
learning rate, $p = 0.9$ momentum, a batch size of $n = 128$, and the total training data samples $m = 59904$ (as the incomplete final batch is dropped).
% \begin{itemize}
%     \item Batch Size -- 128
%     \item Epochs -- 10
%     \item Momentum -- 0.9
%     \item Learning rate -- 0.02
% \end{itemize}
These parameters were chosen as they showed best results during preliminary testing and are equivalent to the parameters set for the AriaNN implementation~\cite{ryffel2020ariann}. This lets us accurately compare our results with AriaNN without the parameters making a large impact on the comparison.

\subsection{Evaluation}
\label{subsec:evaluation}
%This subsection documents the results of our experiments. The subsection is divided into~3 parts, each covering different architectures defined in \autoref{sec:architecture}. 
This section provides an overview of the outcomes derived from our experiments. It is organized into three distinct parts, with each part dedicated to a specific architectural framework as defined in \autoref{sec:architecture}.
%Firstly, \autoref{subsubsec:localmodel_eval} covers the local models with and without SL. These models are run in plaintext without making use of FSS to create a baseline to compare further results and also measure the effects of FSS on the performance of the models. Secondly, \autoref{subsubsec:FSSwithoutsplit_eval} provides the reproduced results of AriaNN~\cite{ryffel2020ariann}, which we use to compare to our proposed models. Thirdly, \autoref{subsubsec:FSSsplit_eval} shows the results of the private vanilla SL model and its comparison with the results of AriaNN. 
%Lastly, in \autoref{subsubsec: FSSUsplit eval} we perform experiments on the private U-shaped SL model and discuss the benefits of making use of this architecture. 
First, the local models with and without SL are covered. These models run in plaintext without making use of FSS to create a baseline to compare further results and also measure the effects of FSS on the performance of the models. Second, the reproduced results of AriaNN~\cite{ryffel2020ariann}, which we use to compare to our proposed models, are provided. Finally, the results of the private vanilla SL model and its comparison with the results of AriaNN are presented. 
%Lastly, in \autoref{subsubsec: FSSUsplit eval} we perform experiments on the private U-shaped SL model and discuss the benefits of making use of this architecture. 

\medskip

%\subsubsection{Plaintext models}
\noindent \textit{Plaintext models.} \enskip
\label{subsubsec:localmodel_eval}
%As previously mentioned, training the models on plaintext data does not take into account FSS. These models include a local model and a vanilla SL model. %, and a U-shaped SL model.
%We have successfully implemented and reproduced the results for these models. Following the training of our local model for~10 epochs, we observed the highest test accuracy to be~99.36\%, with the second and third highest test accuracies at~99.33\% and~99.28\%, respectively. Further, our experiments, show that training the vanilla and U-shaped SL model on plaintext data produce similar accuracy results when compared to training the local model. As given in \autoref{table:performance}, the three best results for vanilla SL model are~99.29\%,~99.28\% and 99.21\%. %while for U-shaped SL model are~99.34\%,~99.33\%, and~99.23\%. 
%These accuracy values closely align, suggesting that the SL model can be applied to CNN models on MNIST dataset without significant degradation in classification accuracy.
As previously mentioned, it's important to note that training the models on plaintext data does not account for FSS. These models consist of both a local model and a standard SL model. We have successfully implemented and replicated the results for these models. After training our local model for approximately 10 epochs, we achieved the highest test accuracy at 99.36\%, with the second and third highest test accuracies at 99.33\% and 99.28\%, respectively.

Furthermore, our experiments indicate that training the vanilla SL model %and the U-shaped SL model 
on plaintext data yield similar accuracy results when compared to training the local model. As illustrated in \autoref{table:performance}, the top three results for the vanilla SL model are 99.29\%, 99.28\%, and 99.21\%. These accuracy values closely align, suggesting that the SL model can be effectively applied to CNN models on the MNIST dataset without experiencing a significant degradation in classification accuracy.

We also consider the training time and communication cost of the vanilla % and U-shaped 
SL model and compare it to the local model. The total training time %for all the three models -- local model, vanilla SL and U-shaped SL model -- 
for~10 epochs is nearly identical, ranging between~1 and~2 minutes. Interestingly, even though the vanilla %and U-shaped SL 
model involve communication between the client and server to exchange $ATm$ and gradients, their training times remain similar to that of the local model. This similarity arises because the local model's training time encompasses the duration during which the client %(data owner)
transmits images to the server%(model owner)
. Consequently, both models have nearly the same training times. 

Regarding the communication cost, we are considering costs on both the client-side and server-side, as well as the total costs during training and testing. In \autoref{table:performance} we show the variation between the client-side and server-side communication costs between the local and SL model. We observe that the SL approach has slightly higher server-side communication costs, than the local model, but client-side communications are about $3\times$ larger on the local model. This difference arises from the size of the data being shared as in the local model, the client must send $28\times28$ images, while in the SL model, the client send the intermediate $ATm$, which is only $8\times8$.
%In the \autoref{table:performance} presented, it is noticeable that the U-shaped SL model has higher total communication costs compared to the local and vanilla SL models. This difference arises because, in the U-shaped SL model, there is more back-and-forth communication between the server and client involving $ATm$ and gradients, resulting in relatively higher server-side communication costs when compared to both the local and vanilla SL models. However, the client-side communication costs show only minor variations.  
Even though the vanilla %and U-shaped 
SL model provides similar accuracy and has the same training and communication cost when compared to the locally computed model, it is essential to note that the vanilla SL model introduce a privacy leakage, which we mitigate using FSS.%. 

\medskip

%\subsubsection{FSS without SL}
\noindent \textit{FSS without SL.} \enskip
\label{subsubsec:FSSwithoutsplit_eval}
For our benchmarking we used AriaNN~\cite{ryffel2020ariann} -- a promising solution for doing PPML on sensitive data. The authors have shown that AriaNN provides competitive performance compared to existing work, and have demonstrated its practicality by running private training on MNIST using models like LeNet. Given the promising outcomes achieved by AriaNN in comparison to other works, we adopted it as the base paper for our study.

In our experimentation, we replicated the results for AriaNN, which we refer to as the private local model. To thoroughly assess and compare the impact of FSS on performance, we initialized two models: one before applying FSS (referred to as the ``public local model'') and one after applying FSS, both with the same set of initial weights. Our analysis, clearly indicates that the accuracy obtained before and after applying FSS remains nearly identical. In simpler terms, the implementation of FSS exhibits no noticeable effect on the performance of the models.
As can be seen in \autoref{table:performance}, the private local model achieves an impressive test accuracy of approximately~97.29\%, closely followed by accuracies of~97.04\% and~97.61\%. This showcases that the private local model attains a commendable level of accuracy, lacking only by a modest~2\% to~3\% when compared to the public local models.

However, when we examine the training time for private local model over~10 epochs, a significant contrast emerges. Private local model necessitates approximately~1288 minutes, which is roughly~900$\times$ more time-consuming than the training of the public local models. Additionally, the training communication cost incurred by private local model is~2$\times$ higher, and the testing communication cost is~4$\times$ greater compared to the public local models. The communication costs reported in \autoref{table:performance} do not take into account the preprocessing communication costs for transmitting the FSS keys and securely computing each function. The preprocessing costs for the private local model is~10.572851 MB per batch.

In conclusion, private local model demonstrates its potential as a promising solution for maintaining data privacy while achieving commendable accuracy in ML. Nonetheless, it is important to remember that using private local model comes with significant drawbacks in terms of the time it takes to train models and the extra communication needed. This should be carefully weighed when deciding if it is the right choice when selecting a privacy-preserving framework for specific use cases.% As previously discussed, a limitation of the private local model is that, at the end of the protocol, both servers need to communicate to obtain the final prediction output, potentially revealing client data privacy. 

\begin{table*}
    \centering
    \caption{Performance comparison: training and testing results for various privacy-preserving models}
    \label{table:performance}
\begin{tabular}{|ll|llll|ll|}
\hline
\rowcolor{gray}
\multicolumn{2}{|c|}{\color{white}\textbf{PETs}} &
  \multicolumn{4}{c|}{\color{white}\textbf{Training Statistics}} &
  \multicolumn{2}{c|}{\color{white}\textbf{Testing Statistics}} \\ \hline
  \rowcolor{antiquewhite}
\multicolumn{1}{|l|}{FSS} &
  SL &
  \multicolumn{1}{l|}{\begin{tabular}[c]{@{}l@{}}Client \\ Communication  \\(MB)\end{tabular}} &
  \multicolumn{1}{l|}{\begin{tabular}[c]{@{}l@{}}Server \\ Communication \\(MB)\end{tabular}} &
  \multicolumn{1}{l|}{\begin{tabular}[c]{@{}l@{}}Training\\ Time  (minutes)\end{tabular}} &
  \begin{tabular}[c]{@{}l@{}}Training \\ Communication \\ (MB)\end{tabular} &
  \multicolumn{1}{l|}{\begin{tabular}[c]{@{}l@{}}Testing\\ Communication \\(MB)\end{tabular}} &
  \begin{tabular}[c]{@{}l@{}}Testing \\ Accuracy (\%)\end{tabular} \\ \hline
\multicolumn{1}{|l|}{\multirow{6}{*}{Public}} &
  \multirow{3}{*}{Local} &
  \multicolumn{1}{l|}{1931.855552} &
  \multicolumn{1}{l|}{1.656024} &
  % \multicolumn{1}{l|}{86.41157245635986} &
  \multicolumn{1}{l|}{1.4402} &
  1933.511576 &
  \multicolumn{1}{l|}{313.394140} &
  99.36 \\ \cline{3-8} 
\multicolumn{1}{|l|}{} &
   &
  \multicolumn{1}{l|}{1931.853762} &
  \multicolumn{1}{l|}{1.656288} &
  % \multicolumn{1}{l|}{85.84440231323242} &
  \multicolumn{1}{l|}{1.4307} &
  1933.510050 &
  \multicolumn{1}{l|}{313.394040} &
  99.33 \\ \cline{3-8} 
\multicolumn{1}{|l|}{} &
   &
  \multicolumn{1}{l|}{1931.854346} &
  \multicolumn{1}{l|}{1.658358} &
  % \multicolumn{1}{l|}{85.59089756011963} &
  \multicolumn{1}{l|}{1.4265} &
  1933.512704&
  \multicolumn{1}{l|}{313.394140} &
  99.28 \\ \cline{2-8} 
\multicolumn{1}{|l|}{} &
  \multirow{3}{*}{Vanilla} &
  \multicolumn{1}{l|}{641.172596} &
  \multicolumn{1}{l|}{25.745798} &
  % \multicolumn{1}{l|}{85.1001033782959} &
 \multicolumn{1}{l|}{1.4183} &
  666.702220 &
  \multicolumn{1}{l|}{102.522446} &
  99.29 \\ \cline{3-8}  
\multicolumn{1}{|l|}{} &
   &
  \multicolumn{1}{l|}{641.170668} &
  \multicolumn{1}{l|}{25.745258} &
  % \multicolumn{1}{l|}{85.94416046142578} &
  \multicolumn{1}{l|}{1.4324} &
  666.699764 &
  \multicolumn{1}{l|}{102.522176} &
  99.28 \\ \cline{3-8} 
\multicolumn{1}{|l|}{} &
   &
  \multicolumn{1}{l|}{641.173568} &
  \multicolumn{1}{l|}{25.746834} &
  % \multicolumn{1}{l|}{85.56299781799316} &
   \multicolumn{1}{l|}{1.426} &
  666.704240 &
  \multicolumn{1}{l|}{102.522820} &
  99.21 \\ \cline{2-8} 
% \multicolumn{1}{|l|}{} &
%   \multirow{3}{*}{U-shaped} &
%   \multicolumn{1}{l|}{641.173850} &
%   \multicolumn{1}{l|}{51.493164} &
%   % \multicolumn{1}{l|}{86.1868371963501} &
%   \multicolumn{1}{l|}{1.4364} &
%   692.450840 &
%   \multicolumn{1}{l|}{106.798264} &
%   99.34 \\ \cline{3-8} 
% \multicolumn{1}{|l|}{} &
%    &
%   \multicolumn{1}{l|}{641.170546} &
%   \multicolumn{1}{l|}{51.491264} &
%   % \multicolumn{1}{l|}{85.12164688110352} &
%   \multicolumn{1}{l|}{1.4187} &
%   692.445636 &
%   \multicolumn{1}{l|}{106.798474} &
%   99.33 \\ \cline{3-8} 
% \multicolumn{1}{|l|}{} &
%    &
%   \multicolumn{1}{l|}{641.172058} &
%   \multicolumn{1}{l|}{51.491868} &
%   % \multicolumn{1}{l|}{86.41077899932861} &
%   \multicolumn{1}{l|}{1.4402} &
%   692.447752 &
%   \multicolumn{1}{l|}{106.797552} &
%   99.23 \\ 
\hline
\multicolumn{1}{|l|}{\multirow{6}{*}{Private}} &
  \multirow{3}{*}{Local} &
  \multicolumn{1}{l|}{3863.333306}
  &
  \multicolumn{1}{l|}{3.571760} &
  % \multicolumn{1}{l|}{77321.75610351562} &
  \multicolumn{1}{l|}{1288.6959} &
  3866.905066 &
  \multicolumn{1}{l|}{1253.010860} &
  97.04 \\ \cline{3-8} 
\multicolumn{1}{|l|}{} &
   &
  \multicolumn{1}{l|}{3860.620376} &
  \multicolumn{1}{l|}{3.574165} &
  % \multicolumn{1}{l|}{76299.31762695312} &
  \multicolumn{1}{l|}{1271.6553} &
  3864.194541 &
  \multicolumn{1}{l|}{1253.010740} &
  97.29 \\ \cline{3-8} 
\multicolumn{1}{|l|}{} &
   &
  \multicolumn{1}{l|}{3861.326498} &
  \multicolumn{1}{l|}{3.576877} &
  % \multicolumn{1}{l|}{71435.74780273438} &
  \multicolumn{1}{l|}{1190.5958} &
  3864.903375 &
  \multicolumn{1}{l|}{1253.010820} &
  96.71 \\ \cline{2-8}
\multicolumn{1}{|l|}{} &
  \multirow{3}{*}{Vanilla} &
  \multicolumn{1}{l|}{1332.080326} &
  \multicolumn{1}{l|}{3.573196} &
  % \multicolumn{1}{l|}{11498.9580078125} &
  \multicolumn{1}{l|}{191.6493} &
  1335.653522 &
  \multicolumn{1}{l|}{204.857314} &
  97.26 \\ \cline{3-8} 
\multicolumn{1}{|l|}{} &
   &
  \multicolumn{1}{l|}{1332.082984} &
  \multicolumn{1}{l|}{3.572936} &
  % \multicolumn{1}{l|}{13227.383422851562} &
  \multicolumn{1}{l|}{220.4564} &
  1335.655920 &
  \multicolumn{1}{l|}{204.857348} &
  97.14 \\ \cline{3-8} 
\multicolumn{1}{|l|}{} &
   &
  \multicolumn{1}{l|}{1332.086146} &
  \multicolumn{1}{l|}{3.574748} &
  % \multicolumn{1}{l|}{9689.528015136719} &
\multicolumn{1}{l|}{161.4921} &
  1335.660894 &
  \multicolumn{1}{l|}{204.857618} &
  96.89 \\ \cline{2-8} 
% \multicolumn{1}{|l|}{} &
%   \multirow{3}{*}{U-shaped} &
%   \multicolumn{1}{l|}{1334.177854} &
%   \multicolumn{1}{l|}{99.418080} &
%   % \multicolumn{1}{l|}{9799.123413085938} &
%   \multicolumn{1}{l|}{163.3187} &
%   1433.595934 &
%   \multicolumn{1}{l|}{205.128190} &
%   97.38 \\ \cline{3-8} 
% \multicolumn{1}{|l|}{} &
%    &
%   \multicolumn{1}{l|}{1334.178550} &
%   \multicolumn{1}{l|}{99.418672} &
%   % \multicolumn{1}{l|}{10571.716735839844} &
% \multicolumn{1}{l|}{176.1953} &
%   1433.597222 &
%   \multicolumn{1}{l|}{205.128546} &
%   97.32 \\ \cline{3-8} 
% \multicolumn{1}{|l|}{} &
%    &
%   \multicolumn{1}{l|}{1334.175216} &
%   \multicolumn{1}{l|}{99.417462} &
%   % \multicolumn{1}{l|}{10276.938110351562} &
%   \multicolumn{1}{l|}{171.2823} &
%   1433.592678 &
%   \multicolumn{1}{l|}{205.128334} &
% 96.93 \\ 
\hline
\end{tabular}
\end{table*}

%\subsubsection{Private vanilla SL}
\noindent \textit{Private vanilla SL.} \enskip
\label{subsubsec:FSSsplit_eval}
The private vanilla SL model addresses privacy concerns in public models while maintaining nearly the same accuracy, experiencing only a 2\% reduction compared to both public local and public vanilla models (see~\autoref{table:performance}). This is because, as previously discussed, both SL and FSS can be implemented without a notable reduction in model accuracy. It also achieves comparable accuracy to the private local model in around~192 minutes, a significant advantage. The shorter training time is due to executing all client-side layers in plain, unlike the private local model, which employs various techniques such as beaver triples (for convolution layers) and FSS (for ReLU). Additionally, the total communication cost of private vanilla SL is~2.9$\times$ lower than that of the private local model. Notably, the server-side communication cost for private vanilla SL is remarkably low (similar to the public vanilla SL model), compared to the private local model. It is important to note that private vanilla SL additionally requires~0.119634 MB of preprocessing communications per batch, which is a~88.3$\times$ reduction when compared to private local FSS (10.572851 MB per batch). This is because we execute client-side layers on plaintext, whereas private local model involves more interaction between parties, increasing communication costs. In summary, private vanilla SL model outperforms the private local model in terms of training time and communication cost while maintaining similar accuracy.  Moreover, it offers better privacy than either public models, though with a slight accuracy trade-off. It is important to mention that our private vanilla SL model has higher training time and communication costs compared to the public models.

\noindent \textbf{Open Science and Reproducible Research:} 
%\label{subsec:openscience}
To support open science and reproducible research, and provide researchers with opportunity to use, test, and extend our work, source code used for the evaluations is publicly available\footnote{\url{https://github.com/UnoriginalOrigi/SplitFSS}}.

\section{Conclusion}
\label{sec:conclusion}

%While split learning was considered a promising collaborative learning technique for preserving the privacy of client data, recent studies have raised concern about its effectiveness due to potential privacy leakage of client data. To address this, we proposed a hybrid approach using split learning and another privacy-preserving technique called function secret sharing. Our work shows how function secret sharing can reduce privacy leakage in split learning while maintaining efficiency. We have shown that our proposed function secret sharing based vanilla %and U-shaped 
%split learning approach is more efficient in terms of communication and is less complex, all while achieving the same level of accuracy compared to existing protocols (private local model) in the field.
Although split learning was once viewed as a promising collaborative learning method for safeguarding the privacy of client data, recent studies have highlighted concerns about its effectiveness in preventing potential privacy leaks from client data. In response to these concerns, we have introduced a hybrid approach that combines split learning with another privacy-preserving technique known as function secret sharing. Our research demonstrates how the utilization of function secret sharing can effectively mitigate privacy breaches in split learning while still maintaining efficiency.

Our findings indicate that our proposed approach, which combines function secret sharing with vanilla split learning, outperforms existing protocols (such as private local models) in terms of communication efficiency and complexity. Notably, it achieves the same level of accuracy while offering improved privacy protections.

\begin{acks}
This work was funded by the HARPOCRATES EU research project (No. 101069535) and the Technology Innovation Institute (TII), UAE, for the project ARROWSMITH.
\end{acks}

\bibliographystyle{ACM-Reference-Format}
\balance
\bibliography{main}


\begin{thebibliography}{39}

%%% the Web address.

\ifx \showCODEN    \undefined \def \showCODEN     #1{\unskip}     \fi
\ifx \showDOI      \undefined \def \showDOI       #1{#1}\fi
\ifx \showISBNx    \undefined \def \showISBNx     #1{\unskip}     \fi
\ifx \showISBNxiii \undefined \def \showISBNxiii  #1{\unskip}     \fi
\ifx \showISSN     \undefined \def \showISSN      #1{\unskip}     \fi
\ifx \showLCCN     \undefined \def \showLCCN      #1{\unskip}     \fi
\ifx \shownote     \undefined \def \shownote      #1{#1}          \fi
\ifx \showarticletitle \undefined \def \showarticletitle #1{#1}   \fi
\ifx \showURL      \undefined \def \showURL       {\relax}        \fi
% invisible to TeX
\providecommand\bibfield[2]{#2}
\providecommand\bibinfo[2]{#2}
\providecommand\natexlab[1]{#1}
\providecommand\showeprint[2][]{arXiv:#2}

\bibitem[Abuadbba et~al\mbox{.}(2020)]%
        {abuadbba2020can}
\bibfield{author}{\bibinfo{person}{Sharif Abuadbba}, \bibinfo{person}{Kyuyeon Kim}, \bibinfo{person}{Minki Kim}, \bibinfo{person}{Chandra Thapa}, \bibinfo{person}{Seyit~A Camtepe}, \bibinfo{person}{Yansong Gao}, \bibinfo{person}{Hyoungshick Kim}, {and} \bibinfo{person}{Surya Nepal}.} \bibinfo{year}{2020}\natexlab{}.
\newblock \showarticletitle{Can we use split learning on 1d cnn models for privacy preserving training?}. In \bibinfo{booktitle}{\emph{Proceedings of the 15th ACM Asia Conference on Computer and Communications Security}}. \bibinfo{pages}{305--318}.
\newblock


\bibitem[Agarwal et~al\mbox{.}(2022)]%
        {agarwal2022communication}
\bibfield{author}{\bibinfo{person}{Amit Agarwal}, \bibinfo{person}{Stanislav Peceny}, \bibinfo{person}{Mariana Raykova}, \bibinfo{person}{Phillipp Schoppmann}, {and} \bibinfo{person}{Karn Seth}.} \bibinfo{year}{2022}\natexlab{}.
\newblock \showarticletitle{Communication Efficient Secure Logistic Regression}.
\newblock \bibinfo{journal}{\emph{Cryptology ePrint Archive}} (\bibinfo{year}{2022}).
\newblock


\bibitem[Agrawal et~al\mbox{.}(2019)]%
        {agrawal2019quotient}
\bibfield{author}{\bibinfo{person}{Nitin Agrawal}, \bibinfo{person}{Ali Shahin~Shamsabadi}, \bibinfo{person}{Matt~J Kusner}, {and} \bibinfo{person}{Adri{\`a} Gasc{\'o}n}.} \bibinfo{year}{2019}\natexlab{}.
\newblock \showarticletitle{QUOTIENT: two-party secure neural network training and prediction}. In \bibinfo{booktitle}{\emph{Proceedings of the 2019 ACM SIGSAC Conference on Computer and Communications Security}}. \bibinfo{pages}{1231--1247}.
\newblock


\bibitem[Boyle et~al\mbox{.}(2021)]%
        {boyle2021function}
\bibfield{author}{\bibinfo{person}{Elette Boyle}, \bibinfo{person}{Nishanth Chandran}, \bibinfo{person}{Niv Gilboa}, \bibinfo{person}{Divya Gupta}, \bibinfo{person}{Yuval Ishai}, \bibinfo{person}{Nishant Kumar}, {and} \bibinfo{person}{Mayank Rathee}.} \bibinfo{year}{2021}\natexlab{}.
\newblock \showarticletitle{Function secret sharing for mixed-mode and fixed-point secure computation}. In \bibinfo{booktitle}{\emph{Advances in Cryptology--EUROCRYPT 2021: 40th Annual International Conference on the Theory and Applications of Cryptographic Techniques, Zagreb, Croatia, October 17--21, 2021, Proceedings, Part II}}. Springer, \bibinfo{pages}{871--900}.
\newblock


\bibitem[Boyle et~al\mbox{.}(2015)]%
        {boyle2015function}
\bibfield{author}{\bibinfo{person}{Elette Boyle}, \bibinfo{person}{Niv Gilboa}, {and} \bibinfo{person}{Yuval Ishai}.} \bibinfo{year}{2015}\natexlab{}.
\newblock \showarticletitle{Function secret sharing}. In \bibinfo{booktitle}{\emph{Annual international conference on the theory and applications of cryptographic techniques}}. Springer, \bibinfo{pages}{337--367}.
\newblock


\bibitem[Boyle et~al\mbox{.}(2016)]%
        {boyle2016function}
\bibfield{author}{\bibinfo{person}{Elette Boyle}, \bibinfo{person}{Niv Gilboa}, {and} \bibinfo{person}{Yuval Ishai}.} \bibinfo{year}{2016}\natexlab{}.
\newblock \showarticletitle{Function secret sharing: Improvements and extensions}. In \bibinfo{booktitle}{\emph{Proceedings of the 2016 ACM SIGSAC Conference on Computer and Communications Security}}. \bibinfo{pages}{1292--1303}.
\newblock


\bibitem[Boyle et~al\mbox{.}(2019)]%
        {boyle2019secure}
\bibfield{author}{\bibinfo{person}{Elette Boyle}, \bibinfo{person}{Niv Gilboa}, {and} \bibinfo{person}{Yuval Ishai}.} \bibinfo{year}{2019}\natexlab{}.
\newblock \showarticletitle{Secure computation with preprocessing via function secret sharing}. In \bibinfo{booktitle}{\emph{Theory of Cryptography: 17th International Conference, TCC 2019, Nuremberg, Germany, December 1--5, 2019, Proceedings, Part I 17}}. Springer, \bibinfo{pages}{341--371}.
\newblock


\bibitem[Gentry(2009)]%
        {gentry2009fully}
\bibfield{author}{\bibinfo{person}{Craig Gentry}.} \bibinfo{year}{2009}\natexlab{}.
\newblock \bibinfo{booktitle}{\emph{A fully homomorphic encryption scheme}}.
\newblock \bibinfo{publisher}{Stanford university}.
\newblock


\bibitem[Gupta et~al\mbox{.}(2023)]%
        {gupta2023sigma}
\bibfield{author}{\bibinfo{person}{Kanav Gupta}, \bibinfo{person}{Neha Jawalkar}, \bibinfo{person}{Ananta Mukherjee}, \bibinfo{person}{Nishanth Chandran}, \bibinfo{person}{Divya Gupta}, \bibinfo{person}{Ashish Panwar}, {and} \bibinfo{person}{Rahul Sharma}.} \bibinfo{year}{2023}\natexlab{}.
\newblock \showarticletitle{Sigma: Secure gpt inference with function secret sharing}.
\newblock \bibinfo{journal}{\emph{Cryptology ePrint Archive}} (\bibinfo{year}{2023}).
\newblock


\bibitem[Gupta et~al\mbox{.}(2022)]%
        {gupta2022llama}
\bibfield{author}{\bibinfo{person}{Kanav Gupta}, \bibinfo{person}{Deepak Kumaraswamy}, \bibinfo{person}{Nishanth Chandran}, {and} \bibinfo{person}{Divya Gupta}.} \bibinfo{year}{2022}\natexlab{}.
\newblock \showarticletitle{LLAMA: A Low Latency Math Library for Secure Inference}.
\newblock \bibinfo{journal}{\emph{Proceedings on Privacy Enhancing Technologies}}  \bibinfo{volume}{4} (\bibinfo{year}{2022}), \bibinfo{pages}{274--294}.
\newblock


\bibitem[Gupta and Raskar(2018)]%
        {gupta2018distributed}
\bibfield{author}{\bibinfo{person}{Otkrist Gupta} {and} \bibinfo{person}{Ramesh Raskar}.} \bibinfo{year}{2018}\natexlab{}.
\newblock \showarticletitle{Distributed learning of deep neural network over multiple agents}.
\newblock \bibinfo{journal}{\emph{Journal of Network and Computer Applications}}  \bibinfo{volume}{116} (\bibinfo{year}{2018}), \bibinfo{pages}{1--8}.
\newblock


\bibitem[Hesamifard et~al\mbox{.}(2016)]%
        {hesamifard2016cryptodl}
\bibfield{author}{\bibinfo{person}{Ehsan Hesamifard}, \bibinfo{person}{Hassan Takabi}, {and} \bibinfo{person}{Mehdi Ghasemi}.} \bibinfo{year}{2016}\natexlab{}.
\newblock \showarticletitle{Cryptodl: towards deep learning over encrypted data}. In \bibinfo{booktitle}{\emph{Annual Computer Security Applications Conference (ACSAC 2016), Los Angeles, California, USA}}, Vol.~\bibinfo{volume}{11}.
\newblock


\bibitem[Hitaj et~al\mbox{.}(2017)]%
        {hitaj2017deep}
\bibfield{author}{\bibinfo{person}{Briland Hitaj}, \bibinfo{person}{Giuseppe Ateniese}, {and} \bibinfo{person}{Fernando Perez-Cruz}.} \bibinfo{year}{2017}\natexlab{}.
\newblock \showarticletitle{Deep models under the GAN: information leakage from collaborative deep learning}. In \bibinfo{booktitle}{\emph{Proceedings of the 2017 ACM SIGSAC conference on computer and communications security}}. \bibinfo{pages}{603--618}.
\newblock


\bibitem[Hu et~al\mbox{.}(2022)]%
        {hu2022membership}
\bibfield{author}{\bibinfo{person}{Hongsheng Hu}, \bibinfo{person}{Zoran Salcic}, \bibinfo{person}{Lichao Sun}, \bibinfo{person}{Gillian Dobbie}, \bibinfo{person}{Philip~S Yu}, {and} \bibinfo{person}{Xuyun Zhang}.} \bibinfo{year}{2022}\natexlab{}.
\newblock \showarticletitle{Membership inference attacks on machine learning: A survey}.
\newblock \bibinfo{journal}{\emph{ACM Computing Surveys (CSUR)}} \bibinfo{volume}{54}, \bibinfo{number}{11s} (\bibinfo{year}{2022}), \bibinfo{pages}{1--37}.
\newblock


\bibitem[Jawalkar et~al\mbox{.}(2023)]%
        {jawalkar2023orca}
\bibfield{author}{\bibinfo{person}{Neha Jawalkar}, \bibinfo{person}{Kanav Gupta}, \bibinfo{person}{Arkaprava Basu}, \bibinfo{person}{Nishanth Chandran}, \bibinfo{person}{Divya Gupta}, {and} \bibinfo{person}{Rahul Sharma}.} \bibinfo{year}{2023}\natexlab{}.
\newblock \showarticletitle{Orca: FSS-based Secure Training and Inference with GPUs}. In \bibinfo{booktitle}{\emph{2024 IEEE Symposium on Security and Privacy (SP)}}. IEEE Computer Society, \bibinfo{pages}{63--63}.
\newblock


\bibitem[Khan et~al\mbox{.}(2021)]%
        {khan2021blind}
\bibfield{author}{\bibinfo{person}{Tanveer Khan}, \bibinfo{person}{Alexandros Bakas}, {and} \bibinfo{person}{Antonis Michalas}.} \bibinfo{year}{2021}\natexlab{}.
\newblock \showarticletitle{Blind faith: Privacy-preserving machine learning using function approximation}. In \bibinfo{booktitle}{\emph{2021 IEEE Symposium on Computers and Communications (ISCC)}}. IEEE, \bibinfo{pages}{1--7}.
\newblock


\bibitem[Khan et~al\mbox{.}(2024)]%
        {khan2024wildest}
\bibfield{author}{\bibinfo{person}{Tanveer Khan}, \bibinfo{person}{Mindaugas Budzys}, \bibinfo{person}{Khoa Nguyen}, {and} \bibinfo{person}{Antonis Michalas}.} \bibinfo{year}{2024}\natexlab{}.
\newblock \showarticletitle{Wildest Dreams: Reproducible Research in Privacy-preserving Neural Network Training}.
\newblock \bibinfo{journal}{\emph{Proceedings on Privacy Enhancing Technologies}} \bibinfo{volume}{2024}, \bibinfo{number}{3}.
\newblock


\bibitem[Khan and Michalas(2023)]%
        {khan2023learning}
\bibfield{author}{\bibinfo{person}{Tanveer Khan} {and} \bibinfo{person}{Antonis Michalas}.} \bibinfo{year}{2023}\natexlab{}.
\newblock \showarticletitle{Learning in the Dark: Privacy-Preserving Machine Learning using Function Approximation}.
\newblock  (\bibinfo{year}{2023}).
\newblock


\bibitem[Khan et~al\mbox{.}(2023a)]%
        {khan2023more}
\bibfield{author}{\bibinfo{person}{Tanveer Khan}, \bibinfo{person}{Khoa Nguyen}, {and} \bibinfo{person}{Antonis Michalas}.} \bibinfo{year}{2023}\natexlab{a}.
\newblock \showarticletitle{A More Secure Split: Enhancing the Security of Privacy-Preserving Split Learning}. In \bibinfo{booktitle}{\emph{Nordic Conference on Secure IT Systems}}. Springer, \bibinfo{pages}{307--329}.
\newblock


\bibitem[Khan et~al\mbox{.}(2023b)]%
        {khan2023split}
\bibfield{author}{\bibinfo{person}{Tanveer Khan}, \bibinfo{person}{Khoa Nguyen}, {and} \bibinfo{person}{Antonis Michalas}.} \bibinfo{year}{2023}\natexlab{b}.
\newblock \showarticletitle{Split Ways: Privacy-Preserving Training of Encrypted Data Using Split Learning}. In \bibinfo{booktitle}{\emph{2023 Workshops of the EDBT/ICDT Joint Conference, EDBT/ICDT-WS 2023, 28 March 2023}}. CEUR-WS.
\newblock


\bibitem[Khan et~al\mbox{.}(2023c)]%
        {khan2023love}
\bibfield{author}{\bibinfo{person}{Tanveer Khan}, \bibinfo{person}{Khoa Nguyen}, \bibinfo{person}{Antonis Michalas}, {and} \bibinfo{person}{Alexandros Bakas}.} \bibinfo{year}{2023}\natexlab{c}.
\newblock \showarticletitle{Love or Hate? Share or Split? Privacy-Preserving Training Using Split Learning and Homomorphic Encryption}. In \bibinfo{booktitle}{\emph{2023 20th Annual International Conference on Privacy, Security and Trust (PST)}}. IEEE Computer Society, \bibinfo{pages}{1--7}.
\newblock


\bibitem[Kone{\v{c}}n{\`y} et~al\mbox{.}(2016)]%
        {konevcny2016federated}
\bibfield{author}{\bibinfo{person}{Jakub Kone{\v{c}}n{\`y}}, \bibinfo{person}{H~Brendan McMahan}, \bibinfo{person}{Daniel Ramage}, {and} \bibinfo{person}{Peter Richt{\'a}rik}.} \bibinfo{year}{2016}\natexlab{}.
\newblock \showarticletitle{Federated optimization: Distributed machine learning for on-device intelligence}.
\newblock \bibinfo{journal}{\emph{arXiv preprint arXiv:1610.02527}} (\bibinfo{year}{2016}).
\newblock


\bibitem[Lecun et~al\mbox{.}(1998)]%
        {lecun1998CNN}
\bibfield{author}{\bibinfo{person}{Y. Lecun}, \bibinfo{person}{L. Bottou}, \bibinfo{person}{Y. Bengio}, {and} \bibinfo{person}{P. Haffner}.} \bibinfo{year}{1998}\natexlab{}.
\newblock \showarticletitle{Gradient-based learning applied to document recognition}.
\newblock \bibinfo{journal}{\emph{Proc. IEEE}} \bibinfo{volume}{86}, \bibinfo{number}{11} (\bibinfo{year}{1998}), \bibinfo{pages}{2278--2324}.
\newblock
\urldef\tempurl%
\url{https://doi.org/10.1109/5.726791}
\showDOI{\tempurl}


\bibitem[LeCun et~al\mbox{.}({[n.\,d.]})]%
        {lecun_cortes_j.c.burges}
\bibfield{author}{\bibinfo{person}{Yann LeCun}, \bibinfo{person}{Corinna Cortes}, {and} \bibinfo{person}{Christopher J.C.~Burges}.} \bibinfo{year}{[n.\,d.]}\natexlab{}.
\newblock \bibinfo{title}{The mnist database}.
\newblock
\newblock
\urldef\tempurl%
\url{http://yann.lecun.com/exdb/mnist/}
\showURL{%
\tempurl}


\bibitem[Lecuyer et~al\mbox{.}(2019)]%
        {lecuyer2019certified}
\bibfield{author}{\bibinfo{person}{Mathias Lecuyer}, \bibinfo{person}{Vaggelis Atlidakis}, \bibinfo{person}{Roxana Geambasu}, \bibinfo{person}{Daniel Hsu}, {and} \bibinfo{person}{Suman Jana}.} \bibinfo{year}{2019}\natexlab{}.
\newblock \showarticletitle{Certified robustness to adversarial examples with differential privacy}. In \bibinfo{booktitle}{\emph{2019 IEEE symposium on security and privacy (SP)}}. IEEE, \bibinfo{pages}{656--672}.
\newblock


\bibitem[Li et~al\mbox{.}(2022)]%
        {li2022ressfl}
\bibfield{author}{\bibinfo{person}{Jingtao Li}, \bibinfo{person}{Adnan~Siraj Rakin}, \bibinfo{person}{Xing Chen}, \bibinfo{person}{Zhezhi He}, \bibinfo{person}{Deliang Fan}, {and} \bibinfo{person}{Chaitali Chakrabarti}.} \bibinfo{year}{2022}\natexlab{}.
\newblock \showarticletitle{Ressfl: A resistance transfer framework for defending model inversion attack in split federated learning}. In \bibinfo{booktitle}{\emph{Proceedings of the IEEE/CVF Conference on Computer Vision and Pattern Recognition}}. \bibinfo{pages}{10194--10202}.
\newblock


\bibitem[McMahan et~al\mbox{.}(2016)]%
        {mcmahan2016federated}
\bibfield{author}{\bibinfo{person}{H~Brendan McMahan}, \bibinfo{person}{Eider Moore}, \bibinfo{person}{Daniel Ramage}, {and} \bibinfo{person}{Blaise~Ag{\"u}era y Arcas}.} \bibinfo{year}{2016}\natexlab{}.
\newblock \showarticletitle{Federated learning of deep networks using model averaging}.
\newblock \bibinfo{journal}{\emph{arXiv preprint arXiv:1602.05629}}  \bibinfo{volume}{2} (\bibinfo{year}{2016}), \bibinfo{pages}{2}.
\newblock


\bibitem[Mohassel and Zhang(2017)]%
        {mohassel2017secureml}
\bibfield{author}{\bibinfo{person}{Payman Mohassel} {and} \bibinfo{person}{Yupeng Zhang}.} \bibinfo{year}{2017}\natexlab{}.
\newblock \showarticletitle{Secureml: A system for scalable privacy-preserving machine learning}. In \bibinfo{booktitle}{\emph{2017 IEEE symposium on security and privacy (SP)}}. IEEE, \bibinfo{pages}{19--38}.
\newblock


\bibitem[Nguyen et~al\mbox{.}(2023)]%
        {nguyen2023split}
\bibfield{author}{\bibinfo{person}{K Nguyen}, \bibinfo{person}{T Khan}, {and} \bibinfo{person}{A Michalas}.} \bibinfo{year}{2023}\natexlab{}.
\newblock \showarticletitle{Split Without a Leak: Reducing Privacy Leakage in Split Learning}. In \bibinfo{booktitle}{\emph{19th EAI International Conference on Security and Privacy in Communication Networks (SecureComm’23)}}. Springer.
\newblock


\bibitem[OpenAI(2021)]%
        {openai_chatgpt}
\bibfield{author}{\bibinfo{person}{OpenAI}.} \bibinfo{year}{2021}\natexlab{}.
\newblock \bibinfo{title}{ChatGPT}.
\newblock \bibinfo{howpublished}{Computer software}.
\newblock
\urldef\tempurl%
\url{https://openai.com}
\showURL{%
\tempurl}


\bibitem[Pasquini et~al\mbox{.}(2021)]%
        {pasquini2021unleashing}
\bibfield{author}{\bibinfo{person}{Dario Pasquini}, \bibinfo{person}{Giuseppe Ateniese}, {and} \bibinfo{person}{Massimo Bernaschi}.} \bibinfo{year}{2021}\natexlab{}.
\newblock \showarticletitle{Unleashing the tiger: Inference attacks on split learning}. In \bibinfo{booktitle}{\emph{Proceedings of the 2021 ACM SIGSAC Conference on Computer and Communications Security}}. \bibinfo{pages}{2113--2129}.
\newblock


\bibitem[Ryffel et~al\mbox{.}(2020)]%
        {ryffel2020ariann}
\bibfield{author}{\bibinfo{person}{Th{\'e}o Ryffel}, \bibinfo{person}{Pierre Tholoniat}, \bibinfo{person}{David Pointcheval}, {and} \bibinfo{person}{Francis Bach}.} \bibinfo{year}{2020}\natexlab{}.
\newblock \showarticletitle{Ariann: Low-interaction privacy-preserving deep learning via function secret sharing}.
\newblock \bibinfo{journal}{\emph{Proceedings on Privacy Enhancing Technologies}} \bibinfo{volume}{2022}, \bibinfo{number}{1} (\bibinfo{year}{2020}), \bibinfo{pages}{291--316}.
\newblock


\bibitem[Sav et~al\mbox{.}(2021)]%
        {sav2021poseidon}
\bibfield{author}{\bibinfo{person}{Sinem Sav}, \bibinfo{person}{Apostolos Pyrgelis}, \bibinfo{person}{Juan~Ramon Troncoso-Pastoriza}, \bibinfo{person}{David Froelicher}, \bibinfo{person}{Jean-Philippe Bossuat}, \bibinfo{person}{Joao~Sa Sousa}, {and} \bibinfo{person}{Jean-Pierre Hubaux}.} \bibinfo{year}{2021}\natexlab{}.
\newblock \showarticletitle{POSEIDON: Privacy-Preserving Federated Neural Network Learning}. In \bibinfo{booktitle}{\emph{28th Annual Network and Distributed System Security Symposium, {NDSS} 2021, virtually, February 21-25, 2021}}. INTERNET SOC.
\newblock


\bibitem[Shokri et~al\mbox{.}(2017)]%
        {shokri2017membership}
\bibfield{author}{\bibinfo{person}{Reza Shokri}, \bibinfo{person}{Marco Stronati}, \bibinfo{person}{Congzheng Song}, {and} \bibinfo{person}{Vitaly Shmatikov}.} \bibinfo{year}{2017}\natexlab{}.
\newblock \showarticletitle{Membership inference attacks against machine learning models}. In \bibinfo{booktitle}{\emph{2017 IEEE symposium on security and privacy (SP)}}. IEEE, \bibinfo{pages}{3--18}.
\newblock


\bibitem[Thapa et~al\mbox{.}(2022)]%
        {thapa2022splitfed}
\bibfield{author}{\bibinfo{person}{Chandra Thapa}, \bibinfo{person}{Pathum Chamikara~Mahawaga Arachchige}, \bibinfo{person}{Seyit Camtepe}, {and} \bibinfo{person}{Lichao Sun}.} \bibinfo{year}{2022}\natexlab{}.
\newblock \showarticletitle{Splitfed: When federated learning meets split learning}. In \bibinfo{booktitle}{\emph{Proceedings of the AAAI Conference on Artificial Intelligence}}, Vol.~\bibinfo{volume}{36}. \bibinfo{pages}{8485--8493}.
\newblock


\bibitem[Truex et~al\mbox{.}(2019)]%
        {truex2019hybrid}
\bibfield{author}{\bibinfo{person}{Stacey Truex}, \bibinfo{person}{Nathalie Baracaldo}, \bibinfo{person}{Ali Anwar}, \bibinfo{person}{Thomas Steinke}, \bibinfo{person}{Heiko Ludwig}, \bibinfo{person}{Rui Zhang}, {and} \bibinfo{person}{Yi Zhou}.} \bibinfo{year}{2019}\natexlab{}.
\newblock \showarticletitle{A hybrid approach to privacy-preserving federated learning}. In \bibinfo{booktitle}{\emph{Proceedings of the 12th ACM workshop on artificial intelligence and security}}. \bibinfo{pages}{1--11}.
\newblock


\bibitem[Wagh(2022)]%
        {wagh2022pika}
\bibfield{author}{\bibinfo{person}{Sameer Wagh}.} \bibinfo{year}{2022}\natexlab{}.
\newblock \showarticletitle{Pika: Secure Computation using Function Secret Sharing over Rings}.
\newblock \bibinfo{journal}{\emph{Proceedings on Privacy Enhancing Technologies}}  \bibinfo{volume}{4} (\bibinfo{year}{2022}), \bibinfo{pages}{351--377}.
\newblock


\bibitem[Yang et~al\mbox{.}(2023)]%
        {yang2023fssnn}
\bibfield{author}{\bibinfo{person}{Peng Yang}, \bibinfo{person}{Zoe~L Jiang}, \bibinfo{person}{Shiqi Gao}, \bibinfo{person}{Jiehang Zhuang}, \bibinfo{person}{Hongxiao Wang}, \bibinfo{person}{Junbin Fang}, \bibinfo{person}{Siuming Yiu}, {and} \bibinfo{person}{Yulin Wu}.} \bibinfo{year}{2023}\natexlab{}.
\newblock \showarticletitle{FssNN: Communication-Efficient Secure Neural Network Training via Function Secret Sharing}.
\newblock \bibinfo{journal}{\emph{Cryptology ePrint Archive}} (\bibinfo{year}{2023}).
\newblock


\bibitem[Yao(1986)]%
        {yao1986generate}
\bibfield{author}{\bibinfo{person}{Andrew Chi-Chih Yao}.} \bibinfo{year}{1986}\natexlab{}.
\newblock \showarticletitle{How to generate and exchange secrets}. In \bibinfo{booktitle}{\emph{27th Annual Symposium on Foundations of Computer Science (sfcs 1986)}}. IEEE.
\newblock


\end{thebibliography}

\end{document}